%% file: gracft3.tex
\documentclass[12pt]{article}
\usepackage[utf8]{inputenc}
\usepackage{amsmath}
\usepackage{amsfonts,mathrsfs}
\usepackage{amssymb,amsthm,xcolor}
\usepackage[pdftex,plainpages=false,hypertexnames=false,pdfpagelabels]{hyperref}
\usepackage{xcolor}
\usepackage{enumerate}

\usepackage{tikz}
\usetikzlibrary{calc}
\usetikzlibrary{patterns}
\usepackage{tocvsec2}
\usepackage{caption}
\definecolor{dark-red}{rgb}{0.7,0.25,0.25}
\definecolor{dark-blue}{rgb}{0.15,0.15,0.55}
\definecolor{medium-blue}{rgb}{0,0,.8}
\definecolor{DarkGreen}{RGB}{0,150,0}
\definecolor{rho}{named}{red}
\hypersetup{
   colorlinks, linkcolor={purple},
   citecolor={medium-blue}, urlcolor={medium-blue}
}
\usepackage[top=1in, bottom=1in, left=0.87in, right=0.87in]{geometry}

\usepackage[shortlabels]{enumitem}

\newcommand{\frg}{\mathfrak{g}}
\newcommand{\fso}{\mathfrak{so}}

\newcommand{\Diffc}{\Diff_c(S^1)}
\newcommand{\Ran}{\operatorname{Ran}}
\newcommand{\Rep}{\operatorname{Rep}}
\newcommand{\Int}{\operatorname{int}}
\newcommand{\Irr}{\operatorname{Irr}}

\newcommand{\Ann}{\operatorname{Ann}}

\newcommand{\scA}{\mathscr{A}}

\newcommand{\Rot}{\operatorname{Rot}}

\newcommand{\cB}{\mathcal{B}}

\newcommand{\cF}{\mathcal{F}}
\newcommand{\cL}{\mathcal{L}}

\newcommand{\cY}{\mathcal{Y}}
\newcommand{\cX}{\mathcal{X}}
\newcommand{\cH}{\mathcal{H}}

\newcommand{\C}{\mathbb{C}}
\newcommand{\bbC}{\mathbb{C}}
\newcommand{\cA}{\mathcal{A}}

\newcommand{\cI}{\mathcal{I}}
\newcommand{\cU}{\mathcal{U}}

\newcommand{\cS}{\mathcal{S}}

\newcommand{\bbR}{\mathbb{R}}
\newcommand{\Z}{\mathbb{Z}}
\newcommand{\bbZ}{\mathbb{Z}}

\newcommand{\fsl}{\mathfrak{sl}}

\newcommand{\cl}{\operatorname{cl}}

\newcommand{\id}{\operatorname{id}}

\newcommand{\abs}[1]{\left|#1\right|}
\newcommand{\babs}[1]{\big|#1\big|}

\newcommand{\norm}[1]{\left\|#1\right\|}

\newcommand{\vertex}[3]{\binom{#1}{#2 \, #3}}
\newcommand{\ip}[1]{\left\langle#1\right\rangle}

\newcommand{\bip}[1]{\big\langle#1\big\rangle}

\newcommand{\bbD}{\mathbb{D}}

\newcommand{\cR}{\mathcal{R}}
\newcommand{\Diff}{\operatorname{Diff}}

\newcommand{\Span}{\operatorname{span}}
\newcommand{\Ad}{\operatorname{Ad}}
\newcommand{\End}{\operatorname{End}}
\newcommand{\Hom}{\operatorname{Hom}}

\newcommand{\Vir}{\operatorname{Vir}}
\newcommand{\Mob}{\operatorname{M\ddot{o}b}}

\newcommand{\interior}[1]{\mathring{#1}}

\newtheorem{thmalpha}{Theorem}
\newtheorem{coralpha}[thmalpha]{Corollary}

\newtheorem{Theorem}{Theorem}[section]
\newtheorem*{Theorem*}{Theorem}
\newtheorem{Lemma}[Theorem]{Lemma}
\newtheorem{Proposition}[Theorem]{Proposition}
\newtheorem{Corollary}[Theorem]{Corollary}
\newtheorem*{Corollary*}{Corollary}
\theoremstyle{definition}
\newtheorem{Remark}[Theorem]{Remark}
\newtheorem{Definition}[Theorem]{Definition}
\newtheorem{Conjecture}[Theorem]{Conjecture}
\newtheorem*{Definition*}{Definition}

\newtheorem*{Problem*}{Problem}

\numberwithin{equation}{section}
\numberwithin{figure}{section}

\usepackage[nottoc]{tocbibind}
\settocbibname{References}

\begin{document}

\title{Fusion and positivity in chiral conformal field theory}
\author{James E. Tener}
\date{}
%\date{\today}

\maketitle

\abstract{
In this article we show that the conformal nets corresponding to WZW models are rational, resolving a long-standing open problem.
Specifically, we show that the Jones-Wassermann subfactors associated with these models have finite index.
This result was first conjectured in the early 90s but had previously only been proven in special cases, beginning with Wassermann's landmark results in type A.
The proof relies on a new framework for the systematic comparison of tensor products (a.k.a. `fusion') of conformal net representations with the corresponding tensor product of vertex operator algebra modules.
This framework is based on the geometric technique of `bounded localized vertex operators,' which realizes algebras of observables via insertion operators localized in partially thin Riemann surfaces.
We obtain a general method for showing that Jones-Wassermann subfactors have finite index, and apply it to additional families of important examples beyond WZW models.
We also consider applications to a class of positivity phenomena for VOAs, and use this to outline a program for identifying unitary tensor product theories of VOAs and conformal nets even for badly-behaved models.
}

\newpage

\setcounter{tocdepth}{2}
\tableofcontents
\newpage

%%%%%%%%%%%%%%%%%%%%%%%%%%%%%%%%%%%%%%%%%%%%%%%%%%

\input{Sec_1_IntroductionGRACFT3.tex}

\input{Sec_2_PreliminariesGRACFT3.tex}

\input{Sec_3_Positivity.tex}

\input{Sec_4_FusionRules.tex}

\input{Sec_5_Examples.tex}

\newpage
\input{gracft3.bbl}

%\bibliographystyle{../../alphaabbr}
%\bibliography{../../ffbib} 

\end{document}

%% file: Sec_1_IntroductionGRACFT3.tex
\section{Introduction} 

The Haag-Kastler axioms describe quantum field theories in terms of nets of operators algebras acting on a Hilbert space \cite{HaKa64}.
\emph{Conformal nets} are an axiomatiztaion of two-dimensional chiral conformal field theories (CFTs) in the style of the Haag-Kastler axioms \cite{FrReSc92}. 
Conformal nets have been at the center of extensive investigation aimed at providing a rigorous foundation for CFT and a better understanding of the mathematical phenomena that arise.
One such phenomenon is the \emph{index theory of subfactors} (first introduced in a landmark article of Vaughan Jones \cite{Jones83}) which is closely linked to the notion of \emph{rationality} of conformal field theories.
Indeed, a conformal net is called rational if it has finitely many irreducible representations and the subfactors corresponding to these representations all have finite index.
Because of this second requirement, rationality of conformal nets is much harder to verify than what is typically thought of as `rationality' of conformal field theories.
However under this definition rational conformal nets have been proven to exhibit the behavior expected from a rational CFT, such as modularity of the representation category \cite{KaLoMu01}.

Among the most fundamental examples of rational CFTs are the Wess-Zumino-Witten (WZW) models corresponding to simple complex Lie algebras at positive integer level.
The importance of showing that the conformal nets corresponding to WZW models are rational (specifically, showing that the corresponding subfactors have finite index) was identified at the very beginning of the study of conformal nets \cite{GabbianiFrohlich93}, and the difficulty that this foundational problem has exhibited over nearly 30 years speaks to the challenges inherent in rigorous analysis of quantum field theories.
We give a complete resolution:

\begin{thmalpha}\label{thm: intro wzw rational}
The WZW conformal nets are rational for all finite-dimensional simple complex Lie algebras and all positive integer levels.
In particular, the Jones-Wassermann subfactors corresponding to irreducible representations of these nets have finite index.
\end{thmalpha}

Special cases of Theorem \ref{thm: intro wzw rational} have appeared in the literature.
The first result in this direction was Wassermann's proof that subfactors arising from WZW models of type $A$ have finite index \cite{Wa98}, and shortly thereafter Toledano Laredo proved an analogous result for type $D$ at odd levels \cite{TL97}.
These results, when combined with a classification of irreducible conformal net representations (implicit in \cite{Xu00b,KaLoMu01} for type $A$, and obtained for all types in \cite{HenriquesColimits}), establish Theorem \ref{thm: intro wzw rational} for WZW models of type $A$ (all levels) and type $D$ (odd levels).
Later, Gui proved Theorem \ref{thm: intro wzw rational} for WZW models of type $C$ and $G$ \cite{gui21categoricalextension}.
The proof we give covers all cases.

In order to prove Theorem \ref{thm: intro wzw rational}, we will compare conformal nets with other axiomatizations of chiral conformal field theory, specifically \emph{vertex operator algebras} \cite{Borcherds86,FLM88} and \emph{Segal CFTs} (similar to the better known Atiyah-Segal axioms of topological quantum field theory) \cite{SegalDef, MoSe89}.
These axioms are supposed to describe (essentially) the same objects as conformal nets, but look and behave very differently.
The project of comparing these notions has developed significant momentum in recent years, with the goal of obtaining a single unified framework for the mathematical study of CFTs.
One of the most important, and most difficult, areas in which to make this comparison is that of tensor products (`fusion products') of representations, as many important questions regarding fusion products have been solved in one of the axiomatizations but are fundamental open problems in another.
Computing fusion products in the context of conformal nets in turn allows one to show that the corresponding subfactors have finite index \cite[Thm. 4.1]{Longo90}.

In this article we will develop a framework for the systematic comparison of fusion product theories in VOAs and conformal nets, using the idea of Segal CFT (extended to `partially thin' Riemann surfaces as introduced by Henriques \cite[\S2]{Henriques14}) to provide a new geometric understanding of `fusion' in the setting of operator algebras.
This framework is used to give a short proof of Theorem \ref{thm: intro wzw rational}, but also to prove rationality of other conformal nets, such as the discrete series $W$-algebras of type $A$ and $E$.
Rationality of these $W$-algebra conformal nets had previously only been proven for type $A_1$ \cite{KaLo04} (in which case the discrete series $W$-algebras are the unitary Virasoro minimal models).
As another application, we provide new links at the categorical level and use conformal nets to demonstrate the unitarity of representation categories for all WZW VOAs (building upon \cite{GuiUnitarityI,GuiUnitarityII,GuiG2}) as well as many $W$-algebra examples.%
\footnote{In the time between this article's posting on the arXiv and its submission for publication, Gui in fact proved the representation categories of these WZW and $W$-algebra conformal nets coincide, which also provides an alternate proof of Theorem \ref{thm: intro wzw rational} \cite{Gui20unbddax}.}
%We are also able to recover existing results related to tensor product theory (e.g. many results from \cite{Wa98,TL97,Loke,GuiG2}) in a uniform and straightforward manner.
Our approach is sufficiently general that the methods used should apply to any model (without any additional assumption, such as rationality), and we describe how it could pave the way for a uniform understanding of the relationship between fusion product theories in unitary VOAs and conformal nets.

In the remainder of the introduction we will outline our framework for comparing VOA and conformal net fusion product theory and describe further applications.

\subsubsection*{Bounded localized vertex operators}

We now give a brief overview of the geometric method for relating VOAs and conformal nets developed in \cite{GRACFT1,TenerGRACFT2}.
For a more detailed overview, see Section \ref{sec: BLVO} or \cite[\S1.2]{TenerGRACFT2}.
Given a unitary VOA and numbers $R > 1 > r$, the operator $R^{-L_0}Y(v,z)r^{L_0}$ inserts the state $v$ at the point $z$ inside the annulus $\{R > \abs{w} > r\}$.
$$
\begin{tikzpicture}[scale=0.55,baseline={([yshift=-.5ex]current bounding box.center)}]
	\fill[fill=red!10!blue!20!gray!30!white] (0,0) circle (2cm);
	\draw (0,0) circle (2cm);
	\node at (2,0) {\textbullet};
	\node at (2.37,0) {$R$};

	\fill[fill=white] (0,0) circle (0.6cm);
	\draw (0,0) circle (0.6cm);
	\node at (0.6,0) {\textbullet};
	\node at (0.3,0) {$r$};

%	\filldraw[fill=white] (190:1.5cm) circle (.3cm); 
	\node at ([shift=(207:1.37cm)]0,0.3cm) {$z$};
	\node at ([shift=(187:1.50cm)]0,0.3cm) {$v$};
	\node at ([shift=(190:1.5cm)]0,0cm) {\textbullet};
\end{tikzpicture}
\,\,
\sim
R^{-L_0}Y(v,z)r^{L_0}
$$
For these values of $z$, we expect $R^{-L_0}Y(v,z)r^{L_0}$ to define a continuous (i.e., bounded) operator on the Hilbert space completion of $V$.
While this construction encodes the insertion operators of a VOA in an analytic language which is closer to that of conformal nets, the resulting operators are not local in the sense of algebraic conformal field theory.

A breakthrough idea of Henriques was to consider insertion operators localized in `degenerate' or `partially thin' annuli (see \cite{Henriques14}).
Specifically, we replace the operators $r^{L_0}$ and $R^{-L_0}$ with operators $A$ and $B$ corresponding to partially thin annuli supported in an interval $I$ of the unit circle.
We say that a unitary VOA $V$ has \emph{bounded localized vertex operators} if the von Neumann algebra $\cA_V(I)$ generated by insertions $BY(v,z)A$ supported in $I$ are bounded and local in the sense of ACFT, in which case $\cA_V(I)$ is a conformal net.
$$
\begin{tikzpicture}[baseline={([yshift=-.5ex]current bounding box.center)}]
	\coordinate (a1) at (150:1cm);
	\coordinate (b1) at (270:1cm);
	\coordinate (c1) at (210:1.5cm);
% BIG DISK
%	\draw (0,0) circle (1cm);
% CURVED BOUNDARY REGION
	\fill[red!10!blue!20!gray!30!white] (a1)  .. controls ++(250:.6cm) and ++(120:.4cm) .. (c1) .. controls ++(300:.4cm) and ++(190:.6cm) .. (b1) arc (270:510:1cm);
	\draw ([shift=(270:1cm)]0,0) arc (270:510:1cm);
	\coordinate (a) at (120:1cm);
	\coordinate (b) at (240:1cm);
	\coordinate (c) at (180:.25cm);
% BIG DISK
	\fill[fill=red!10!blue!20!gray!30!white] (0,0) circle (1cm);
%	\draw (0,0) circle (1cm);
% CURVED BOUNDARY REGION
	\fill[fill=white] (a)  .. controls ++(210:.6cm) and ++(90:.4cm) .. (c) .. controls ++(270:.4cm) and ++(150:.6cm) .. (b) -- ([shift=(240:1cm)]0,0) arc (240:480:1cm);
	\draw ([shift=(240:1cm)]0,0) arc (240:480:1cm);
	\draw (a) .. controls ++(210:.6cm) and ++(90:.4cm) .. (c);
	\draw (b) .. controls ++(150:.6cm) and ++(270:.4cm) .. (c);
	\draw (a) arc (120:150:1cm) (a1) .. controls ++(250:.6cm) and ++(120:.4cm) .. (c1);
	\draw (b1) .. controls ++(190:.6cm) and ++(300:.4cm) .. (c1);
%	\filldraw[fill=white] (197:.9cm) circle (.25cm);
	
%
	\node at (197:.9cm) {\textbullet};
	\node at (204:.7cm) {$z$};
	\node at (183:.9cm) {$v$};
%
%
%
% INNER DISK
%	\filldraw[fill=white] (180:.65cm) circle (.25cm); 
% COORDINATE LABELS
%	\node at (a) {(a)};
%	\node at (b) {(b)};
%	\node at (c) {(c)};
% INTERVAL LABEL
	\draw (275:1.45cm) -- (276:1.65cm);
	\draw (120:1.4cm) -- (120:1.6cm);
	\draw (120:1.5cm) arc (120:200:1.7cm) arc (200:290:1.3cm);
	\node at (180:1.9cm) {\scriptsize{$I$}};
\end{tikzpicture}
\,\, \sim
BY(v,z)A \in \cA_V(I)
$$
In \cite{GRACFT1,TenerGRACFT2} it was shown that many VOAs have bounded localized vertex operators, including all WZW models, their tensor products, subtheories, and so on.
A key step was to interpret the degenerate insertion operators as arising from Segal CFT, thereby simultaneously linking the three frameworks.
The geometric part of this argument was developed in \cite{Ten16}, and the analytic part is explored further in \cite{PutinarTener18}.
%In joint work in progress with Henriques, we expect to show that every conformal net arises from this construction. 

If $M$ is a unitary $V$-module, then the $\cA_V$-representation $\pi^M$ associated to $M$  (if it exists) is the one which takes insertion operators to insertion operators. 
That is, $\pi^M(B Y(v,z) A) = B Y^M(a,z) A$, where we have allowed the degenerate annuli $B$ and $A$ to act on the Hilbert space completion of $M$ via the Virasoro action (see Section \ref{sec: BLVO}).
The correspondence $M \leftrightarrow \pi^M$ was explored in \cite{TenerGRACFT2}.

\subsubsection*{Translating fusion between VOAs and conformal nets}

In this article we extend the above correspondence between vertex operator algebras and conformal nets and address fusion products of representations.
Our approach is informed by the geometry of Segal CFT, and unlike other approaches (e.g. \cite{CKLW18,GuiUnitarityI,GuiUnitarityII}) deals exclusively with continuous operators. 

The first challenge that one encounters is that the definitions of fusion product for conformal nets and VOAs are very different in nature.
For conformal nets, the fusion product of representations $\pi$ and $\lambda$ is given by composition of DHR endomorphisms \cite{DHR71,FrReSc92}, or equivalently the Connes-Sauvageot relative tensor product, known as `Connes fusion' in this context \cite{Wa98} (see also Section \ref{sec: conformal nets}).
The construction proceeds by putting a new inner product on a subspace of $\cH_\pi \otimes_{\bbC} \cH_\lambda$, and building a representation $\pi \boxtimes \lambda$ on its Hilbert space completion.
The fusion product of two representations always exists, and is given explicitly, although in practice it is quite difficult to identify the resulting representation (i.e. as a direct sum of known representations).

The VOA fusion product of modules, on the other hand, is generally defined by a universal property as opposed to an explicit construction.
The universal property makes explicit reference to a choice of category of modules, and the correct choices of modules (and universal property) are not clear for poorly behaved unitary VOAs (such as Virasoro VOAs with central charge $c \ge 1$).

If one accepts the existence of a correspondence between conformal net representations and VOA modules then the Connes fusion representation $\pi^M \boxtimes \pi^N$ of $\cA_V$ should be of the form $\pi^K$ for some $V$-module $K$.
Moreover, this VOA module $K$ should be given by an explicit construction, since the fusion product of conformal net representations $\pi^M \boxtimes \pi^N$ is.

We give a translation of Connes' fusion into the language of VOAs in Section \ref{sec: positivity}, where we introduce the notion of a transport module $M \boxtimes_t N$ for a pair of unitary $V$-modules $M$ and $N$.
Transport modules appear implicitly in the work of Wassermann \cite{Wa98} for type A WZW models, and again in the work of Gui \cite{GuiUnitarityI, GuiUnitarityII} for strongly rational VOAs.
In Section \ref{sec: positivity} we develop basic theory of transport modules in a way which does not depend on any special properties of the VOA, and in Section \ref{sec: positivity from BLVO} we show that the framework of bounded localized vertex operators is a powerful tool for establishing their existence.
As an application, we show the modular categories associated to WZW VOAs and type $ADE$ discrete series $W$-algebras are unitary in the remaining open cases (Theorem \ref{thm: UMTC for WZW and W}).

A key feature of our approach is that we are able to obtain positivity and unitarity results for VOAs simply from the existence of conformal net representations $\pi^M$, without requiring any hard analysis on intertwining operators.
In previous approaches, a model-by-model analysis of intertwining operators was often necessary, and this analysis often relied on model specific features of these operators.
In contrast, the treatment of examples in Section \ref{sec: examples} does not require any technical details, and the methods easily apply to other models as well.
In fact, a proof that every irreducible representation of a conformal net is of the form $\pi^M$ seems well within reach of our methods.

\subsubsection*{Geometric analysis of Connes' fusion}

Transport modules are a translation of Connes' fusion from the language of conformal nets to VOAs, and we use the framework of bounded localized vertex operators to make rigorous connections between the two settings.
We show in Section \ref{sec: transport and fusion} that if $\pi^M$ and $\pi^N$ exist, and there is a transport module $M \boxtimes_t N$, then there is a natural and explicit unitary $U: \cH_M \boxtimes \cH_N \to \cH_{M \boxtimes_t N}$ between the Hilbert space of conformal net fusion and the Hilbert space of the VOA transport module.
In the present article, we do not show in full generality that $U$ is an isomorphism between $\pi^M \boxtimes \pi^N$ and $\pi^{M \boxtimes_t N}$ (although such a result seems within reach of current techniques; see Section \ref{sec: outlook}).
Instead, we identify a maximal submodule $M \boxtimes_{loc} N \subset M \boxtimes_t N$ such that $\pi^{M \boxtimes_{loc} N}$ is embedded in $\pi^M \boxtimes \pi^N$ via $U$.

This identification is powerful enough to deduce finiteness of the index of subfactors and rationality of conformal nets.
A representation $\pi$ of a conformal net is said to have finite index if the associated Jones-Wassermann subfactor $\pi_I(\cA(I)) \subset \pi_{I^\prime}(\cA(I^\prime))^\prime$ has finite index, and a conformal net is said to be rational if it has finitely many irreducible representations, each with finite index.
Over the past several decades, the problem of showing specific examples are rational has proven to be an extremely difficult one, and arguments have relied on specific features of the model under consideration.
The primary difficulty is establishing the finiteness of the index of subfactors, or equivalently the rigidity of the representation category of the net.
Section \ref{sec: fusion rules for bimodules} provides the first general purpose tools for this problem.

\begin{thmalpha}\label{thmalpha: finite index}
Let $V$ be simple unitary VOA with bounded localized vertex operators, let $W$ be a unitary (not necessarily conformal) subalgebra which is regular, and let $M$ be a $W$-submodule of $V$.
Then $\pi^M$ has finite index.
\end{thmalpha}

The condition of regularity imposed on $W$ in Theorem~\ref{thmalpha: finite index} is equivalent to rationality plus $C_2$-cofiniteness in the context of simple unitary VOAs \cite[Thm. 4.5]{ABD04}.
The class of regular VOAs includes all familiar unitary rational VOAs, such as unitary Virasoro minimal models and WZW models at positive integer level.

\begin{coralpha}
Let $V$ be a simple unitary VOA with bounded localized vertex operators and let $W$ be a unitary conformal subalgebra which is regular.
Then the inclusion $\cA_W \subset \cA_V$ has finite index, so that $\cA_W$ is rational if and only if $\cA_V$ is.
\end{coralpha}

The proof of this theorem is quite flexible, and can be adapted to other models if desired.

Looking forward, future VOA research will immediately provide additional information about the corresponding conformal nets via Theorem \ref{thmalpha: finite index}.
The framework presented here allows for passing information back and forth between conformal nets and VOAs with a minimum of technical overhead, and we expect that it will provide new opportunities and inspiration for collaboration between VOA and conformal net specialists.

\subsubsection*{Structure of the article}

In Sections \ref{sec: unitary VOA} and \ref{sec: conformal nets} we give standard background information on unitary VOAs and conformal nets, respectively.
Section \ref{sec: BLVO} gives an overview of bounded localized vertex operators, adapted from \cite{TenerGRACFT2}.
Section \ref{sec: positivity} introduces transport modules from a purely VOA-theoretic perspective and establishes basic results relating transport, positivity, and fusion products.
Section \ref{sec: fusion rules for bimodules} connects transport modules to Connes fusion of conformal net representations.
Our main theorems are proven in this section, exploiting this connection between VOAs and conformal nets.
Section \ref{sec: examples} applies our main theorems to families of examples, primarily WZW models and $W$-algebras, and also describes opportunities for future work to unify the fusion product theory of VOAs and conformal nets.

\subsubsection*{Acknowledgments}

The author would like to thank A. Henriques for explaining his geometric idea of conformal nets, without which this article could not have been written. The author would also like to thank V.F.R. Jones, M. Bischoff, T. Creutzig, B. Gui, and R. McRae for helpful discussions related to the article.
The author gratefully acknowledges the hospitality and support of the Max Planck Institute for Mathematics, Bonn, which played a crucial role in this project.
The author was also supported by ARC Discovery Project DP200100067 and an AMS-Simons Travel Grant.

%% file: Sec_2_PreliminariesGRACFT3.tex
\newpage

\section{Background}\label{sec: preliminaries}

\subsection{Unitary VOAs and their representations}\label{sec: unitary VOA}

We assume that the reader has some familiarity with unitary vertex operator algebras (VOAs) and their representation theory.
For more detailed background on these subjects, we suggest \cite{TenerGRACFT2,GuiUnitarityI,CKLW18, DongLin14}.
For an overview of fusion product theory for VOAs, see e.g. \cite{CreutzigKanadeMcRae2017ax} for a concise treatment (or alternatively \cite{HuangLepowsky13} and related articles).
For completeness and convenience of the reader, we will now provide the relevant definitions and more specific references.

\subsubsection{Unitary VOAs}

\begin{Definition}\label{def: VOA}
A  vertex operator algebra is given by:
\begin{itemize}
\item a vector space $V$
\item vectors $\Omega,\nu \in V$ called the \emph{vacuum vector} and the \emph{conformal vector}, respectively.
\item a state-field correspondence $Y:V  \to \End(V)[[x^{\pm 1}]]$, denoted 
\begin{equation}
Y(a,x) = \sum_{n \in \Z} a_{(n)} x^{-n-1}.
\end{equation}
Here $\End(V)[[x^{\pm 1}]]$ is the vector space of formal series with coefficients in $\End(V)$.
\end{itemize}
This data must satisfy:
\begin{itemize}
%\item For every $a,b \in V$, we have $a_{(n)}b = 0$ for $n$ sufficiently large. 
\item For every $a \in V$, we have $a_{(n)}\Omega = 0$ for $n \ge 0$ and $a_{(-1)}\Omega = a$.
\item $Y(\Omega,x) = 1_V$.
\item For every  $a,b \in V$ and every $m,k,n \in \Z$, we have the Borcherds (or Jacobi) identity:
$$
\sum_{j = 0}^\infty \binom{m}{j} \big(a_{(n+j)}b\big)_{(m+k-j)} = \sum_{j=0}^\infty (-1)^j \binom{n}{j}\left( a_{(m+n-j)}b_{(k+j)} 
-(-1)^{n}b_{(n+k-j)}a_{(m+j)}\right) 
$$
\item If we write $Y(\nu,x) = \sum_{n \in \Z} L_n x^{-n-2}$, then the $L_n$ give a representation of the Virasoro algebra. 
That is, 
$$
[L_m,L_n] = (m-n)L_{m+n} + \tfrac{c}{12} (m^3-m)\delta_{m,-n}1_V
$$
for a number $c \in \C$, called the \emph{central charge}.
\item If we write $V_\alpha = \ker (L_0 - \alpha 1_V)$, then we have a decomposition of $V$ as an algebraic direct sum
$$
V = \bigoplus_{\alpha \in \Z_{\ge 0}} V_\alpha
$$
with $\dim V_\alpha < \infty$.
\item For every $a \in V$ we have $[L_{-1}, Y(a,x)] = \frac{d}{dx} Y(a,x)$.
\end{itemize}
\end{Definition}
Note that the eigenvalues of $L_0$ are non-negative by assumption.
If $a \in V_\alpha$, then we say that $a$ is homogeneous, with conformal dimension $\alpha$.

The special cases of the Borcherds identity with $n=0$ is called the Borcherds commutator formula, and it says
$$
%[a_{(m)},b_{(k)}] = \sum_{j =0}^\infty \binom{m}{j} \big(a_{(j)}b\big)_{(m+k-j)},
[a_{(m)}, Y(b,x)] = \sum_{j=0}^\infty \binom{m}{j} x^{m-j} Y(a_{(j)}b,x).
$$

There are natural notions of subalgebra, ideal, and homomorphism for VOAs, and a VOA with no ideals besides $\{0\}$ or $V$ is called simple (see \cite{FHL93} or \cite{CKLW18}).

\begin{Definition}\label{def: unitary VOA}
A \emph{unitary vertex operator algebra} is a VOA $V$ equipped with an inner product and an antilinear automorphism $\theta$ satisfying:
\begin{enumerate}
\item $\bip{b, Y(\theta a, \overline{x})c} = \bip{Y(e^{x L_1} (-x^{-2})^{L_0} a, x^{-1})b, c}$ for all $a,b,c \in V$.
\item $\ip{\Omega,\Omega} = 1$
\end{enumerate}
\end{Definition}
We write $\cH_V$ for the Hilbert space completion of $V$, consisting of vectors $\sum_{n \in \bbZ_{\ge 0}} v_n$ with $v_n \in V_n$ and $\sum \norm{v_n}^2 < \infty$.

A unitary VOA is simple if and only if $\dim V_0 = 1$, which means that simple unitary VOAs are of ``CFT type'' \cite[Prop. 5.3]{CKLW18}.
All VOAs considered in this article will be simple and unitary.
Some ways to construct new simple unitary VOAs from old ones include tensor products, fixed points under unitary automorphisms, or more generally unitary subalgebras and cosets with respect to unitary subalgebras (see \cite[\S5]{CKLW18} or \cite[\S2.2]{GRACFT1}).

\subsubsection{Unitary modules}

\begin{Definition}
Let $V$ be a vertex operator algebra.
A \emph{$V$-module} is given by a vector space $M$ along with a state-field correspondence $Y^M:V \to \End(M)[[x^{\pm 1}]]$, written
$$
Y^M(a,x) = \sum_{n \in \Z} a_{(n)}^M x^{-n-1},
$$
which is required to be linear, and satisfy the following additional properties.
\begin{itemize}
\item $Y^M(\Omega,x) = 1_M$.
\item For every $a \in V$ and $b \in M$, we have $a^M_{(n)}b = 0$ for $n$ sufficiently large.
\item For every homogeneous $a,b \in V$ and $m,k,n \in \Z$, the Borcherds(/Jacobi) identity holds:
\begin{align*}
\hspace{-.5cm}\sum_{j = 0}^\infty \binom{m}{j} \big(a_{(n+j)}b\big)^M_{(m+k-j)} = \sum_{j=0}^\infty (-1)^j \binom{n}{j}\left( a_{(m+n-j)}^Mb_{(k+j)}^M 
-(-1)^{n}  b_{(n+k-j)}^M a_{(m+j)}^M\right). 
\end{align*}
\item $M$ is graded by conformal dimensions.
That is, if $Y^M(\nu,x) =: \sum_{n \in \Z} L_n^M x^{-n-2}$ and we set $M_\alpha = \ker(L_0^M - \alpha 1_M)$ then we require that $\dim M_\alpha < \infty$ and
$$
M = \bigoplus_{\alpha \in \bbC} M_\alpha.
$$
\end{itemize}
\end{Definition}
Our definition of module is sometimes called a `strong' or `ordinary' module, in light of the requirement that $M$ be graded by finite-dimensional $L_0$-eigenspaces.
We restrict to this class of modules for convenience, although many of the ideas of this article apply in a broader context.

There are obvious notions of submodules and direct sums of $V$-modules, as well as $V$-module homomorphisms.
If $M$ has no proper, non-trivial submodules then it is called a \emph{simple} module.

We will be interested in unitary modules over unitary VOAs, which first appeared in \cite{DongLin14}.
\begin{Definition}
Let $V$ be a unitary VOA, and let $M$ be a $V$-module.
A sesquilinear form $\ip{ \, \cdot \, , \, \cdot \,}$ on $M$ is called \emph{invariant} if
\begin{equation}\label{eqnModuleInnerProductInvariance}
\bip{b,Y^M(\theta a,\overline{x})d} = \bip{Y^M(e^{x L_1} (-x^{-2})^{L_0} a, x^{-1})b,d},
\end{equation}
for all $a \in V$ and $b,d \in M$.
We call $M$ a \emph{unitary} module if it is equipped with an invariant inner product.
For the purpose of this article, we will also require that unitary modules have the property that only countably many of the weight spaces $M_\alpha$ are non-zero.
\end{Definition}
A unitary module $M$ has a Hilbert space completion $\cH_M$, and the requirement that only countably many weight spaces are non-zero is equivalent to requiring that $\cH_M$ be separable.
We impose this restriction to avoid considering non-separable Hilbert spaces.

Unitary modules provide positive energy representations of the Virasoro algebra, so $M_\alpha$ is non-zero only when $\alpha \ge 0$ (see \cite[Lem. 2.5]{DongLin14} or \cite[\S2]{TenerGRACFT2}).
There are natural notions of orthogonal direct sum and unitary isomorphism of unitary modules.
When talking about unitary VOAs and unitary modules, we will reserve the symbol $\bigoplus$ for orthogonal direct sums.
An infinite direct sum $\bigoplus M_i$ of non-zero unitary modules is again a unitary module if and only if the direct sum is a module (i.e. the $L_0$-eigenspaces are finite dimensional) and the collection of modules $M_i$ is countable.

Complete reducibility for unitary modules was shown in \cite[Lem. 2.5]{DongLin14}.
Since a simple unitary module has a unique invariant inner product up to scalar \cite[Prop. 2.19]{TenerGRACFT2}, this means that any unitary module $M$ can be orthogonally decomposed 
$$
M = \bigoplus M_i \otimes X_i
$$
where the $M_i$ are pairwise non-isomorphic simple unitary modules and the $X_i$ are finite-dimensional Hilbert spaces.

If $M$ is a unitary module, then the complex conjugate vector space $\overline{M}$ has a natural $V$-module structure given by 
$$
Y^{\overline{M}}(a,x)\overline{b} = \overline{Y^M(\theta a,x)b}
$$
where we have denoted by $b \mapsto \overline{b}$ the natural conjugate linear map $M \to \overline{M}$ (as well as its extension to the algebraic completion of $M$).
The reader is cautioned that $\overline{M}$ does not refer to the algebraic completion, as that notion will make only a brief appearance in this article.
An invariant inner product on $M$ can be understood as an isomorphism between $\overline{M}$ and the contragredient module $M^\prime$ (see \cite[Lem. 2.18]{TenerGRACFT2}).

\subsubsection{Intertwining operators}\label{sec: intertwining operators}

If $M$ and $N$ are vector spaces, we write $\cL(M,N)$ for the space of linear maps from $M$ to $N$, and $\cL(M,N)\{x\}$ for the space of all formal series 
$$
\sum_{n \in \C} a_{(n)} x^{-n-1}
$$
with $a_{(n)} \in \cL(M,N)$.
\begin{Definition}
Let $V$ be a VOA, and let $M, N$ and $K$ be $V$-modules.
An \emph{intertwining operator} of type $\binom{K}{M \, N}$ is a linear map $\cY: M \to \cL(N, K)\{x\}$, written
$$
\cY(a,x) = \sum_{n \in \C} a^{\cY}_{(n)} x^{-n-1},
$$
which satisfies the following properties. 
\begin{enumerate}
\item For every $a \in M$, $b \in N$ and $k \in \C$, we have $a^{\cY}_{(k+n)}b = 0$ for all sufficiently large $n \in \Z$.
\item For every $a \in M$, $\cY$ satisfies the $L_{-1}$-derivative property:
\begin{equation}
\cY(L_{-1} a, x) = \frac{d}{dx} \cY(a,x)
\end{equation}

\item For every homogeneous $a \in V$, $b \in M$, every $m,n \in \Z$, and every $k \in \C$, the Borcherds(/Jacobi) identity holds:
\begin{align*}
\hspace{-.5cm}\sum_{j = 0}^\infty \binom{m}{j} \big(a^M_{(n+j)}b\big)^\cY_{(m+k-j)} = \sum_{j=0}^\infty (-1)^j \binom{n}{j} \left(a_{(m+n-j)}^K b_{(k+j)}^\cY  
- (-1)^{n} b_{(n+k-j)}^\cY a_{(m+j)}^N\right).
\end{align*}
\end{enumerate}
We denote by $I \binom{K}{M \, N}$ the vector space of all intertwining operators of the indicated type.
\end{Definition}

The module operators $Y^M$ are intertwining operators of type $\binom{M}{V M}$, and $\dim I \binom{M}{V M} = 1$ precisely when $M$ is simple.

Specializing the Borcherds identity to $n=0$ produces the Borcherds commutator formula
$$
a_{(m)}^K \cY(b,x) - \cY(b,x) a_{(m)}^N = \sum_{j =0}^\infty \binom{m}{j}x^{m-j}\cY(a^M_{(j)}b,x).
$$
When $a = \nu$ is the conformal vector and $m=1$, we have 
\[
L^K_0\cY(b,x) - \cY(b,x)L^N_0 = x \cY(L_{-1}^M b, x) + \cY(L_0^M b, x) = x \frac{d}{dx} \cY(b,x) + \cY(L_0^M b, x).
\]
At the level of modes this says that 
\[
L_0^K b^\cY_{(n)} - b^\cY_{(n)} L_0^N = (-n-1 + \Delta_b) b ^\cY_{(n)}
\]
when $b$ is homogeneous with conformal dimension $\Delta_b$.
It follows that 
\[b_{(n)}^\cY N_\alpha \subset K_{\alpha-n-1+\Delta_b}.
\]
Thus if the conformal dimensions of all vectors in $M$, $N$, and $K$ are real (in particular, if they are unitary modules), then $b_{(n)}^\cY = 0$ unless $n \in \bbR$.
In the following we assume that this is the case.

\begin{Definition}
An intertwining operator $\cY \in I \binom{K}{M N}$ is called \emph{dominant} if
$$
K = \Span \{ a^\cY_{(n)} b : a \in M, b \in N, n \in \bbR\}.
$$
\end{Definition}
Of course any intertwining operator can be made dominant by replacing $K$ with a submodule.

If $M$, $N$, and $K$ are unitary modules and $\cY \in I \binom{K}{M N}$, then there is a unique intertwining operator $\cY^\dagger \in I \binom{N}{\overline{M} K}$ such that for all $a \in N$, $b \in M$, and $d \in K$
$$
\bip{\cY(\theta_{\overline{x}} a,\overline{x}^{-1})b, d} = 
\bip{b, \cY^\dagger(a,x)d}
$$
where $\theta_{\overline{x}} a = e^{\overline{x} L_1} e^{i \pi L_0} \overline{x}^{-2L_0} \overline{a}$.
Moreover, $\cY^{\dagger \dagger} = \cY$.
See \cite[\S1.3]{GuiUnitarityI} for more details.
As a consequence of the definition of unitary module and unitary VOA we have ${Y^M}^\dagger = Y^M$ after identifying $V \cong \overline{V}$ via $\theta$.

After module operators $Y^M$, the most basic intertwining operators associated to unitary modules are the annihilation and creation operators.

\begin{Definition}
Let $M$ be a unitary $V$-module.
The creation operator $\cY_M^+ \in I \binom{M}{M V}$ is given by 
$$
\cY_M^+(a,x)b = e^{x L_{-1}} Y^M(b,-x)a.
$$
The annihilation operator $\cY_M^- \in I \binom{V}{\overline{M} M}$ is given by $\cY_M^- = (\cY_M^+)^\dagger$.
\end{Definition}
The creation operator is obtained from $Y^M$ by the braiding \cite[\S5.4]{FHL93}, and does not require unitarity.
The creation and annihilation operators first appeared in the context of unitary VOAs in \cite[\S1]{GuiUnitarityI}.
Observe that creation and annihilation operators are compatible with direct sums of modules.

In this article we will be interested in analytic properties for intertwining operators, and the most fundamental thing we can ask is that an intertwining operator take values in the Hilbert space completion.

\begin{Definition}\label{def: Hilb intertwiner}
Let $V$ be a unitary VOA, let $M$, $N$, and $K$ be unitary $V$-modules. 
The space of intertwining operators with Hilbert space values, denoted $I_{Hilb} \vertex{K}{M}{N}$, consists of $\cY \in I \vertex{K}{M}{N}$ such that $\cY(a,z)b \in \cH_K$ whenever $0 < \abs{z} < 1$.
\end{Definition}

To interpret the condition $\cY(a,z)b \in \cH_K$, recall that in general $\cY(a,z)b$ lies in the algebraic completion $\prod_{n \in \bbR_{\ge 0}} K_{n}$.
Since the $K_{n}$ are pairwise orthogonal, we may realize $\cH_K$ as the subspace of the algebraic completion consisting of vectors whose components have square-summable norms.
It follows from basic properties of VOAs that $Y^M \in I_{Hilb}$ for any unitary $V$-module $M$ (as in \cite[Prop. 2.15]{GRACFT1}).

It is almost, but not quite, obvious that if $\cY \in I_{Hilb} \binom{K}{M N}$ then for $a \in M$ and $b \in N$ the expression $\cY(a,z)b$ defines a multi-valued holomorphic function from the punctured unit disk into $\cH_K$.
For completeness we include the details.

\begin{Lemma}\label{lem: Hilb intertwiners give holomorphic functions}
Let $V$ be a unitary VOA, let $M$, $N$, and $K$ be unitary $V$-modules, and let $\cY \in I_{Hilb} \vertex{K}{M}{N}$.
Then for every $a \in M$ and $b \in N$, the map $z \mapsto \cY(a,z)b$ defines a multi-valued holomorphic function on the punctured unit disk $\{1 > \abs{z} > 0\}$ with values in $\cH_K$, equipped with the norm topology.
\end{Lemma}
\begin{proof}
By standard results \cite[Ch. 5]{TaylorLay} it suffices to prove that $\ip{\cY(a,z)b,\xi}$ is holomorphic for every $\xi \in \cH_K$.\footnote{Note that it does not suffice to prove that $\ip{\cY(a,z)b,d}$ is holomorphic for every $d \in K$, which necessitates additional care. An alternate approach would proceed by showing that $\norm{\cY(a,z)b}$ is locally bounded.}
It also suffices to establish the result when $a$ and $b$ are homogeneous.
By definition, since $\cY \in I_{Hilb}$ we have
$$
\norm{\cY(a,z)b}^2 = \sum_{n \in \bbR} \|{a_{(n)}^\cY b z^{-n-1}}\|^2
$$
and thus $\sum_{n \in \bbR} a_{(n)}^{\cY} b z^{-n-1}$ converges in $\cH_K$ to $\cY(a,z)b$.
It follows that for every $z$ in the punctured unit disk we have 
$$
\bip{\cY(a,z)b,\xi} = \sum_{n \in \bbR} \bip{a_{(n)}^\cY b,\xi} z^{-n-1}. %= \sum_{n \in \bbR_{\le -1}} \bip{a_{(n)}^\cY b,\xi} z^{-n-1} + \sum_{n \in \bbR_{> -1}} \bip{a_{(n)}^\cY b,\xi} z^{-n-1}
$$
with pointwise absolute convergence.
Note that $a_{(n)}^\cY b = 0$ for $n \gg 0$.
By a straightforward application of the Weierstrass M-test this series converges uniformly on any neighborhood of $z$ which has been equipped with a branch of the logarithm and which is compactly contained in the punctured unit disk.
Thus $\ip{\cY(a,z)b,\xi}$ is locally given by a single-valued holomorphic function.
\end{proof}

\subsubsection{Modes of intertwining operators as unbounded operators}\label{sec: unbounded operators}

We briefly present the basic notions of unbounded operators on Hilbert spaces.
More detail can be found in any standard functional analysis textbook (e.g. \cite{AnalysisNOW}).

An unbounded operator on a Hilbert space $\cH$ is a pair $(x,D)$ where $D \subseteq \cH$ is a subspace of $\cH$ (generally required to be dense), and $x:D \to \cH$ is a linear map.
The space $D$ is called the domain of $x$.
An unbounded operator is called closed if its graph $\Gamma := \{ (\xi,x\xi) : \xi \in D\}$ is a closed subspace of $\cH \oplus \cH$, and it is called closable if the closure $\overline{\Gamma}$ is the graph of a (necessarily closed) operator, which is called the closure of $x$.
A dense subspace $D^\prime \subset D$ is called a core for a closed operator $x$ if the closure of $x|_{D^\prime}$ is equal to $x$.

The adjoint of an unbounded operator $(D,x)$ is the operator $(D^*,x^*)$ with domain given by 
$$
D^* = \{ \xi \in \cH : D \ni \eta \mapsto \ip{\xi,x\eta} \mbox{ is bounded}\}.
$$
We then use the Riesz representation theorem to characterize $x^*\xi$ by $\ip{x^*\xi,\eta} = \ip{\xi, x\eta}$ for all $\xi \in D^*$ and $\eta \in D$.
An operator $(D,x)$ is closable if and only if $D^*$ is dense.

It was observed in \cite[\S6]{CKLW18} that if $V$ is a unitary VOA and we regard $v_{(n)}$ as an unbounded operator with domain $V$, then $V$ is also contained in the domain of ${v_{(n)}}^*$.
Thus the modes $v_{(n)}$ are closable operators.
The same argument shows that the modes of a unitary module $v_{(n)}^M$ are closable.

The reader is cautioned that subtleties regarding domains of unbounded operators are somewhat notorious.
It will be an extremely important fact for us (Lemma \ref{lem: intertwiner image is a core}) that if $\cY \in I_{Hilb}\binom{K}{M N}$ is dominant then $\Span \{ \cY(a,z)b : a \in M, b \in N\}$ is a core for the closure of $v_{(n)}^K$ whenever $0 < \abs{z} < 1$.

\subsection{Conformal nets and Connes' fusion}\label{sec: conformal nets}

For a detailed introduction to conformal nets, we suggest \cite[\S1-3]{CKLW18}.
To fix notation we give the definition and basic notions.
In this article, all Hilbert spaces are assumed separable.

We denote by $\Diff_c(S^1)$ a certain central extension of $\Diff(S^1)$, the group of orientation preserving diffeomorphisms of the circle, by $U(1)$.
Representations of $\Diff_c(S^1)$ in which the central $U(1)$ acts standardly (i.e. $w \in U(1)$ acts as multiplication by $w$) are the same thing as projective representations of $\Diff(S^1)$ with central charge $c$ (see \cite[\S2.2]{HenriquesColimits} and \cite[\S II.2]{KhesinWendt}).
We will use the term ``representation of $\Diff_c(S^1)$'' to mean a representation in which the central $U(1)$ acts standardly.

For an interval $I \subset S^1$, we denote by $\Diff(I)$ the subgroup of $\Diff(S^1)$ consisting of diffeomorphisms which act as the identity outside of $I$, and write $\Diff_c(I)$ for the central extension of $\Diff(I)$ obtained by pulling back the extension $\Diff_c(S^1)$.
There is a canonical embedding 
$$
\Mob(\bbD) := PSU(1,1) \hookrightarrow \Diff_c(S^1).
$$
We denote by $\Diff_c^{(\infty)}(S^1)$ the $\bbZ$-central extension of $\Diff_c(S^1)$ reflecting the covering $\bbR \to S^1$, and by $\Mob^{(\infty)}(\bbD)$ the induced central extension of $\Mob(\bbD)$.
Further details on these extensions may be found in \cite{Weiner05, CKLW18, HenriquesColimits}.

A conformal net is a $\Diff(S^1)$-covariant family of von Neumann algebras indexed by intervals of the circle, where an \emph{interval} $I \subset S^1$ is an open, connected, non-empty, non-dense subset.
We denote by $\cI$ the set of all intervals.
If $I \in \cI$, we denote by $I^\prime$ the complementary interval.
\begin{Definition}
A conformal net of central charge $c$ is given by the data:
\begin{enumerate}
\item A Hilbert space $\cH$.
\item A strongly continuous unitary representation $U:\Diff_c(S^1) \to \cU(\cH)$.
\item For every $I \in \cI$, a von Neumann algebra $\cA(I) \subset \cB(\cH)$.
\end{enumerate}
The data is required to satisfy:
\begin{enumerate}
\item If $I,J \in \cI$ and $I \subset J$, then $\cA(I) \subset \cA(J)$.
\item If $I,J \in \cI$ and $I \cap J = \emptyset$, then $[\cA(I), \cA(J)] = \{0\}$.
\item $U(\gamma)\cA(I)U(\gamma)^* = \cA(\gamma(I))$ for all $\gamma \in \Diff_c(S^1)$, and $U(\gamma)xU(\gamma)^* = x$ when $x \in \cA(I)$ and $\gamma \in \Diff_c(I^\prime)$.
\item There is a unique (up to scalar) unit vector $\Omega \in \cH$, called the \emph{vacuum vector}, which satisfies $U(\gamma)\Omega = \Omega$ for all $\gamma \in \Mob(\bbD)$.
This vacuum vector is required to be cyclic for the von Neumann algebra $\cA(S^1):=\bigvee_{I \in \cI} \cA(I)$.
\item The generator $L_0$ of the one-parameter group $U(r_{\theta})$ is positive.
\end{enumerate}
\end{Definition}

Some of the key properties of conformal nets are (see \cite[\S3.3]{CKLW18}):
\begin{itemize}
\item (Haag duality) $\cA(I^\prime) = \cA(I)^\prime$
\item (Reeh-Schlieder) $\cH = \overline{\cA(I)\Omega}$ for every $I \in \cI$.
\item (Additivity) If a family of intervals $I_i$ satisfies $I = \bigcup_i I_i$, then $\cA(I)$ is generated as a von Neumann algebra by the subalgebras $\cA(I_i)$.
\item $\cA(I)$ is a type III factor for every interval $I \in \cI$, unless $\cH = \bbC$.
\end{itemize}

\begin{Definition}
A \emph{representation} of a conformal net $\cA$ is a Hilbert space $\cH_\pi$ and a family of representations (i.e. normal $*$-homomorphisms) $\pi_I:\cA(I) \to \cB(\cH_\pi)$, indexed by $I \in \cI$, which satisfy $\pi_{I}|_J = \pi_J$ when $J \subset I$.
\end{Definition}
The defining representation $\cH$ of $\cA$ is called the vacuum sector.
We point out that when $\cH_{\pi}$ is separable (as it always will be in this article) the normality of $\pi_I$ is automatic \cite[Thm. V.5.1]{TakesakiTOAI}.

By Haag duality, $U(\gamma) \in \cA(I)$ when $\gamma \in \Diff_c(I)$, and so given a representation of $\cA$ we obtain strongly continuous representations $\pi_I \circ U$ of every $\Diff_c(I)$ which are compatible with inclusions $I \subseteq J$.
By \cite[Thm. 12]{HenriquesColimits}, $\Diff_c^{(\infty)}(S^1)$ is the colimit of the $\Diff_c(I)$ along the inclusions $\Diff_c(I_1) \hookrightarrow \Diff_c(I_2)$ in the category of topological groups.
Thus the representations $\pi_I \circ U$ assemble to a strongly continuous representation 
$U^\pi$ of $\Diff_c^{(\infty)}(S^1)$ on $\cH_\pi$ (this was originally proven for irreducible representations in \cite{DFK04}).
The representation $\pi$ is covariant with respect to $U^\pi$ (as in \cite[Lem. 3.1]{KaLo04}), meaning that
$$
 U^\pi(\gamma) \pi_I(x) U^\pi(\gamma)^* = \pi_{\gamma(I)}(U(\gamma)xU(\gamma)^*)
$$
for all $\gamma \in \Diff_c^{(\infty)}(S^1)$.
Using \cite[Thm. 3.8]{Weiner06} one can show that the generator of rotation $U^\pi(r_\theta)$ has positive energy (see \cite[\S2.2]{TenerGRACFT2}).

If $\pi$ and $\lambda$ are representations of $\cA$, and $I \in \cI$, then the local intertwiners are given by
$$
\Hom_{\cA(I)}(\cH_\pi,\cH_\lambda) := \{ x \in \cB(\cH_\pi,\cH_\lambda) \, : \, x\pi_I(y) = \lambda_I(y)x  \mbox{ for all } y \in \cA(I)\}.
$$
Haag duality says that in the vacuum sector we have $\End_{\cA(I)}(\cH) = \cA(I^\prime)$.

We now discuss `fusion' products of representations of a conformal net.
Often this is done in the context of localized endomorphisms and DHR theory \cite{DHR71,FrReSc92,BKLR15}.
For our purposes it will be more convenient to use the ``Connes' fusion'' formulation of the fusion product first used by Wassermann \cite{Wa98} and put on solid categorical footing by Gui \cite{gui21categoricalextension}.

Given a choice of interval $I \in \cI$, the corresponding fusion product $\cH_\pi \boxtimes_I \cH_\lambda$ of two representations is the Hilbert space completion/quotient of the algebraic tensor product $\Hom_{\cA(I^\prime)}(\cH, \cH_\pi) \otimes_\bbC \Hom_{\cA(I)}(\cH,\cH_\lambda)$ with respect to the inner product
$$
\ip{x_1 \otimes y_1, x_2 \otimes y_2} = \ip{y_2^*y_1x_2^*x_1\Omega,\Omega}.
$$
We denote by $x \boxtimes_I y$ the image of $x \otimes y$ in $\cH_\pi \boxtimes_I \cH_\lambda$.
Observe that $\cA(I)$ and $\cA(I^\prime)$ act on $\cH_\pi \boxtimes \cH_\lambda$ by 
$$
(\pi \boxtimes_I \lambda)_I(r)(x \boxtimes_I y) = \pi_I(r) x \boxtimes_I y
$$
and
$$
(\pi \boxtimes_I \lambda)_{I^\prime}(r)(x \boxtimes_I y) = x \boxtimes_I \lambda_{I^\prime}(r)y.
$$

We embed $\Hom_{\cA(I)}(\cH,\cH_\lambda) \hookrightarrow \cH_\lambda$ by $y \mapsto y\Omega$.
The map $x \otimes y \mapsto x \boxtimes_I y$ extends continuously to a map $\Hom_{\cA(I)^\prime}(\cH,\cH_\pi) \otimes_\C \cH_\lambda \to \cH_\pi \boxtimes_I \cH_\lambda$, which we again denote by $\boxtimes_I$.
We also have a continuous extension to a map $\cH_\pi \otimes_\bbC \Hom_{\cA(I)}(\cH,\cH_\lambda) \to \cH_\pi \boxtimes_I \cH_\lambda$.
If $r \in \cA(I)$, we have $x r \boxtimes_I \xi = x \boxtimes_I \lambda_I(r)\xi$ for all $\xi \in \cH_\lambda$.

We want to define actions of $\cA(J)$ on $\cH_\pi \boxtimes \cH_\lambda$ for all intervals $J$ in such a way as to construct a representation of $\cA$.
When the Hilbert spaces are separable, the usual DHR argument shows the existence of such a representation.
If one has the equivalence of local and global intertwiners for localized endomorphisms (e.g. for strongly additive nets or when the sectors have finite index \cite[Thm. 2.3]{GuidoLongo96}) then such an extension is unique, but it is not clear that this is always the case.

Instead of the DHR argument, we follow the approach in \cite[\S2]{gui21categoricalextension} which gives a natural $\cA$-representation structure on $\cH_\pi \boxtimes_I \cH_\lambda$.
To describe this representation, we observe that when $\tilde I \subset I$ is a subinterval, the inclusion $\Hom_{\cA(\tilde I^\prime)}(\cH,\cH_\pi) \subset \Hom_{\cA(I^\prime)}(\cH,\cH_\pi)$ extends to a unitary map
$$
U_{I \leftarrow \tilde I} : \cH_\pi \boxtimes_{\tilde I} \cH_\lambda \to \cH_\pi \boxtimes_{I} \cH_\lambda.
$$
Gui showed that there is a (necessarily unique) way to endow each $\cH_\pi \boxtimes_I \cH_\lambda$ with the structure of a sector, compatible with the natural actions of $\cA(I)$ and $\cA(I^\prime)$, in such a way that the $U_{I \leftarrow \tilde I}$ are isomorphisms of sectors.
We denote this representation $\pi \boxtimes_I \lambda$, or $\pi \boxtimes \lambda$ when the particular interval $I$ is not important.
In the same article, Gui also shows that this `Connes' fusion' tensor product makes the category of $\cA$-representations into a braided tensor category.

\subsection{Bounded localized vertex operators}\label{sec: BLVO}

In \cite{GRACFT1,TenerGRACFT2} we introduced the notion of \emph{bounded localized vertex operators} as a way of formalizing the relationship between unitary VOAs and conformal nets (an alternate approach to the same problem was given in \cite{CKLW18}).
In these articles we gave definitions of what it means for a conformal net $\cA$ to come from a unitary VOA $V$, and what it means for an $\cA$-representation to come from a unitary $V$-module.
We will now describe the aspects of this relationship which are necessary for this article, and refer the reader to \cite{TenerGRACFT2} for more details.

The Segal-Neretin semigroup of annuli consists of compact Riemann surfaces which are isomorphic to annuli, equipped with boundary parametrizations \cite{SegalDef, Neretin90}.
Sufficiently nice representations of $\Diff_c^{(\infty)}(S^1)$ extend to representations of a central extension of the semigroup of annuli.
The annuli $\{1 > \abs{z} > r\}$ and $\{ R > \abs{z} > 1\}$ act by $r^{L_0}$ and $R^{-L_0}$, respectively.
If $V$ is a VOA, the operator $R^{-L_0}Y(v,z)r^{L_0}$ corresponds to inserting a state inside the annulus $\{R > \abs{z} > r\}$.
In general this should be a bounded operator.
We depict this situation:
$$
\hspace{-.1in}
\begin{tikzpicture}[scale=0.55,baseline={([yshift=-.5ex]current bounding box.center)}]
	\fill[fill=red!10!blue!20!gray!30!white] (0,0) circle (2cm);
	\draw (0,0) circle (2cm);
	\node at (2,0) {\textbullet};
	\node at (2.39,0) {$R$};

	\fill[fill=white] (0,0) circle (1cm);
	\draw (0,0) circle (1cm);
	\node at (1,0) {\textbullet};
	\node at (0.7,0) {$1$};
\end{tikzpicture}
\,\, \sim R^{-L_0},
\qquad
\begin{tikzpicture}[scale=0.55,baseline={([yshift=-.5ex]current bounding box.center)}]
	\fill[fill=red!10!blue!20!gray!30!white] (0,0) circle (1cm);
	\draw (0,0) circle (1cm);
	\node at (1,0) {\textbullet};
	\node at (1.3,0) {$1$};

	\fill[fill=white] (0,0) circle (0.6cm);
	\draw (0,0) circle (0.6cm);
	\node at (0.6,0) {\textbullet};
	\node at (0.3,0) {$r$};
\end{tikzpicture}
\,\, \sim r^{L_0}, 
\qquad
\begin{tikzpicture}[scale=0.55,baseline={([yshift=-.5ex]current bounding box.center)}]
	\fill[fill=red!10!blue!20!gray!30!white] (0,0) circle (2cm);
	\draw (0,0) circle (2cm);
	\node at (2,0) {\textbullet};
	\node at (2.37,0) {$R$};

	\fill[fill=white] (0,0) circle (0.6cm);
	\draw (0,0) circle (0.6cm);
	\node at (0.6,0) {\textbullet};
	\node at (0.3,0) {$r$};

%	\filldraw[fill=white] (190:1.5cm) circle (.3cm); 
	\node at ([shift=(207:1.37cm)]0,0.3cm) {$z$};
	\node at ([shift=(187:1.50cm)]0,0.3cm) {$v$};
	\node at ([shift=(190:1.5cm)]0,0cm) {\textbullet};
\end{tikzpicture}
\,\,
\sim
R^{-L_0}Y(v,z)r^{L_0}
$$
While the field $Y(v,z)$ is local in the Wightman sense, the operator $R^{-L_0}Y(v,z)r^{L_0}$ encoding the insertion of $v$ in the annulus is not local in the sense of algebraic conformal field theory.
That is, if $\cA_V$ is the conformal net which is associated to the same CFT, then the insertion in the annulus $\{R > \abs{z} > r\}$ should not lie in any $\cA_V(I)$.

A breakthrough idea of Henriques was to extend the semigroup of annuli to also include `degenerate' (or `partially thin') annuli, where the incoming boundary and outgoing boundary are allowed to overlap \cite{Henriques14}.
If one knows that the representation of the semigroup of annuli on $\cH_V$ extends to this larger semigroup, one can consider operators which insert states in `local' annuli.
Here, an annulus is localized in an interval $I$ if the incoming and outgoing boundary parametrizations agree on $I^\prime$.

Specifically, we wish to consider partially thin annuli $\Sigma$ whose outgoing boundary lies in $\{\abs{z} \ge 1\}$ and whose incoming boundary lies in $\{\abs{z} \le 1\}$.
We divide $\Sigma$ into two annuli, $B := \Sigma \cap \{\abs{z} \ge 1\}$ and $A := \Sigma \cap \{\abs{z} \le 1\}$.
The unit circle is the incoming boundary of $B$ and the outgoing boundary of $A$, and we parametrize both by the identity.
If $\pi$ is the representation of the semigroup of partially thin annuli on the positive energy representation $\cH_V$, then the operator which inserts the state $v$ at $z \in \Sigma$ is $\pi(B)Y(v,z)\pi(A)$.
We depict the situation as follows:
$$
\begin{tikzpicture}[baseline={([yshift=-.5ex]current bounding box.center)}]
	\coordinate (a) at (150:1cm);
	\coordinate (b) at (270:1cm);
	\coordinate (c) at (210:1.5cm);
% BIG DISK
	\draw (0,0) circle (1cm);
% CURVED BOUNDARY REGION
	\fill[red!10!blue!20!gray!30!white] (a)  .. controls ++(250:.6cm) and ++(120:.4cm) .. (c) .. controls ++(300:.4cm) and ++(190:.6cm) .. (b) arc (270:510:1cm);
	\filldraw[fill=white] (0,0) circle (1cm);
	\draw ([shift=(270:1cm)]0,0) arc (270:510:1cm);
	\draw (a) .. controls ++(250:.6cm) and ++(120:.4cm) .. (c);
	\draw (b) .. controls ++(190:.6cm) and ++(300:.4cm) .. (c);
% INNER DISK
%	\filldraw[fill=white] (180:.65cm) circle (.25cm); 
% COORDINATE LABELS
%	\node at (a) {(a)};
%	\node at (b) {(b)};
%	\node at (c) {(c)};
\end{tikzpicture}
\,\, \sim B,
\qquad
\begin{tikzpicture}[baseline={([yshift=-.5ex]current bounding box.center)}]
	\coordinate (a) at (120:1cm);
	\coordinate (b) at (240:1cm);
	\coordinate (c) at (180:.25cm);
% BIG DISK
	\fill[fill=red!10!blue!20!gray!30!white] (0,0) circle (1cm);
	\draw (0,0) circle (1cm);
% CURVED BOUNDARY REGION
	\fill[fill=white] (a)  .. controls ++(210:.6cm) and ++(90:.4cm) .. (c) .. controls ++(270:.4cm) and ++(150:.6cm) .. (b) -- ([shift=(240:1cm)]0,0) arc (240:480:1cm);
	\draw ([shift=(240:1cm)]0,0) arc (240:480:1cm);
	\draw (a) .. controls ++(210:.6cm) and ++(90:.4cm) .. (c);
	\draw (b) .. controls ++(150:.6cm) and ++(270:.4cm) .. (c);
% INNER DISK
%	\filldraw[fill=white] (180:.65cm) circle (.25cm); 
% COORDINATE LABELS
%	\node at (a) {(a)};
%	\node at (b) {(b)};
%	\node at (c) {(c)};
\end{tikzpicture}
\,\, \sim A,
\qquad
\begin{tikzpicture}[baseline={([yshift=-.5ex]current bounding box.center)}]
	\coordinate (a1) at (150:1cm);
	\coordinate (b1) at (270:1cm);
	\coordinate (c1) at (210:1.5cm);
% BIG DISK
%	\draw (0,0) circle (1cm);
% CURVED BOUNDARY REGION
	\fill[red!10!blue!20!gray!30!white] (a1)  .. controls ++(250:.6cm) and ++(120:.4cm) .. (c1) .. controls ++(300:.4cm) and ++(190:.6cm) .. (b1) arc (270:510:1cm);
	\draw ([shift=(270:1cm)]0,0) arc (270:510:1cm);
	\coordinate (a) at (120:1cm);
	\coordinate (b) at (240:1cm);
	\coordinate (c) at (180:.25cm);
% BIG DISK
	\fill[fill=red!10!blue!20!gray!30!white] (0,0) circle (1cm);
%	\draw (0,0) circle (1cm);
% CURVED BOUNDARY REGION
	\fill[fill=white] (a)  .. controls ++(210:.6cm) and ++(90:.4cm) .. (c) .. controls ++(270:.4cm) and ++(150:.6cm) .. (b) -- ([shift=(240:1cm)]0,0) arc (240:480:1cm);
	\draw ([shift=(240:1cm)]0,0) arc (240:480:1cm);
	\draw (a) .. controls ++(210:.6cm) and ++(90:.4cm) .. (c);
	\draw (b) .. controls ++(150:.6cm) and ++(270:.4cm) .. (c);
	\draw (a) arc (120:150:1cm) (a1) .. controls ++(250:.6cm) and ++(120:.4cm) .. (c1);
	\draw (b1) .. controls ++(190:.6cm) and ++(300:.4cm) .. (c1);
%	\filldraw[fill=white] (197:.9cm) circle (.25cm);
	
%
	\node at (197:.9cm) {\textbullet};
	\node at (204:.7cm) {$z$};
	\node at (183:.9cm) {$v$};
%
%
%
% INNER DISK
%	\filldraw[fill=white] (180:.65cm) circle (.25cm); 
% COORDINATE LABELS
%	\node at (a) {(a)};
%	\node at (b) {(b)};
%	\node at (c) {(c)};
\end{tikzpicture}
\,\, \sim
\pi(B)Y(v,z)\pi(A)
$$
We now attempt to construct a conformal net $\cA_V$ whose local algebras $\cA_V(I)$ are generated by operators of the form $\pi(B)Y(v,z)\pi(A)$, as $v$ runs over all of the states of $V$, $(B,A)$ runs over all pairs $(B, A)$ localized in $I$, and $z \in \Sigma = B \circ A$.

At present, the theory of the semigroup of degenerate annuli and its representations has not been completed, although we hope to take up this question in future work with Henriques.
What we do have is a way of assigning bounded operators to certain degenerate annuli, and direct verification that these operators behave in many of the ways that one would expect if they did come from a semigroup representation.
We do not wish to focus too much on the specifics of these operators, and so we abstract out certain key properties that they satisfy.
We call the resulting structure a \emph{system of generalized annuli}.
Just below we will introduce systems of generalized annuli, and then later discuss how to carry out the construction of conformal nets from VOAs in this context.

\subsubsection{Systems of generalized annuli}

%A detailed discussion of generalized annuli is given in \cite[\S3]{TenerGRACFT2}.
%We will only use certain properties of generalized annuli in this article, and so we omit certain aspects of the definition here.

For $c \in \{ 1 - \frac{6}{(m+2)(m+3)} : m \ge 1 \} \cup [1,\infty)$, let $L(c,0)$ be the unitary irreducible representation of the Virasoro algebra with central charge $c$ and let $\cH_{c,0}$ be its Hilbert space completion.
For every such $c$ there is a strongly continuous unitary representation $U_{c,0}$ of $\Diff_c(S^1)$ on $\cH_{c,0}$, and the local diffeomorphisms generate a conformal net $\cA_c$ called the Virasoro net of central charge $c$.
For more details, see \cite[\S2.4]{Carpi04} and references therein.

As described above, the Hilbert spaces $\cH_{c,0}$ carry (projective) representations of the Segal-Neretin semigroup of annuli, and in future work with Henriques we will show that these representations extend to a larger semigroup containing ``partially thin'' annuli (with the incoming and outgoing boundary allowed to coincide).
This extended representation would provide a one-dimensional space of bounded operators $\pi(A)$ on $\cH_{c,0}$ for each region $A$ of the form $A=D_2 \setminus \overline{D_1}$, where $D_1 \subseteq D_2$ are Jordan domains in $\bbC$ with smooth boundary, equipped with parametrizations of the boundaries $\partial D_i$ by the unit circle $S^1$.
Then for each interval $I \subset S^1$ we could form the collection of operators $\scA_I^{in} \subset \cB(\cH_{c,0})$ of the form $\pi(A)$ where $A = \bbD \setminus \overline{D}$,  $\bbD$ is the unit disk in $\bbC$, $D \subset \bbD$ a Jordan domain which has smooth boundary and which contains $0$, $S^1 = \partial \bbD$ is parametrized by the identity, and $\partial D$ is parametrized by the unit circle in such a way that $I$ is mapped into $\partial D \cap S^1$.

Similarly we could form the collections of operators $\scA_I^{out} \subset \cB(\cH_{c,0})$ of the form $\pi(B)$ where $B = D \setminus \overline{\bbD}$ where $D$ is a Jordan domain with smooth boundary that contains $\bbD$, the boundary $\partial \bbD$ is parametrized by the identity, and $\partial D$ is parametrized in such a way that $I$ maps into $\partial D \cap S^1$.
We could then consider the collection $\scA_I \subset \scA_I^{out} \times \scA_I^{in}$ consisting of pairs $(\pi(B),\pi(A))$ where $B=D_2 \setminus \overline{\bbD}$ and $A = \bbD \setminus \overline{D_1}$ as above and the parametrizations of $\partial D_2$ and $\partial D_1$ coincide when restricted to the interval $I$.
The operator $\pi(B)\pi(A)$ would be the one assigned to the annulus $D_2 \setminus \overline{D_1}$ by the representation of the semigroup of partially thin annuli, and would lie in $\cA_c(I)$.

However, as indicated above, the joint work in progress to construct representations of the semigroup of partially thin annuli has not yet been completed, and thus we cannot form the sets of operators $\scA_I^{in/out}$ and $\scA_I$ just described.
Since operators assigned to annuli will play a major role in the next section when we describe the correspondence between vertex operator algebras and conformal nets, we will require an alternative to the above construction.
We will now introduce the placeholder notion of ``system of generalized annuli'' which consists of families of bounded operators $\scA_I^{in/out}$ and $\scA_I$ enjoying many of the properties that one would expect from the above construction via representations of partially thin annuli.
We now give the full list of properties required of an abstract system of generalized annuli, and then we will describe the system of generalized annuli that was constructed in \cite[\S 3]{TenerGRACFT2}.

\begin{Definition}[Left and right localizable operators]\label{def: localizable operators}
Let $U=U_{c,0}$ be the vacuum representation of $\Diffc$ on $\cH_{c,0}$.
An operator $A \in \cB(\cH_{c,0})$ is \emph{left localizable in $I$} if there exists $\gamma \in \Diffc$ such that $U(\gamma)A \in \cA_c(I)$.
Similarly, $B \in \cB(\cH_{c,0})$ is \emph{right localizable in $I$} if there exists $\gamma \in \Diffc$ such that $AU(\gamma) \in \cA_c(I)$.
We write $\Ann^{\ell}_I(\cH_{c,0})$ for the class of operators which are left localizable in $I$ and which have dense image.
We write $\Ann^{r}_I(\cH_{c,0})$ for the class of operators which are right localizable in $I$ and are injective.
We abbreviate these to $\Ann^{\ell/r}_I$ when $c$ is understood from context.
\end{Definition}

\begin{Definition}[System of incoming/outgoing generalized annuli]
\label{def: system of incoming annuli}
A system of incoming generalized annuli with central charge $c$ is given by a family of subsets $\scA^{in}_I \subset \Ann^{\ell}_I(\cH_{c,0})$ for all intervals $I$, along choices of open subsets $\Int(A) \subset \bbC$ for every $A \in \scA^{in}_I$.
These must satisfy the following conditions:
\begin{enumerate}
\item $1 \in \scA^{in}_I$ for all $I$
\item If $I \subset J$ then $\scA^{in}_I \subset \scA^{in}_J$ 
\item If $A \in \scA^{in}_I$ then $U(r_{\theta})AU(\gamma)^* \in \scA^{in}_{\gamma(I)}$ for all $\gamma \in \Diffc$ and all $\theta \in \bbR$.
\item For every $A \in \scA^{in}_I$, the open set $\Int(A)$ is of the form $\Int(A) = \bbD \setminus \overline{D}$ where $\bbD \subset \bbC$ is the open unit disk and $D \subset \bbD$ is a Jordan domain with $C^\infty$ boundary such that $0 \in D$.
\item The interior function satisfies $\Int(r_{\theta} A U(\gamma)^*) = r_{\theta}(\Int(A))$.
\item If $U(\gamma)A \in \cA_c(I)$, then $\gamma^{-1}(I^\prime) \subseteq \partial D \cap S^1$.
\end{enumerate}
The associated system of generalized outgoing annuli is given by the adjoints 
$$
{\scA_I^{out}}:=\{A^* : A \in \scA_I^{in}\}.
$$
If $B \in \scA_I^{out}$ then we define $\Int(B) = \{\overline{z}^{-1} : z \in \Int(B^*)\}$.
\end{Definition}

\begin{Definition}[System of generalized annuli of central charge $c$]
\label{def: system of annuli central charge c}
We denote by $\Ann_I$ the collection of pairs $(B,A)$ with $A \in \Ann^{\ell}_I$ and $B \in \Ann^{r}_I$ for which there exists a common $\gamma \in \Diffc$ such that $U(\gamma)A, BU(\gamma)^* \in \cA_c(I)$.
A system of generalized annuli of central charge $c$ is a system of incoming and outgoing generalized annuli $\scA^{in/out}_I$, and a family of subsets $\scA_I \subset {\scA^{out}_I} \times \scA_I^{in}$ such that
\begin{enumerate}
\item $\scA_I \subset \Ann_I$
\item if $I \subset J$ then $\scA_I \subset \scA_J$
\item if $(B,A) \in \scA_I$ then $(A^*,B^*) \in \scA_I$
\item if $(B,1) \in \scA_I$ and $(1,A) \in \scA_I$ then $(B,A) \in \scA_I$
\item if $\gamma \in \Diff_c(I)$, then $(1,U(\gamma)) \in \scA_I$
\item if $\gamma \in \Diffc$, $r_\theta \in \Rot(S^1)$, $(B,A) \in \scA_I$, then $(U(\gamma)BU(r_{\theta})^*, U(r_\theta)AU(\gamma)^*) \in \scA_{\gamma(I)}$.
\item if $r_\theta \in \Rot(S^1)$, $I$ is an interval containing both the intervals $J$ and $r_\theta(J)$, and $(B,A) \in \scA_J$, then $(B, U(r_{\theta})AU(r_\theta)^*) \in \scA_I$
\item for all $I$ there exists an $A \in \scA_I^{in}$ such that $(1,A) \in \scA_I$ and $\Int(A) \ne \emptyset$.
\end{enumerate}
For $(B,A) \in \scA_I$, we define 
$$
\Int(B,A) = {(\cl(\Int(A)) \cup \cl(\Int(B))}\interior{\,\,}.
$$ 
\end{Definition}

In \cite[\S3.4]{TenerGRACFT2}, we constructed a system of generalized annuli explicitly in terms of exponentials of smeared Virasoro fields.
We will give a brief overview of this construction here, with technical details left to the reference.

Consider a one-parameter semigroup $\varphi_t$ of univalent self-maps of the unit disk $\mathbb{D}$ which fix the point $0$, and let $\rho$ be the holomorphic vector field that generates the semigroup.
Let $L(\rho)$ be the unbounded operator $\sum_{n \in \bbZ} L_n \hat \rho(n)$ acting on $L(c,0) \subset \cH_{c,0}$, where $\hat \rho(n)$ are the Fourier coefficients of $\rho$.
Then $\scA_I^{in}$ consists of all operators $e^{-tL(\rho)}U(\gamma)$, as $\rho$ ranges over generators of one-parameter semigroups of univalent maps $\varphi_t$ as above, $t \ge 0$, and $\gamma \in \Diff_c(S^1)$ is required to be compatible with the map $\varphi_t$ and the interval $I$ via the condition $\varphi_t(\gamma^{-1}(I')) \subset S^1$.
The interior is given by $\Int(e^{-tL(\rho)}U(\gamma)) = \bbD \setminus \overline{\varphi_t(\bbD)}$.
The system of generalized annuli is given by pairs $(B,A)$ with $B^*,A \in \scA_I^{in}$ such that $BA \in \cA_c(I)$.
It is shown in \cite[Lem. 3.23]{TenerGRACFT2} that this is indeed a system of generalized annuli.

%The abstract notion of a system of generalized annuli also allows the possibility to replace the systems constructed in \cite{TenerGRACFT2} with arbitrary degenerate annuli once the representation theory of this semigroup has been established.

\subsubsection{VOAs with bounded localized vertex operators}

Given a fixed system of generalized annuli of central charge $c$ and a unitary VOA with the same central charge, we would like to define insertion operators $BY(a,z)A$ for all $(B,A) \in \scA_I$.
The generalized annuli $A$ and $B$ are operators on the vacuum Hilbert space $\cH_{c,0}$ of the Virasoro net.
We use the Virasoro net to transport the action of $A$ and $B$ to $\cH_V$, noting that $\cH_V$ decomposes as a direct sum of irreducible representations of the Virasoro algebra and therefore carries a natural representation $\pi$ of $\cA_c$.
Thus for left localizable annuli $A \in \Ann^\ell_I$ we define an action $\pi_I(A) = U^\pi(\gamma)^*\pi_I(U(\gamma)A)$ which is independent of a choice of localizing unitary $\gamma$.
More generally, if $\pi$ is any representation of $\cA_c$, we have an action $\pi_I(A) = U^\pi(\tilde \gamma)^*\pi_I(U(\gamma)A)$ for any lift $\tilde \gamma$ of $\gamma$ to $\Diff_c^{(\infty)}(S^1)$, which is well-defined up to a multiple of $e^{2 \pi i L_0}$.
We can similarly define the action of right localizable annuli, and we have $\pi_I(A)^* = \pi_I(A^*)$.
These actions are described in detail in \cite[\S3.2]{TenerGRACFT2}.
For simplicity, we drop the notation $\pi_I(A)$ and simply let $A$ act on representations of $\cA_c$, and similarly for $B$.

We will interpret $BY(v,z)A$, as a sesquilinear form.
We introduce the notion of incoming and outgoing annuli to ensure that it is densely defined.
\begin{Definition}\label{def: incoming and outgoing annuli}
Let $M$ be a unitary representation of the Virasoro algebra on which $L_0$ is diagonalizable, and let $\cH$ be the Hilbert space completion of $M$.
An \emph{incoming generalized annulus} on $\cH$ is an operator $A \in \cB(\cH)$ such that $M \subset \Ran(A)$ and $A^{-1}(M)$ is dense in $\cH$.
An operator $B \in \cB(\cH)$ is an \emph{outgoing generalized annulus} if $A^*$ is an incoming generalized annulus.
We write $\Ann^{in}(\cH)$ and $\Ann^{out}(\cH)$ for the set of incoming and outgoing generalized annuli, respectively.
\end{Definition}

If $B \in \Ann^{out}(\cH_V)$ and $A \in \Ann^{in}(\cH_V)$, then we can evaluate $\ip{BY(v,z)A \xi, \eta} = \ip{Y(v,z)(A\xi), (B^*\eta)}$ whenever $A\xi \in V$ and $B^*\eta \in V$.
Thus $BY(v,z)A$ is a densely defined sesquilinear form, which may or may not extend to a bounded operator.

\begin{Definition}[Bounded insertions] \label{def: bounded insertions}
Let $V$ be a simple unitary VOA of central charge $c$ and fix a system $\scA$ of generalized annuli with the same central charge.
We say that $V$ has bounded insertions for $\scA$ if for every interval $I \in \cI$ and every $(B,A) \in \scA_I$ we have $B \in \Ann^{out}(\cH_V)$, $A \in \Ann^{in}(\cH_V)$ and for every $z \in \Int(B,A)$ there is some $s > 0$ such that $BY(s^{L_0} \cdot,z)A$ defines a bounded operator $\cH_V \otimes \cH_V \to \cH_V$.
\end{Definition}

The map $BY(s^{L_0} \cdot,z)A$ corresponds to the picture

$$
\begin{tikzpicture}[baseline={([yshift=-.5ex]current bounding box.center)}]
	\coordinate (a1) at (150:1cm);
	\coordinate (b1) at (270:1cm);
	\coordinate (c1) at (210:1.5cm);
% BIG DISK
%	\draw (0,0) circle (1cm);
% CURVED BOUNDARY REGION
	\fill[red!10!blue!20!gray!30!white] (a1)  .. controls ++(250:.6cm) and ++(120:.4cm) .. (c1) .. controls ++(300:.4cm) and ++(190:.6cm) .. (b1) arc (270:510:1cm);
	\draw ([shift=(270:1cm)]0,0) arc (270:510:1cm);
	\coordinate (a) at (120:1cm);
	\coordinate (b) at (240:1cm);
	\coordinate (c) at (180:.25cm);
% BIG DISK
	\fill[fill=red!10!blue!20!gray!30!white] (0,0) circle (1cm);
%	\draw (0,0) circle (1cm);
% CURVED BOUNDARY REGION
	\fill[fill=white] (a)  .. controls ++(210:.6cm) and ++(90:.4cm) .. (c) .. controls ++(270:.4cm) and ++(150:.6cm) .. (b) -- ([shift=(240:1cm)]0,0) arc (240:480:1cm);
	\draw ([shift=(240:1cm)]0,0) arc (240:480:1cm);
	\draw (a) .. controls ++(210:.6cm) and ++(90:.4cm) .. (c);
	\draw (b) .. controls ++(150:.6cm) and ++(270:.4cm) .. (c);
	\draw (a) arc (120:150:1cm) (a1) .. controls ++(250:.6cm) and ++(120:.4cm) .. (c1);
	\draw (b1) .. controls ++(190:.6cm) and ++(300:.4cm) .. (c1);
	\filldraw[fill=white] (197:.9cm) circle (.25cm);
\end{tikzpicture}
$$

\begin{Definition}[Bounded localized vertex operators]
If $V$ has bounded insertions for $\scA$, we define the local algebras corresponding to intervals $I \subset S^1$ by
$$
\cA_V(I) = \{ BY(a,z)A : (B,A) \in \scA_I, z \in \Int(B,A), a \in V \}^{\prime\prime} \vee \{BA : (B,A) \in \scA_I\}^{\prime\prime}.
$$
We say that $V$ has bounded localized vertex operators if $\cA_V(I)$ and $\cA_V(J)$ commute whenever $I$ and $J$ are disjoint.
\end{Definition}

It is not hard to show that if $V$ has bounded localized vertex operators then $\cA_V$ is a conformal net \cite[Prop. 4.8]{GRACFT1} \cite[Prop. 4.15]{TenerGRACFT2}.

\subsubsection{Representation theory and (local) intertwiners}

If $M$ is a unitary $V$-module, then it decomposes as a direct sum of irreducible representations of the Virasoro algebra, and thus $\cH_M$ carries a representation of the Virasoro net $\cA_c$.
If $(B,A) \in \scA_I$, then $B$ and $A$ are simultaneously localizable in $I$.
As described above, the actions of $A$ and $B$ on $\cH_M$ are defined using a lift $\tilde \gamma \in \Diff_c^{(\infty)}(S^1)$ of a localizing diffeomorphism $\gamma$, and they are in general only well-defined up to a a multiple of $e^{2\pi i L_0}$.
The sesquilinear form $BY^M(a,z)A$, however, is canonically defined by requiring that $B$ and $A$ act using the same lift $\tilde \gamma$.
Provided $A \in \Ann^{in}(\cH_M)$ and $B \in \Ann^{out}(\cH_M)$ (see Definition \ref{def: incoming and outgoing annuli}), the sesquilinear form $BY^M(a,z)A$ is densely defined.

\begin{Definition}\label{def: piM}
Let $V$ be a simple unitary VOA with bounded localized vertex operators, and let $M$ be a unitary $V$-module.
Suppose that for every interval $I$ and every $(B,A) \in \scA_I$ we have $A \in \Ann^{in}(\cH_M)$, $B \in \Ann^{out}(\cH_M)$.
Then the $\cA_V$-representation corresponding to $M$, if it exists, is the representation $\pi^M$ on $\cH_M$ such that for every $I$ and $(B,A)$ as above, and every $z \in \Int(B,A)$ and $a \in V$ we have
\begin{equation}\label{eqn: piM formulas}
\pi^M_I(BY(a,z)A) = BY^M(a,z)A.
\end{equation}
\end{Definition}

In particular this implies that the forms $BY^M(a,z)A$ are bounded.
In fact, since $\cH_M$ is separable as a part of our definition of unitary $V$-module, there exists a unitary operator $u:\cH_V \to \cH_M$ that implements $\pi^M_I$ (i.e. $\pi^M_I(x) = uxu^*$ for all $x \in \cA_V(I)$).
If $z \in \Int(B,A)$ then for some $s > 0$ the map $BY(s^{L_0}-,z)A$ is bounded, and thus so is $BY^M(s^{L_0}-,z)A = u BY(s^{L_0}-,z)A u^*$.

Our next step is to consider operators of the form $B \cY(a,z)A$ for $(B,A) \in \scA_I$ and $\cY \in \binom{K}{M N}$.
We need to assume that $K$ and $N$ are unitary so that $A$ and $B$ act on $\cH_N$ and $\cH_K$, respectively.
Moreover, in order for $B \cY(a,z)A$ to be densely defined we need to assume that $B \in \Ann^{out}(\cH_K)$ and $\Ann^{in}(\cH_N)$.
In practice, it is easiest to just assume that $\pi^N$ and $\pi^K$ exist.
We also generally assume that $M$ is a unitary module, although strictly speaking that is not necessary.

Unlike the case of $BY^M(a,z)A$, there is in general a genuine ambiguity in the definition of $B\cY(a,z)A$.
While we can simultaneously localize $B$ and $A$, the $e^{2 \pi i L_0}$ ambiguities in their action on $\cH_K$ and $\cH_N$ do not necessarily cancel.
However, we will often be interested in the special case where $(1,A) \in \scA_I$, in which case $A \in \cA_c(I)$.
We then define the action of $A$ on $\cH_N$ in the natural way using $\operatorname{id} \in \Diff_c^{(\infty)}(S^1)$ as the localizing unitary.

\begin{Definition}\label{def: localized intertwiner}
Let $V$ be a simple unitary vertex operator algebra, let $M$, $N$, and $K$ be unitary $V$-modules, and assume that $\pi^N$ and $\pi^K$ exist.
An intertwining operator $\cY \in I \binom{K}{M N}$ is said to have \emph{bounded insertions} if for every interval $I$, every $(B,A) \in \scA_I$, and every $z \in \Int(B,A)$ there is some $s > 0$ such that the operator $B \cY(s^{L_0}-,z) A$ defines a bounded bilinear map $\cH_M \times \cH_N \to \cH_K$.
If moreover we always have
$$
B \cY(a,z) A \in \Hom_{\cA_V(I^\prime)}(\cH_N, \cH_K)
$$
for all $a \in M$ then we say that $\cY$ is \emph{localized}.
We write $I_{loc} \binom{K}{M N}$ for the space of localized intertwining operators.
\end{Definition}

Note that even though the actions of $A$ and $B$ are generally only defined up to a multiple of $e^{2 \pi i L_0}$, the boundedness and locality of the operators is independent of those choices.

\begin{Remark}
In the definition of bounded localized vertex operators we required that $BY(s^{L_0}-,z)A$ define a bounded map from the Hilbert space tensor product $\cH_V \otimes \cH_V$ into $\cH_V$ for some $s > 0$.
In contrast, in Definition \ref{def: localized intertwiner} we have required that $B\cY(s^{L_0}-,z)A$ define a bounded bilinear map $\cH_M \times \cH_N \to \cH_K$, or equivalently that it define a bounded map on the projective tensor product $\cH_M \otimes_\pi \cH_N$, which is a weaker condition.
The weaker condition is more appropriate for intertwining operators (as e.g. the proof of Proposition \ref{prop: dagger localized} would run into problems).
In fact, it is possible that the weaker condition would be more appropriate for bounded localized vertex operators as well, and all of the results from \cite{GRACFT1,TenerGRACFT2} would go through verbatim under the weaker hypothesis.
However we have not made this change to the definition of bounded localized vertex operators here so as not to muddle the relationship between this article and the preceding ones.
We also point out that if $B Y(s^{L_0}-, z)A$ defines a bounded bilinear map and $r^{L_0}$ is trace class on $\cH_V$, then $B Y((rs)^{L_0}-,z)A$ defines a bounded map on the Hilbert space tensor product.
Thus so long as $r^{L_0}$ is trace class for some $r$ the weaker and stronger definitions are actually equivalent.
\end{Remark}

By the Banach-Steinhaus theorem, $B \cY(s^{L_0}-,z)A$ defining a bounded bilinear map is equivalent to saying that for every $a \in M$ the sesquilinear form $B \cY(a,z)A$ defines a bounded map, and the map $a \mapsto B \cY(s^{L_0}a,z)A$ is bounded $M \to \cB(\cH_N,\cH_K)$.

We in saw in \cite[Lem. 4.8]{TenerGRACFT2} that if $B \cY(s^{L_0}-,z)A$ was bounded then for any $r < s$ the maps $B \cY(r^{L_0}-,w)A$ are uniformly bounded in norm when $\abs{w-z}$ is sufficiently small.
In particular this implies by \cite[Lem. 4.1]{TenerGRACFT2} that for any $a \in M$ the map $B\cY(a,z)A$ is holomorphic from $\Int(B,A)$ to $\cB(\cH_N,\cH_K)$.

We will establish one last property of localized intertwining operators.

\begin{Proposition}\label{prop: dagger localized}
Let $V$ be a simple unitary VOA with bounded localized vertex operators, let $M$, $N$, and $K$ be unitary $V$-modules, and assume that $\pi^K$ and $\pi^N$ exist.
Suppose that $\cY \in I_{loc} \binom{K}{M N}$.
Then $\cY^\dagger \in I_{loc}\binom{N}{\overline{M} K}$.
\end{Proposition}
\begin{proof}
Let $(B,A) \in \scA_I$, and let $z \in \Int(B,A)$.
For $a \in M$, the map $B\cY^\dagger(a,z)A$ is bounded, as it is the adjoint of $A^*\cY(\theta_{\overline{z}}a,\overline{z}^{-1})B^*$ and $\overline{z}^{-1} \in \Int(A^*,B^*)$.
Moreover it follows that $B\cY^\dagger(a,z)A \in \Hom_{\cA(I)^\prime}(\cH_K,\cH_N)$.
Thus all we need to show is that $a \mapsto B\cY^\dagger(s^{L_0}a,z)A$ defines a bounded map $\overline{M} \to \cB(\cH_K,\cH_N)$.

Since $\cY \in I_{loc}$, for some $s > 0$ the bilinear map $A^*\cY(s^{L_0}-,\overline{z}^{-1})B^*$ is bounded.
That is, there is some constant $C$ such that $\norm{A^*\cY(s^{L_0}a,\overline{z}^{-1})B^*} \le C \norm{a}$ for $a \in M$.
Hence taking adjoints we have for any $r > 0$
\begin{align*}
\|B \cY^\dagger((rs)^{L_0}a,z)A\|
&=
\norm{A^* \cY(\theta_{\overline{z}}(rs)^{L_0}a,\overline{z}^{-1})B^*}\\
&=
\norm{A^* \cY(s^{L_0} \theta_{s\overline{z}}r^{L_0}a,\overline{z}^{-1})B^*}\\
&\le C \norm{\theta_{s\overline{z}}r^{L_0}a}.
\end{align*}
Hence if we can find $r$ small enough that $\theta_{s \overline{z}} r^{L_0}$ is a bounded operator then 
$$
a \mapsto B \cY^\dagger((rs)^{L_0}a,z)A
$$
will be bounded as well.

Recalling that $\theta_w a = e^{w L_1} e^{i \pi L_0} w^{-2L_0} \overline{a}$, it suffices to show that for any $w$ there is a $r$ such that $e^{w L_1} r^{L_0}$ is bounded.
Equivalently, we must show that $r^{L_0}e^{rwL_1}$ is bounded.
We saw in \cite[Lem. 4.6]{TenerGRACFT2} that when $r + \abs{rw} < 1$ the operator $e^{wrL_{-1}}r^{L_0}$ is bounded, and so taking adjoints we see that $r^{L_0} e^{rwL_1}$ is bounded as well.
Thus $B\cY^\dagger((rs)^{L_0}-,z)A$ defines a bounded bilinear form for $r$ sufficiently small.
\end{proof}

%% file: Sec_3_Positivity.tex
\newpage
\section{Transport modules and positivity for unitary VOAs}\label{sec: positivity}

In this section we introduce the \emph{transport module} associated to a pair of unitary modules $M$ and $N$ for some unitary VOA $V$.
We will see in Section \ref{sec: transport for rational} that if $V$ is sufficiently nice (e.g. regular) VOA and a transport module exists, then it is isomorphic to the fusion product $M \boxtimes N$ and moreover gives a canonical inner product on $M \boxtimes N$.
Closely related ideas have been used by Gui \cite{GuiUnitarityI,GuiUnitarityII,GuiQSystem} to construct unitary modular tensor categories of $V$-modules for many examples.
The results of this section and Section \ref{sec: examples} will provide many examples of transport modules, and thereby more examples where Gui's machinery can be applied.

Of additional interest is the fact that transport modules can be defined for an arbitrary unitary VOA $V$, and we conjecture that they exist for any unitary modules $M$ and $N$ which are reasonably nice.
This provides a proposal for a unitary fusion product theory for an arbitrary unitary VOA, just as any conformal net has a tensor category of representations.
In fact, the transport form (Definition \ref{def: transport form}) is a translation of the Connes' fusion inner product from conformal nets into the language of VOAs, which we will use in Section \ref{sec: fusion rules for bimodules} to compare VOA and conformal net fusion products.
A version of this translation was used by Wassermann in his work on WZW models of type A \cite{Wa98}, and both our approach and the approach used by Gui were inspired by Wassermann's work.

\subsection{The transport forms and transport modules}\label{sec: transport forms}

Let $V$ be a simple unitary VOA and let $M$ and $N$ be unitary $V$-modules.
If $K$ is a unitary $V$-module equipped with an intertwining operator $\cY \in I_{Hilb} \vertex{K}{M}{N}$, then we can define a family of pre-inner products on $M \otimes_\bbC N$ by 
\begin{equation}\label{eqn: pre inner product cY}
[a_1 \otimes b_1, a_2 \otimes b_2] := \ip{\cY(a_1,z)b_1, \cY(a_2,z)b_2} = \ip{\cY^\dagger(\theta_{\overline{z}}a_2,\overline{z}^{-1})\cY(a_1,z)b_1,b_2},
\end{equation}
for $z$ in the punctured unit disk (equipped with a choice of $\log z$), where $\theta_{\overline{z}}$ is as in Section \ref{sec: unitary VOA}.

We follow the philosophy that the unitary $V$-module $K$, its Hilbert space completion $\cH_K$, and the intertwining operator $\cY$ can be recovered from the above pre-inner product $[\, \cdot \, , \, \cdot \,]$, which we now explain.
Under favorable circumstances (see Lemma \ref{lem: intertwiner image is a core}) vectors of the form $\cY(a,z)b$ span a dense subspace of $\cH_K$, and so we can identify $\cH_K$ with the Hilbert space completion (and quotient) $\widehat{M \otimes N}$ of $M \otimes N$ with respect to the pre-inner product above for some choice of $z$ (and $\log z$).
In other words, we can recover the Hilbert space $\cH_K$, along with the dense subspace spanned by vectors of the form $\cY(a,z)b$, from the pre-inner product $[ \, \cdot \, ,\, \cdot \, ]$.

In fact, we can also recover the $V$-module $K$ from the pre-inner product (under the same favorable circumstances as in Lemma \ref{lem: intertwiner image is a core}).
Under the identification $\cH_K \cong \widehat{M \otimes N}$, the canonical map $M \otimes N \to \widehat{M \otimes N}$ coincides with the map $a \otimes b \mapsto \cY(a,z)b$.
By the Borcherds commutator formula we have
\[
v_{(n)}^K \cY(a,z)b = \cY(a,z)v_{(n)}^N b + \sum_{j=0}^\infty \binom{n}{j} z^{n-j} \cY(v_{(j)}^M a, z)b 
\]
where we have temporarily suppressed technical considerations with regard to the domain of the operator $v_{(n)}^K$.
Thus the unique candidate for a $V$-module action on the image of $M \otimes N$ in $\widehat{M \otimes N}$ that is compatible with the identification $\cH_K \cong \widehat{M \otimes N}$ is
\begin{equation}\label{eqn: canonical action on M otimes N}
v^{M \otimes N}_{(n)}(a \otimes b) = a \otimes v^{N}_{(n)}b + \sum_{j=0}^\infty \binom{n}{j}z^{n-j} v_{(j)}^Ma \otimes b.
\end{equation}
We regard the maps $v^{M \otimes N}_{(n)}$ as unbounded operators on $\widehat{M \otimes N}$ (see Section \ref{sec: unbounded operators}), and take their closure.
Under the hypotheses of Lemma \ref{lem: intertwiner image is a core}, 
there is a canonical subspace of $\widehat{M \otimes N}$ such that restricting the modes $v^{M \otimes N}_{(n)}$ to this subspace yields a unitary $V$-module that is isomorphic to $K$ via the identification $\cH_K \cong \widehat{M \otimes N}$ previously established.
Thus we have recovered the module $K$ and the intertwining operator $\cY$ from only the pre-inner product \eqref{eqn: pre inner product cY} and the canonical action \eqref{eqn: canonical action on M otimes N}\footnote{More precisely, we have recovered the $P(z)$-intertwining map (in the sense of \cite{HuangLepowsky95,HuangLepowskyPartIII}) associated to $\cY$ and the point $z$ in the universal cover of the punctured disk that was used to define \eqref{eqn: pre inner product cY}}. 

%In fact, the $V$-module structure on $K$ can also be recovered from the fact that the map $M \otimes N \to \cH_K$ is supposed to come from an intertwining operator (more precisely, the map should be a $P(z)$-intertwining map in the sense of \cite{HuangLepowsky95,HuangLepowskyPartIII}).
%%In such circumstances, the $V$-module structure on $K$ can be recovered from the map $M \otimes N \to \cH_K$ given by $a \otimes b \mapsto \cY(a,z)b$ (which is a $P(z)$-intertwining map in the sense of \cite{HuangLepowsky95,HuangLepowskyPartIII}) via the Borcherds commutator formula.

Given that the pre-inner product \eqref{eqn: pre inner product cY} encodes the unitary module $K$ and the intertwining operator $\cY$, we will now examine the formula on the right-hand side of \eqref{eqn: pre inner product cY} in more detail.
In the context of vertex tensor categories, we may hope to write the product 
$$
\cY^\dagger(\theta_{\overline{z}}a_2,\overline{z}^{-1})\cY(a_1,z)
$$ 
as an iterate 
\begin{equation}\label{eqn: product to iterate}
\cY^\dagger(\theta_{\overline{z}}a_2,\overline{z}^{-1})\cY(a_1,z) = \cY_1(\cY_2(\theta_{\overline{z}}a_2,\overline{z}^{-1}-z)a_1,z)
\end{equation}
for certain intertwining operators $\cY_1$ and $\cY_2$.

This sort of computation makes sense for any unitary $V$-module $K$ and any intertwining operator $\cY$, but a unitary fusion product should be a pair $(K, \cY)$ such that we have \eqref{eqn: product to iterate} with $\cY_1 = Y^N$ and $\cY_2 = \cY_M^{-}$ (where $\cY_M^- \in I \binom{V}{\overline{M} M}$ is the annihilation operator - see Section \ref{sec: intertwining operators}).
Thus the pre-inner product \eqref{eqn: pre inner product cY} which characterizes both the inner product and $V$-module structure on a unitary fusion product of $M$ and $N$ can be explicitly written in terms of the data of these modules.
The main idea of this section is to introduce the notion of \emph{transport module}, which is a formalization of the notion of unitary fusion product module characterized by an intrinsic pre-inner product.

Now given unitary $V$-modules $M$ and $N$, we define a formal series-valued sesquilinear pairing $\ip{ \, \cdot \, , \, \cdot \,}_{x,y}$ on $M \otimes N$ by
\begin{equation}\label{eqn: formal form}
\ip{a_1 \otimes b_1, a_2 \otimes b_2}_{x,y} = \bip{Y^N(\cY_M^-(\theta_{\overline{y}} a_2, \overline{y}^{-1}-x)a_1,x)b_1,b_2}_N
\end{equation}
where $\theta_{\overline{y}}a_2 = e^{\overline{y} L_1} e^{i \pi L_0} {\overline{y}}^{-2L_0}\overline{a_2}$.
Here $a \mapsto \overline{a}$ is the canonical antilinear map $M \to \overline{M}$.

We wish to evaluate the above pairing at complex numbers $z,w \in \bbC \setminus \{0\}$, and indeed one should expect that the corresponding series converges to a multi-valued holomorphic function on $\abs{z} > \abs{\overline{w}^{-1}-z} > 0$.
However, for our purposes it will suffice to consider a more restricted domain.
Let
\begin{equation}\label{eqn: R def}
\cR = \{(z,w) \in \bbC^2 : 1 > \abs{z} > \abs{\overline{w}^{-1} - z} > 0 \mbox{ and } 1 > \abs{w} > 0\}.
\end{equation}
A point $(z,w) \in \cR$ is one for which it is reasonable to hope that the series defining $\cY_2(a_2,\overline{w}^{-1})\cY_1(a_1,z)$ and $\cY_1(\cY_2(a_2,\overline{w}^{-1}-z)a_1,z)$ converge (when $\cY_i$ are intertwining operators with appropriate target and source modules).
As a convenient bonus, we will see now that $\ip{a_1 \otimes b_1, a_2 \otimes b_2}_{z,w}$ is naturally single-valued on $\cR$, assuming that the relevant series converges.

Let us first consider when $M$ is a simple module.
In this case, the conformal weights of $M$ lie in $\Delta + \bbZ$ for some $\Delta \in \bbR$.
In this case, the powers of $(\overline{y}^{-1}-x)$ which arise in \eqref{eqn: formal form} lie in $-2\Delta + \bbZ$.
The second source of multi-valuedness in \eqref{eqn: formal form} is the term $\theta_{\overline{y}} a_2$, which introduces a factor of the form $\overline{y}^{-2\Delta+m}$ for an integer $m$.
Thus if we expand the series \eqref{eqn: formal form} we obtain 
$$
\ip{a_1 \otimes b_1, a_2 \otimes b_2}_{x,y} 
= \overline{y}^{-2\Delta}(\overline{y}^{-1}-x)^{-2\Delta}f(x,\overline{y})
=(1-\overline{y}x)^{-2\Delta}f(x,\overline{y}).
$$
where $f(x,\overline{y})$ is a formal series involving only integral powers of $x$ and $\overline{y}$.
When $\abs{z},\abs{w} < 1$, we evaluate $(1-\overline{w}z)^{-2\Delta}$ using the standard branch of $\log(1-\overline{w}z)$.
Thus the series defining $\ip{a_1 \otimes b_1, a_2 \otimes b_2}_{z,w}$ is well-defined and single-valued for $(z,w) \in \cR$.

So far, we have only considered simple modules $M$.
In general a unitary $V$-module $M$ may be written $M = \bigoplus M_i$ and the annihilation operator for $M$ decomposes similarly.
Thus the transport form $\ip{a_1 \otimes b_1, a_2 \otimes b_2}_{z,w}$ is naturally single-valued in general.

\begin{Definition}\label{def: transport form}
Let $V$ be a simple unitary VOA and let $M$ and $N$ be unitary $V$-modules.
We say that the transport form for $M$ and $N$ exists if the series 
\begin{equation}\label{eqn: transport form}
\ip{a_1 \otimes b_1, a_2 \otimes b_2}_{z,w} = 
\sum_{s \in \bbR} \bip{Y^N(P_s \cY_M^-(\theta_{\overline{w}} a_2, \overline{w}^{-1}-z)a_1,z)b_1,b_2}_N
\end{equation}
converges absolutely for every
$$
(z,w) \in \cR = \{(z,w) \in \bbC^2 : 1 > \abs{z} > \abs{\overline{w}^{-1} - z} > 0 \mbox{ and } 1 > \abs{w} > 0\},
$$
where $P_s$ is the projection onto vectors of conformal dimension $s$.
In this case we write $\bip{Y^N(\cY_M^-(\theta_{\overline{w}} a_2, \overline{w}^{-1}-z)a_1,z)b_1,b_2}_N$ instead of the sum \eqref{eqn: transport form}.

As explained in the previous paragraph, we regard $\ip{a_1 \otimes b_1, a_2 \otimes b_2}_{z,w}$ as a single-valued function on $\cR$ using the standard branch of $\log(1-\overline{w}z)$.
This function is holomorphic in $z$ and antiholomorphic in $w$.
\end{Definition}

In addition to the transport form, we will be interested in the sesquilinear form 
\begin{equation}\label{eqn: pairing of vectors}
\ip{\cY(a_1,z)b_1, \cY(a_2,w)b_2}
\end{equation}
where $\cY \in I_{Hilb}\binom{K}{M N}$ and $(z,w) \in \cR$.
Recall (Section \ref{sec: intertwining operators}) that the condition $\cY \in I_{Hilb}$ means that $\cY(a_i,z)b_i \in \cH_K$, so the inner product in \eqref{eqn: pairing of vectors} makes sense.
While a priori \eqref{eqn: pairing of vectors} depends on choices of $\log z$ and $\log w$, there is a natural single-valued meaning to this expression on $\cR$.
Indeed, if $(z,w) \in \cR$ then $\abs{1 - \overline{w}z} < 1$, which implies that $z$ and $w$ lie in a common half-plane cut out by a line through the origin of $\bbC$.
We evaluate \eqref{eqn: pairing of vectors} using values of $\log z$ and $\log w$ coming from a common branch of $\log$ defined on this half-plane.
The value of \eqref{eqn: pairing of vectors} is independent of the choice of branch, and of the choice of half-plane, yielding a single-valued meaning \eqref{eqn: pairing of vectors}.
Indeed, this is the unique branch of \eqref{eqn: pairing of vectors} on $\cR$ with the property that
$$
\ip{\cY(a,z)b, \cY(a,z)b} = \norm{\cY(a,z)b}^2
$$
whenever $(z,z) \in \cR$.

\begin{Remark}\label{rmk: U times U log}
Alternatively, note that if $\cU \subset \bbC^\times$ is an open set such that $\cU \times \cU \subset \cR$, and $z \in \cU$, then $\cU$ cannot contain any positive multiple of $-z$.
Thus there exists a holomorphic branch of $\log$ on $\cU$, and we define \eqref{eqn: pairing of vectors} on $\cU \times \cU$ using any fixed branch of $\log$ on both factors.
The value of \eqref{eqn: pairing of vectors} for other $(z,w) \in \cR$ can be obtained by analytic continuation.
\end{Remark}

\begin{Definition}\label{def: transport module}
Let $V$ be a simple unitary VOA and let $M$ and $N$ be unitary $V$-modules. 
Suppose that the transport form exists for $M$ and $N$.
Then a transport module for $M$ and $N$ is a pair $(M \boxtimes_t N, \cY_t)$ where $M \boxtimes_t N$ is a unitary $V$-module and $\cY_t \in I_{Hilb} \vertex{M \boxtimes_t N}{M}{N}$ is a dominant intertwining operator such that
\begin{equation}\label{eqn: transport module}
\ip{a_1 \otimes b_1, a_2 \otimes b_2}_{z,w} = \ip{\cY_t(a_1,z)b_1, \cY_t(a_2,w)b_2}_{M \boxtimes_t N}
\end{equation}
as single-valued functions on 
$$
\cR= \{(z,w) \in \bbC^2 : 1 > \abs{z} > \abs{\overline{w}^{-1} - z} > 0 \mbox{ and } 1 > \abs{w} > 0\}.
$$
\end{Definition}
Note that the left-hand and right-hand sides of \eqref{eqn: transport module} are interpreted as single-valued functions on $\cR$ as described in the paragraphs preceding Definitions \ref{def: transport form} and \eqref{def: transport module}, respectively.

\begin{Remark}
Recall that a dominant intertwining operator $\cY$ is one for which terms $a^{\cY}_{(k)}b$ span the target module.
The requirement of Definition \ref{def: transport module} that $\cY_t$ be dominant should be thought of as a non-degeneracy requirement, and if $\cY_t$ satisfies \eqref{eqn: transport module} but is not dominant then the submodule generated by the modes of $\cY_t$ is a transport module.
\end{Remark}

\begin{Remark}\label{rem: VOA flavor transport def}
Modulo subtleties about multi-valued functions and domains, the transport module condition is equivalent to
$$
\bip{Y^N(\cY_M^-(a_2,w-z)a_1,z)b_1,b_2} = \bip{\cY_t^\dagger(a_2,w)\cY_t(a_1,z)b_1,b_2}
$$
for appropriate value of $w$ and $z$.
Thus it is a particular example of what is called the equivalence of products and iterates in the VOA literature.
The key feature is that we may write the iterate of $Y^N$ and $\cY_M^-$ as a product of an intertwining operator and its adjoint.
\end{Remark}

We will see below (Proposition \ref{prop: transport module unique}) that a transport module $(M \boxtimes_t N, \cY_t)$, if it exists, is unique up to unique unitary isomorphism.
In particular $M \boxtimes_t N$ is unique up to unitary isomorphism, and we will speak of \emph{the} transport module.

It is clear that if $K$ and $\tilde K$ are unitary $V$-modules satisfying Definition \ref{def: transport module}, with corresponding intertwiners $\cY_t$ and $\tilde \cY_t$, then there is an isometry defined on the subspace of $\cH_K$ spanned by vectors of the form $\cY_t(a,z)b$ by $\cY_t(a,z)b \mapsto \tilde \cY_t(a,z)b$.
We must check that the domain and range of this map are dense, and that it gives a homomorphism of $V$-modules.

Recall (Section \ref{sec: unbounded operators}) that the modes $v_{(k)}^K$ and $v_{(k)}^{\tilde K}$ are closable operators densely defined on $K \subset \cH_K$ and $\tilde K \subset \cH_{\tilde K}$, respectively, so it suffices to check that vectors of the form $\cY_t(a,z)b$ are a core for the closure of  $v_{(k)}^K$.
This fact, which we establish just below in Lemma \ref{lem: intertwiner image is a core}, is both philosophically and technically important. 
It justifies our approach to understand fusion products by studying the image of intertwining operators, and it is crucial for rigorously applying functional analytic tools to VOAs.
%For example, in evaluating expressions such as $v_{(k)}^K \cY_t(a,z)$, we have been considering $v_{(k)}^K$ as an operator on the algebraic completion of $K$, and we will need to be careful to see that this expression can be evaluated as an unbounded operator acting on a vector in its domain.

\begin{Lemma}\label{lem: intertwiner image is a core}
Let $V$ be a simple unitary VOA, let $M$, $N$, and $K$ be unitary $V$-modules, and let $\cY \in I_{Hilb} \vertex{K}{M}{N}$ be a dominant intertwining operator.
Let $z \in \bbC$ with $1 > \abs{z} > 0$, and fix a choice of $\log z$.
Set 
$$
K_z = \Span \{ \cY(a,z)b : a \in M, \, b \in N \} \subset \cH_K.
$$
Then $K_z$ is dense in $\cH_K$.
Moreover $K_z$ is an invariant core for the closure of $v_{(n)}^K$, for every $v \in V$ and $n \in \bbZ$.
\end{Lemma}
\begin{proof}
We begin by pointing out that $K_z$ depends on the choice of $\log z$.
For $v \in V$ and $a \in M$ we will write $v_{(n)}$ for $v^K_{(n)}$ and $a_{(k)}$ for $a_{(k)}^\cY$.
Let $D$ be the domain of the closure of $v_{(n)}$.
We show first that $K_z \subset D$.
It suffices to show that $\cY(a,z)b \in D$ when $a$ and $b$ are homogeneous.
Since $\cY(a,z)b \in \cH_K$ the partial sums of the series $\sum_{k \in \bbR} a_{(k)}bz^{-k-1}$, in which all terms are orthogonal, provide a sequence in $K$ converging to $\cY(a,z)b$.
Now by the Borcherds commutator formula,
%\begin{align*}
%\sum_{k \in \bbR} v_{(n)} a_{(k)} b  z^{-k-1} = (-1)^{p(a)p(b)} \sum_{k \in \bbR} a_{(k)} v_{(n)} b + \sum_{k \in \bbR} \sum_{j=0}^A \binom{n}{j}\big(v_{(j)} a\big)_{(n+k-j)} z^{-k-1}
%\end{align*}
\begin{align*}
v_{(n)} a_{(k)}b  =  a_{(k)} v_{(n)}b +  \sum_{j=0}^A \binom{n}{j}\big(v_{(j)} a\big)_{(n+k-j)}b
\end{align*}
for some number $A$.
Hence
\begin{align*}
\sum_{k \in \bbR} v_{(n)}a_{(k)}b z^{-k-1} &= \sum_{k \in \bbR} a_{(k)} v_{(n)}b z^{-k-1} \,+\\
&\quad+\, \sum_{j=0}^A \binom{n}{j} z^{n-j} \sum_{k \in \bbR} \big(v_{(j)}a\big)_{(n+k-j)}b z^{-n-k+j-1}.
\end{align*}
Each sum indexed by $\bbR$ on the right-hand side is of the homogeneous components of a vector in $\cH_K$, and thus converges.
Hence the partial sums of $\sum_{k \in \bbR} a_{(k)} b z^{-k-1}$ provide a witness for the vector $\cY(a,z)b$ lying in $D$, and moreover 
\begin{equation}\label{eqn: vn yazb}
v_{(n)}\cY(a,z)b = \cY(a,z)v_{(n)}b + \sum_{j=0}^A \binom{n}{j} z^{n-j} \cY(v_{(j)}a,z)b
\end{equation}
as one expects.
In particular we see that $K_z$ is invariant under (the closure of) $v_{(n)}$.

As a next step, we will prove that $K_z$ is dense in $\cH_K$.
Recall from Lemma \ref{lem: Hilb intertwiners give holomorphic functions} that functions of the form $\cY(a,z)b$ are holomorphic on the punctured disk.
It follows from the definition of an intertwining operator that $\frac{d}{dz} \ip{\cY(a,z)b,d} = \ip{\cY(L_{-1}a,z)b,d}$ for every $d \in K$, and thus $\frac{d}{dz} \cY(a,z)b = \cY(L_{-1}a,z)b$.
Hence if $0 < \abs{z-w} < \min(\abs{z}, 1-\abs{z})$ we have
\begin{equation}\label{eqn: intertwiner taylor series}
\cY(a,w)b = \sum_{m=0}^\infty \frac{(w-z)^m}{m!} \cY(L_{-1}^m a,z)b \in \overline{K_z}.
\end{equation}
Thus $K_w \subset \overline{K_z}$ for $\abs{z-w}$ sufficiently small, where the choice of $\log w$ is obtained by analytic continuation.
It follows that $e^{i \theta L_0}$ maps $K_z$ into $\overline{K_z}$ when $\theta$ is sufficiently small, and therefore $\overline{K_z}$ is invariant under $e^{i \theta L_0}$ for all $\theta$.
Since $L_0$ is diagonalizable on $\cH_K$, the closed subspace $\overline{K_z}$ decomposes as a direct sum of $L_0$-eigenspaces, and it follows that $a_{(k)}b \in \overline{K_z}$ for all $a \in M$ and $b \in N$.
Since $\cY$ was assumed dominant, this means that $K \subset \overline{K_z}$ and thus $K_z$ is dense.

Since $v_{(k)}$ is closed on $D$ and $K_z \subset D$ is a dense subspace, it follows that $v_{(k)}|_{K_z}$ is closable, as $(v_{(k)}|_D)^* \subset (v_{(k)}|_{K_z})^*$.
Let $D_z$ be the domain of the closure of $v_{(k)}|_{K_z}$.
We have $D_z \subseteq D$, and must prove the reverse inclusion.

We begin by showing that $D_z = D_w$ for any pair $z$ and $w$ in the punctured disk, and any choices of $\log z$ and $\log w$.
By \eqref{eqn: vn yazb}, the function $v_{(n)}\cY(a,z)b$ is holomorphic, and by standard VOA calculations $\frac{d}{dz} \ip{v_{(n)}\cY(a,z)b, d} = \ip{v_{(n)}\cY(L_{-1}a,z)b,d}$ for all $d \in K$.
It follows that $\frac{d}{dz} v_{(n)}\cY(a,z)b = v_{(n)}\cY(L_{-1}a,z)b$, and so when $\abs{z-w}$ is sufficiently small we may Taylor expand as before
$$
v_{(n)} \cY(a,w)b = \sum_{m=0}^\infty \frac{(w-z)^m}{m!} v_{(n)}\cY(L_{-1}^m a,z)b.
$$
Thus the partial sums of $\sum_{m=0}^\infty \frac{(w-z)^m}{m!} \cY(L_{-1}^m a,z)b$ witness the fact that $\cY(a,w)b \in D_z$.
We have shown that $K_w \subset D_z$, and it follows that $D_w \subset D_z$.
But $w$ and $z$ were interchangeable (provided $\abs{z-w}$ is sufficiently small), and so locally we have $D_z = D_w$, where the choice of $\log w$ is obtained by analytically continuing that of $z$.
It follows that $D_z$ does not depend on $z$ or the choice of $\log z$, and so we denote this single domain by $\tilde D$.

Let $\Gamma_z \subset \cH_K \oplus \cH_K$ be the graph of $v_{(n)}|_{K_z}$, and let $\Gamma_z^\prime$ be the corresponding graph when $\log z$ is replaced by $\log z + 2 \pi i$.
%$$
%\Gamma_z = \Span \{ (\cY(a,z)b, v_{(n)}\cY(a,e^{\log z})b) : a \in M, b \in N \}
%$$
%be the graph of $v_{(n)}|_{K_z}$, and let
%$$
%\Gamma_z^\prime = \Span \{ (\cY(a,z)b, v_{(n)}\cY(a,e^{\log z+2\pi i})b) : a \in M, b \in N \}
%$$
%be the corresponding graph when $\log z$ is replaced by $\log z + 2 \pi i$.
We have shown that $\overline{\Gamma_z} = \overline{\Gamma^\prime_z}=: \tilde \Gamma$ is the graph of $v_{(n)}|_{\tilde D}$.

Let $\tau = e^{2 \pi i L_0}$ be the twist on $\cH_K$.
Then $(\tau \oplus \tau)\Gamma_z = \Gamma_z^\prime$, so $(\tau \oplus \tau)\tilde \Gamma = \tilde \Gamma$.
It follows that the projection onto $\tilde \Gamma$ commutes with $\tau \oplus \tau$, and therefore commutes with its spectral projections.
For $h \in \bbR$ let $p_h$ to be the projection onto the closed subspace of $\cH_K$ spanned by homogeneous vectors with conformal dimension equal to $h$ mod $1$.
Then $p_h$ is the projection onto the $e^{2 \pi i h}$-eigenspace of $\tau$, and thus $p_h \oplus p_h$ leaves $\tilde \Gamma$ invariant.

Now for $k \in \bbR$ and homogeneous $a \in M$ and $b \in N$, let $h$ be the conformal dimension of $a_{(k)}b$.
Then the function
$$
f(z) = (p_h z^k \cY(a,z)b, p_h z^k v_{(n)} \cY(a,z)b)
$$
is a single-valued holomorphic function on the punctured disk with values in $\tilde \Gamma$.
Thus $(a_{(k)}b, v_{(n)} a_{(k)} b) = \frac{1}{2 \pi i} \int_{\abs{z} = \frac12} f(z) \, dz \in \tilde \Gamma$.
In particular $a_{(k)}b \in \tilde D$.
Since $\cY$ is dominant, we have $K \subset \tilde D$.
Hence $D \subset \tilde D$, as $D$ is the domain of the closure of $v_{(n)}|_K$.
We conclude that $D = \tilde D$ and thus $K_z$ is a core for $v_{(n)}$.
\end{proof}

We now prove the uniqueness of transport modules up to unitary isomorphism.

\begin{Proposition}\label{prop: transport module unique}
Let $V$ be a simple unitary VOA, let $M$ and $N$ be unitary $V$-modules, and let $(K,\cY_t)$ and $(\tilde K,\tilde \cY_t)$ be transport modules for $M$ and $N$.
Then there is a unique unitary $u:\cH_K \to \cH_{\tilde K}$ such that  $u \cY_t  = \tilde \cY_t$, and $u$ restricts to an isomorphism of $V$-modules $K \cong \tilde K$.
\end{Proposition}
\begin{proof}
Fix a point $z$ such that $(z,z) \in \cR$ (i.e. such that $1 > \abs{z} > 2^{-1/2}$), along with a choice of $\log z$.
The condition that $(K, \cY_t)$ and $(\tilde K, \tilde \cY_t)$ are both transport modules for $M$ and $N$ means that
\begin{equation}\label{eqn: same form}
\bip{\cY_t(a_1,z)b_1,\cY_t(a_2,z)b_2}_{\cH_K} 
= \bip{a_1 \otimes b_1, a_2 \otimes b_2}_{z,z}
= \bip{\tilde \cY_t(a_1,z)b_1,\tilde \cY_t(a_2,z)b_2}_{\cH_{\tilde K}}
\end{equation}
for all $a_i \in M$ and $b_i \in N$.
If we set $K_z = \Span \{ \cY_t(a,z)b : a \in M, b \in N\}$ and $\tilde K_z = \Span \{ \tilde \cY_t(a,z)b : a \in M, b \in N\}$, then there is an isometric map $u:K_z \to \tilde K_z$ such that 
\begin{equation}\label{eqn: u transport}
u \cY_t(a,z)b = \tilde \cY_t(a,z)b.
\end{equation}
Recall that $\cY_t$ and $\tilde \cY_t$ are dominant and lie in $I_{Hilb}$ by the definition of transport module.
Thus by Lemma \ref{lem: intertwiner image is a core}, $K_z$ and $\tilde K_z$ are dense, and so $u$ extends to a unitary $u:\cH_K \to \tilde \cH_K$.

We now verify that $u$ gives an isomorphism of $V$-modules.
Let $v \in V$ and for $n \in \bbZ$ let $v^K_{(n)}$ denote the closure of the mode of $Y^K(v,x)$, and similarly for $v^{\tilde K}_{(n)}$.
Note that $K_z$ and $\tilde K_z$ are invariant cores for $v^K_{(n)}$ and $v^{\tilde K}_{(n)}$, respectively, by Lemma \ref{lem: intertwiner image is a core}.
For $\xi \in K_z$ we have $v^{\tilde K}_{(n)}u\xi = uv^{K}_{(n)}\xi$ by \eqref{eqn: vn yazb}.
Hence $v^{\tilde K}_{(n)}u = u v^{K}_{(n)}$ as closed operators.

In particular, $u$ intertwines the actions of $L_0^K$ and $L_0^{\tilde K}$, and thus $u$ maps $K$ onto $\tilde K$.
Restricting the  identity $v_{(n)}u = u v_{(n)}$ to finite energy vectors we see that $u$ is a $V$-module isomorphism.

Suppose $a \in M$ and $b \in N$ are homogeneous.
Then since $u$ preserves conformal dimensions and $u \cY_t(a,z)b = \tilde \cY_t(a,z)b$, we must have $u a^{\cY_t}_{(k)}b = a^{\tilde \cY_t}_{(k)} b$ for all $k$.
Hence $u \cY_t = \tilde \cY_t$.

Uniqueness of $u$ is clear, as $\cY_t$ is dominant by the definition of transport module and so vectors of the form $\cY_t(a,z)b$ span a dense subspace of $\cH_K$ by Lemma \ref{lem: intertwiner image is a core}.

\end{proof}

We now compute transport for pairs of modules one of which is the vacuum.
The definition of the transport module for $M$ and $N$ is not symmetric in the inputs, and it turns out that the transport module for $V$ and $N$ is much simpler than the opposite order.

\begin{Proposition}
Let $V$ be a simple unitary VOA and let $N$ be a unitary $V$-module.
Then $(N,Y^N)$ is a transport module for $V$ and $N$.
\end{Proposition}
\begin{proof}
As a consequence of the Jacobi identity for modules (\cite{FHL93}, see also \cite[Prop. 2.13]{GuiUnitarityI}), whenever $\abs{z^\prime} > \abs{z} > \abs{z - z^\prime}$ we have
$$
\bip{Y^N(a_2,z^\prime)Y^N(a_1,z)b_1,b2}
=
\bip{Y^N(Y(a_2,z^\prime - z)a_1,z)b_1,b_2}
$$
with both series converging absolutely.
Note that both series involve integer powers of $z$, so the functions are single-valued.
Thus if $(z,w) \in \cR$ we apply the fact that $N$ is a unitary module to obtain
\begin{align*}
\bip{Y^N(a_1,z)b_1,Y^N(a_2,w)b2}
&=
\bip{Y^N(\theta_{\overline{w}} a_2,\overline{w}^{-1})Y^N(a_1,z)b_1,b2}\\
&=
\bip{Y^N(Y(a_2,\overline{w}^{-1} - z)a_1,z)b_1,b_2}\\
&=
\ip{a_1 \otimes b_1, a_2 \otimes b_2}_{z,w}.\\
\end{align*}
Since $Y^N$ is trivially dominant  and $Y^N \in I_{Hilb} \vertex{N}{V}{N}$ we conclude that $(N,Y^N)$ is a transport module.
\end{proof}

Next, we compute a transport module for $M$ and $V$.
This is somewhat more subtle, as it involves intertwining operators as opposed to module operators.
The key step is to establish a `fusion relation' for annihilation and creation operators (defined in Section \ref{sec: intertwining operators}).

\begin{Lemma}\label{lem: ann cr fusion relation}
Let $V$ be a simple unitary VOA and let $M$ be a unitary $V$-module.
Let $a_1 \in M$, $a_2 \in \overline{M}$, and $u, v \in V$.
Then the series
$$
\bip{\cY_M^-(a_2,w)\cY_M^+(a_1,z)v,u} := \sum_{s \in \bbR} \bip{\cY_M^-(a_2,w)P_s\cY_M^+(a_1,z)v,u}
$$
converges absolutely when $\abs{w} > \abs{z} > 0$,
$$
\bip{Y(\cY_M^-(a_2,w-z)a_1,z)v,u}:=\sum_{s \in \bbR} \bip{Y(P_s\cY_M^-(a_2,w-z)a_1,z)v,u}
$$
converges absolutely when $\abs{z} > \abs{w-z} > 0$,
and
$$
\bip{\cY_M^-(a_2,w)\cY_M^+(a_1,z)v,u} = \bip{Y(\cY_M^-(a_2,w-z)a_1,z)v,u}
$$
when $\abs{w} > \abs{z} > \abs{w-z} > 0$ if $\log w-z$ is taken close to $\log w$ when $z$ is small.
\end{Lemma}
\begin{proof}
Assume without loss of generality that all vectors are homogeneous.
We first prove the lemma when $v = \Omega$.

By part 1 of \cite[Lem. 2.16]{GuiUnitarityI}, if $\abs{w} > \abs{z}$ and we choose $\log w - z$ close to $\log w$ when $z$ is small, we have
\begin{align*}
\bip{\cY_M^-(a_2,w)\cY_M^+(a_1,z)\Omega, u}
&:= \sum_{s \in \bbR} \bip{\cY_M^-(a_2,w)P_s\cY_M^+(a_1,z)\Omega, u}\\
&= \sum_{s \in \bbR} \bip{\cY_M^-(a_2,w)P_s e^{zL_{-1}} a_1, u}\\
&= \bip{e^{z L_{-1}}\cY_M^-(a_2,w - z)a_1,u}
\end{align*}
with the sum converging absolutely.
Note that the expansion of $e^{z L_{-1}}$ in the last expression has only finitely many non-zero terms since $L_1^k b = 0$ for $k$ sufficiently large.
On the other hand, if $z \ne w$
\begin{align*}
\bip{Y(\cY_M^-(a_2, w - z)a_1,z)\Omega,b_2}
&:=
\sum_{s \in \bbR} \bip{Y(P_s\cY_M^-(a_2, w - z)a_1,z)\Omega,u}\\
&=
\sum_{s \in \bbR} \bip{e^{zL_{-1}}P_s\cY_M^-(a_2, w - z)a_1,u}\\
&= \bip{e^{zL_{-1}} \cY_M^-(a_2, w - z)a_1,u}
\end{align*}
with the sum having only finitely many non-zero terms as before.
Comparing these two computations, we see that the lemma has been proven when $v = \Omega$.

Now if $v \in V$ is homogeneous with conformal weight $\Delta_v$, we apply the Borcherds commutator formula twice to obtain
\begin{align*}
\bip{\cY_M^-(a_2, w)\cY_M^+(a_1,z)v, u} &:= \nonumber
\sum_{s \in \bbR} \bip{\cY_M^-(a_2, w)P_s\cY_M^+(a_1,z)v_{(-1)}\Omega, u}\\
\nonumber &=
\sum_{s \in \bbR} \sum_{j \ge 0} \bip{\cY_M^-(a_2, w)P_s\cY_M^+(v_{(j)}a_1,z)\Omega, u} (-z)^{-j-1} \quad +\\
\nonumber &\quad+ \sum_{s \in \bbR} \sum_{\ell \ge 0} \bip{\cY_M^-(v_{(\ell)} a_2, w)P_s\cY_M^+(a_1,z)\Omega, u}  (- w)^{-\ell-1} \quad+\\
&\quad+ \sum_{s \in \bbR}  \bip{\cY_M^-(a_2, w)P_s\cY_M^+(a_1,z)\Omega, v_{(-1)}^\dagger u}
\end{align*}
Note that the sums in $j$ and $\ell$ are finite, and thus by the case $v=\Omega$ completed above we have that $\bip{\cY_M^-(a_2,w)\cY_M^+(a_1,z)v, u}$ converges absolutely when $\abs{w} > \abs{z} > 0$ and
\begin{align}\label{eqn: ann cr eqn 0}
\bip{\cY_M^-(a_2,w)\cY_M^+(a_1,z)v, u} &= 
\sum_{j \ge 0} \bip{Y(\cY_M^-(a_2, w - z)v_{(j)} a_1,z)\Omega,u} (-z)^{-1-j} \quad +\nonumber\\
&\quad + \sum_{\ell \ge 0} \bip{Y(\cY_M^-(v_{(\ell)}a_2, w - z) a_1,z)\Omega,u} (- w)^{-1-\ell} \quad +\nonumber\\
&\quad + \bip{Y(\cY_M^-(a_2, w - z) a_1,z)\Omega,v_{(-1)}^\dagger u}.
\end{align}

On the other hand,
\begin{align}\label{eqn: ann cre eqn 1}
\bip{Y(\cY_M^-(a_2, w - z)a_1,z)v,u} \nonumber
&:= \nonumber
\sum_{s \in \bbR} \bip{Y(P_s\cY_M^-(a_2, w - z)a_1,z)v_{(-1)}\Omega,u}\\
&= \nonumber
\sum_{s \in \bbR} \sum_{j \ge 0} \bip{Y(v_{(j)} P_s \cY_M^-(a_2, w - z)a_1,z)\Omega,u} (z- w)^{-1-j} \,\,\, +\\
&\quad +\bip{Y(\cY_M^-(a_2, w - z) a_1,z)\Omega,v_{(-1)}^\dagger u}
\end{align}
Comparing \eqref{eqn: ann cr eqn 0} and \eqref{eqn: ann cre eqn 1}, we see that it suffices to show that the double sum in \eqref{eqn: ann cre eqn 1} converges absolutely when $\abs{z} > \abs{w-z} > 0$ and
\begin{align}\label{eqn: ann cr new goal}
&\sum_{s \in \bbR} \sum_{j \ge 0} \bip{Y(v_{(j)} P_s \cY_M^-(a_2, w - z)a_1,z)\Omega,u} (z- w)^{-1-j}
\quad=\nonumber\\
=\quad&\,\,\sum_{j \ge 0} \bip{Y(\cY_M^-(a_2, w - z)v_{(j)} a_1,z)\Omega,u} (-z)^{-1-j}
\quad+ \nonumber\\
&+\quad\sum_{\ell \ge 0} \bip{Y(\cY_M^-(v_{(\ell)}a_2, w - z) a_1,z)\Omega,u} (- w)^{-1-\ell}
\end{align}
when $\abs{w} > \abs{z} > \abs{z-w} > 0$.

In fact, we establish \eqref{eqn: ann cr new goal} with the order of summation in $s$ and $j$ reversed, but since we will establish absolute convergence the desired result will still follow.
So assume $\abs{z} > \abs{z-w} > 0$, and observe
\begin{align}\label{eqn: ann cre eqn 2}
\sum_{j \ge 0}\sum_{s \in \bbR} &\bip{Y(v_{(j)} P_s \cY_M^-(a_2, w - z)a_1,z)\Omega,u} (-z)^{-1-j} \quad =\nonumber\\
&= \sum_{j \ge 0}\sum_{s \in \bbR} \bip{Y( P_{s+\Delta_v-j-1} v_{(j)}\cY_M^-(a_2,w - z)a_1,z)\Omega,u} (-z)^{-1-j} \nonumber \\
&= \sum_{j \ge 0}\sum_{s \in \bbR} \bip{Y( P_{s} v_{(j)}\cY_M^-(a_2, w - z)a_1,z)\Omega,u} (-z)^{-1-j} \nonumber \\
\nonumber&=   \sum_{j \ge 0}\sum_{s \in \bbR} \bip{Y(P_s \cY_M^-(a_2, w - z)v_{(j)} a_1,z)\Omega,u} (-z)^{-1-j} \quad + \\
\nonumber& \quad  \sum_{j \ge 0}\sum_{s \in \bbR}\sum_{\ell \ge 0} \bip{Y(P_s\cY_M^-(v_{(\ell)} a_2, w - z)a_1,z)\Omega,u} \cdot \binom{j}{\ell} ( w -z)^{j-\ell} (-z)^{-1-j}\\
\nonumber&=   \sum_{j \ge 0}\sum_{s \in \bbR} \bip{Y(P_s\cY_M^-(a_2, w - z)v_{(j)} a_1,z)\Omega,u} (-z)^{-1-j} \quad + \\
\nonumber& \quad  \sum_{\ell \ge 0}\sum_{j \ge 0}\sum_{s \in \bbR} \bip{Y(P_s\cY_M^-(v_{(\ell)} a_2, w - z)a_1,z)\Omega,u} \cdot \binom{j}{\ell} (w - z)^{j-\ell} (-z)^{-1-j}\\
\end{align}
The first equality in \eqref{eqn: ann cre eqn 2} is the fact that $v_{(j)}$ raises energy by $\Delta_v -j-1$, the second is reindexing the sum, the third is the Borcherds commutator formula, and the fourth uses the fact that the sum in $\ell$ and $j$ is finite.
The first of the two terms remaining at the end of \eqref{eqn: ann cre eqn 2} converges absolutely to the first term in \eqref{eqn: ann cr new goal}, using the case $v=\Omega$ already established.
We now show absolute convergence of the second term at the end of \eqref{eqn: ann cre eqn 2} to the second term of \eqref{eqn: ann cr new goal}. 

In fact, since the sum over $\ell$ is finite, it suffices to verify absolute convergence for fixed $\ell$.
We have
\begin{align}\label{eqn: ann cre eqn 3}
\sum_{j \ge 0}\sum_{s \in \bbR} & \bip{Y(P_s\cY_M^-(v_{(\ell)} a_2, w - z)a_1,z)\Omega,u} \cdot \binom{j}{\ell} (w -z)^{j-\ell} (-z)^{-1-j} \quad= \nonumber\\
&=\quad \left(\sum_{s \in \bbR}  \bip{Y(P_s\cY_M^-(v_{(\ell)} a_2,w - z)a_1,z)\Omega,u} \right)\left( \sum_{j \ge 0} \binom{j}{\ell} (w -z)^{j-\ell} (-z)^{-1-j}\right )\nonumber\\
&=\quad \left(\sum_{s \in \bbR}  \bip{Y(P_s\cY_M^-(v_{(\ell)} a_2,\tilde w - z)a_1,z)\Omega,u} \right)(-w)^{-\ell-1}.
\end{align}
Here we have used the fact that when $\abs{z} > \abs{w-z} > 0$ we have absolute convergence
$$
\sum_{j \ge 0}\binom{j}{\ell} (w -z)^{j-\ell} (-z)^{-1-j} = (-w)^{-\ell-1},
$$
which is easily established for $\ell=0$ using geometric series and then the general case may be obtained by differentiating in $w$.
Hence the sum in \eqref{eqn: ann cre eqn 3} converges absolutely to
$$
\bip{Y(\cY_M^-(v_{(\ell)}a_2, w - z) a_1,z)\Omega,u} (- w)^{-1-\ell}.
$$
Plugging this expression into the end of \eqref{eqn: ann cre eqn 2}, we see that \eqref{eqn: ann cr new goal} has been proven, completing the lemma.
\end{proof}

%\begin{proof}
%$$
%\bip{\cY_M^-(\theta_{\overline{w}} a_2,\overline{w}^{-1})\cY_M^+(a_1,z)\Omega,b_2}
%=
%\sum_{s \in \bbR} \bip{\cY_M^+(a_1,z)\Omega, P_s \cY_M^+(a_2,w)b_2}
%=
%\sum_{s \in \bbR} \bip{e^{z L_{-1}} a_1, P_s \cY_M^+(a_2,w)b_2}
%$$
%
%$$
%\bip{Y(\cY_M^-(\theta_{\overline{w}}a_2, \overline{w}^{-1}-z)a_1,z)\Omega,b_2}
%=
%\sum_{s \in \bbR} \bip{Y(P_s\cY_M^-(\theta_{\overline{w}}a_2, \overline{w}^{-1}-z)a_1,z)\Omega,b_2}
%=
%\sum_{s \in \bbR}\bip{e^{z L_{-1}} P_s\cY_M^-(\theta_{\overline{w}}a_2, \overline{w}^{-1}-z)a_1,b_2} 
%=
%\sum_{s \in \bbR}\bip{ P_s\cY_M^-(\theta_{\overline{w}}a_2, \overline{w}^{-1}-z)a_1,e^{\overline{z} L_1}b_2} 
%=
%\sum_{s \in \bbR}\bip{\cY_M^-(\theta_{\overline{w}}a_2, \overline{w}^{-1}-z)a_1,e^{\overline{z} L_1}b_2} 
%$$
%\end{proof}

\begin{Proposition}\label{prop: ann cr transport module}
Let $V$ be a simple unitary VOA and let $M$ be a unitary $V$-module.
Then $\cY_M^+ \in I_{Hilb}\binom{M}{M V}$ and $(M,\cY_M^+)$ is a transport module for $M$ and $V$.
\end{Proposition}
\begin{proof}
Observe that if $1 > \abs{z} > 0$, then by Lemma \ref{lem: ann cr fusion relation}, we have convergence of 
$$
\bip{\cY_M^-(\theta_{\overline{z}} a,\overline{z}^{-1})\cY_M^+(a,z)b,b}
=
\bip{\cY_M^+(a,z)b,\cY_M^+(a,z)b}
$$
for all $a \in M$ and $b \in V$, and thus $\cY_M^+ \in I_{Hilb}\binom{M}{M V}$.
By the same lemma, whenever $(z,w) \in \cR$ the series defining
$$
\ip{a_1 \otimes b_1, a_2 \otimes b_2}_{z,w} := \bip{Y(\cY_M^-(\theta_{\overline{w}}a_2,\overline{w}^{-1}-z)a_1,z)b_1,b_2}
$$
converges absolutely and is equal to
$$
\bip{\cY_M^-(\theta_{\overline{w}}a_2,\overline{w}^{-1})\cY_M^+(a_1,z)b_1,b_2}
=
\bip{\cY_M^+(a_1,z)b_1,\cY_M^+(a_2,w)b_2},
$$
as multi-valued functions which agree when $\log w-z$ is taken close to $\log w$ when $z$ is small.
Note that $\bip{\cY_M^+(a_1,z)b_1,\cY_M^+(a_2,w)b_2}$ is in fact single-valued, and the definition of $\ip{ \, \cdot \, , \, \cdot \,}_{z,w}$ as a single-valued function $\cR$ (preceding Definition \ref{def: transport form}) agrees with taking $\log w$ close to $\log w - z$ when $z$ is small, completing the proof.
\end{proof}

\subsection{Regular VOAs and positivity of transport matrices}\label{sec: transport for rational}

Let $V$ be a simple unitary VOA and let $M$ and $N$ be unitary $V$-modules.
Recall from Section \ref{sec: transport forms} that the transport forms for $M$ and $N$ are defined by
\begin{equation}
\ip{a_1 \otimes b_1, a_2 \otimes b_2}_{z,w} := 
\bip{Y^N(\cY_M^-(\theta_{\overline{w}} a_2, \overline{w}^{-1}-z)a_1,z)b_1,b_2}_N
\end{equation}
for
$$
(z,w) \in \cR := \{(z,w) \in \bbC^2 : 1 > \abs{z} > \abs{\overline{w}^{-1} - z} > 0 \mbox{ and } 1 > \abs{w} > 0\},
$$
provided the relevant series converges absolutely (in which case we say that the transport form exists for $M$ and $N$).
In this section, we discuss the transport form in more detail in the context of regular VOAs.
Regularity is a strong semisimplicity condition on the representation theory of $V$ originally introduced in \cite{DongLiMason97}.
Under the mild assumption that a VOA is of CFT type (which is always true for simple unitary VOAs), ``regular'' is equivalent to ``rational and $C_2$-cofinite'' \cite{ABD04}.

We will not discuss the precise definition of regularity, and instead list the relevant consequences of this property.
A regular VOA $V$ has finitely many simple modules up to isomorphism, every $V$-module can be written as a direct sum of simple modules, and the spaces of intertwining operators $I \binom{K}{MN}$ are finite-dimensional.
Most importantly, the fusion product theory of Huang and Lepowsky applies to regular VOAs with the additional property of being ``strong CFT type,'' which always holds for simple unitary VOAs \cite{HuangLepowsky13,HuangDEs,HuangVerlinde,HuangModularity}.

\begin{Proposition}\label{prop: rational intertwiner Hilb}
Let $V$ be a simple regular unitary VOA, and let $M$, $N$ and $K$ be unitary $V$-modules.
Then we have:
\begin{enumerate}
\item $I \binom{K}{MN} = I_{Hilb} \binom{K}{M N}$.
\item The transport form for $M$ and $N$ exists.
\end{enumerate}
\end{Proposition}
\begin{proof}
First we prove (1).
Let $a \in M$ and $b \in N$ be homogeneous vectors, and suppose $0 < \abs{z} < 1$.
Observe that $\cY(a,z)b \in \cH_K$ if and only if the series defining $\ip{\cY(a,z)b, \cY(a,z)b}$ converges.
As series, we have
$$
\ip{\cY(a,z)b, \cY(a,z)b} = \bip{\cY^\dagger(\theta_{\overline{z}}a,\overline{z}^{-1})\cY(a,z)b,b}.
$$
By Huang's work \cite{HuangDEs}, the series on the right-hand side converges since $\abs{z}^{-1} > 1 > \abs{z}$, and thus we conclude that $\cY(a,z)b \in \cH_K$.

The proof of (2) is similar, as the convergence of the transport form on $\cR$ is a special case of Huang's convergence of iterated intertwining operators in the same reference.
\end{proof}

We now consider a fixed simple unitary regular VOA $V$, and fix a set $\Irr(V)$ of isomorphism class representatives of simple $V$-modules.
If $M$ and $N$ are $V$-modules, then the fusion product $M \boxtimes N$ is a $V$-module which represents the functor $K \mapsto I\binom{K}{M N}$.
That is, $M \boxtimes N$ is a $V$-module equipped with an intertwining operator $\cY_\boxtimes \in \binom{M \boxtimes N}{M N}$ such that for any $\cY \in \binom{K}{M N}$ there is a unique $V$-module homomorphism $\varphi:M \boxtimes N \to K$ such that $\cY = \varphi \cY_\boxtimes$.\footnote{%
In many applications, it is necessary to replace the notion of fusion product with the closely related notion of ``$P(z)$-tensor product.'' A $P(z)$-tensor product is essentially equivalent to a fusion product (as defined here) along with a choice of $\log z$.}
The fusion product always exists, and is given by
$$
M \boxtimes N = \bigoplus_{K \in \Irr(V)} I\binom{K}{M N}^* \otimes K,
$$
equipped with a certain canonical intertwining operator (see \cite[Prop. 12.5]{HuangLepowskyPartIII}).

\begin{Proposition}\label{prop: exists transport iso}
Let $V$ be a simple unitary regular VOA, let $M$ and $N$ be unitary $V$-modules, and suppose that $M \boxtimes N$ admits a unitary structure, which we fix.
Then there exists a unique $\Lambda \in \End_{V}(M \boxtimes N)$ satisfying
\begin{equation}\label{eqn: transport endo}
\bip{a_1 \otimes b_1, a_2 \otimes b_2}_{z,w} = \bip{\Lambda \cY_\boxtimes(a_1,z)b_1, \cY_\boxtimes(a_2,w)b_2}
\end{equation}
for all $a_i \in M$, $b_i \in N$, as single-valued functions on $\cR$.
\end{Proposition}
\begin{proof}
By the work of Huang \cite{HuangDEs}, there exists a $V$-module $K$ and intertwining operators $\cY_1 \in I \binom{K}{M N}$ and $\cY_2 \in I\binom{N}{\overline{M} K}$ such that
\begin{align*}
\ip{a_1 \otimes b_1, a_2 \otimes b_2}_{z,w} &= \bip{Y^N(\cY_M^-(\theta_{\overline{w}} a_2,\overline{w}^{-1}-z)a_1,z)b_1,b2}\\
&= \bip{\cY_2(\theta_{\overline{w}}a_2,\overline{w}^{-1})\cY_1(a_1,z)b_1,b_2}\\
&= \bip{\cY_1(a_1,z)b_1,\cY_2^\dagger(a_2,w)b_2}
\end{align*}
as functions on $\cR$.
By the definition of fusion product, there exist $\varphi_i \in \Hom_V(M \boxtimes N, K)$ such that $\cY_1 = \varphi_1 \cY_\boxtimes$ and $\cY_2^\dagger = \varphi_2 \cY_\boxtimes$.
Thus 
$$
\ip{a_1 \otimes b_1, a_2 \otimes b_2}_{z,w} = \ip{\varphi_2^*\varphi_1 \cY_\boxtimes(a_1,z)b_1, \cY_\boxtimes(a_2,w)b_2},
$$
which establishes existence.

Uniqueness is clear since $\cY_\boxtimes \in I_{Hilb}$ by Proposition \ref{prop: rational intertwiner Hilb} and thus the span of vectors of the form $\cY_{\boxtimes}(a,z)b$ is dense in $\cH_{M \boxtimes N}$ by Lemma \ref{lem: intertwiner image is a core}.
\end{proof}

\begin{Definition}
Under the hypotheses of Proposition \ref{prop: exists transport iso}, the endomorphism $\Lambda \in \End_{V}(M \boxtimes N)$ satisfying \eqref{eqn: transport endo} is called the \emph{transport endomorphism}\footnote{%
The term \emph{transport} in this context originates in the ``transport formula'' of Antony Wassermann \cite[\S 31]{Wa98}, which showed existence and positivity of the transport operator when $V=V(\mathfrak{sl}_n,k)$ and $M$ is the module corresponding to the Young diagram with one box.
One of the key goals of this article is to interpret Wassermann's approach in the general context of unitary VOAs.
}
associated to $M$ and $N$.
\end{Definition}

Note that $\Lambda$ depends on both the inner product chosen on $M \boxtimes N$, as well as the intertwining operator $\cY_\boxtimes$ making $(M \boxtimes N, \cY_\boxtimes)$ into a fusion product.
However, we will be interested in the positivity of $\Lambda$, which will not depend on these choices.
Note that $\Lambda$ extends to a bounded operator on the Hilbert space completion $\cH_{M \boxtimes N}$, and as usual we use the same symbol for this extension.

\begin{Remark}\label{rem: transport matrices}
In the work of Gui \cite{GuiUnitarityI,GuiUnitarityII}, the operator $\Lambda$ is studied as a matrix by choosing bases for the spaces $I \binom{K}{M N}$.
We can compare the two approaches as follows.
Let $\Irr(V) = \{M_1, \ldots, M_n\}$ be a complete list of simple modules, equipped with a unitary structure.
Let $M \boxtimes N = \bigoplus_{\alpha} K_\alpha$ be an orthogonal decomposition into simple modules, and for each $\alpha$ choose an isometry $v_\alpha:M_{i(\alpha)} \to M \boxtimes N$ such that $\operatorname{Ran}(v_\alpha) = K_\alpha$, for the appropriate $i(\alpha) \in \{1, \ldots, n\}.$
Then by the universal property of the fusion product, $\{v_\alpha^* \cY_\boxtimes\}$ is a set of intertwining operators comprising a basis for $\bigoplus_{i=1}^n I \binom{M_i}{M N}$. 

Since the $M_i$ are simple, we have $v_\beta^* \Lambda v_\alpha = \Lambda^{\alpha \beta} \operatorname{id}_{M_{i(\alpha)}}$ for some scalars $\Lambda^{\alpha \beta}$ (which are necessarily zero when $i(\alpha) \ne i(\beta)$).
By construction $\operatorname{id}_{M \boxtimes N} = \sum_{\alpha} v_\alpha v_\alpha^*$, and so we can expand the identity  \eqref{eqn: transport endo} characterizing $\Lambda$ to obtain (for appropriate $z$ and $w$)
\begin{align}\label{eqn: transport matrix}
\bip{Y^N(\cY_M^-(a_2,w-z)a_1)b_1,b_2} &=
\bip{\cY_\boxtimes^\dagger(a_2,w)\Lambda\cY_\boxtimes(a_1,z)b_1,b_2} \nonumber\\
&= \sum_{\alpha,\beta}\bip{\cY_\boxtimes^\dagger(a_2,w)(v_\beta v_\beta^*)\Lambda(v_\alpha v_\alpha^*)\cY_\boxtimes(a_1,z)b_1,b_2} \nonumber\\
&= \sum_{\alpha,\beta}\Lambda^{\alpha \beta} \bip{\cY_\beta^\dagger(a_2,w)\cY_\alpha(a_1,z)b_1,b_2},
\end{align}
where $\cY_\alpha = v_{\alpha}^*\cY_\boxtimes$ and the product $\cY_\beta^\dagger(a_2,w)\cY_\alpha(a_1,z)$ is understood as $0$ if $i(\alpha) \ne i(\beta)$.
Thus the matrix $\Lambda^{\alpha\beta}$ is exactly the transport matrix \cite[Eqn. (6.16)]{GuiUnitarityII} considered by Gui.
By construction, the operator $\Lambda \ge 0$ if and only if the matrix $\Lambda^{\alpha \beta} \ge 0$.
\end{Remark}

In this context, Gui argued that the matrix $\Lambda^{\alpha \beta}$ is invertible, invoking the deep work of Huang on the rigidity of the category $\mathrm{Rep}(V)$ of $V$-modules.
We summarize the argument here for the convenience of the reader.

\begin{Lemma}[{\cite[Thm. 6.7, Proof step 3]{GuiUnitarityII}}]
\label{lem: transport invertible}
Let $V$ be a simple unitary regular VOA, let $M$ and $N$ be unitary $V$-modules, and suppose that $M \boxtimes N$ admits a unitary structure.
Then the transport endomorphism $\Lambda$ associated to $M$ and $N$ is invertible.
\end{Lemma}
\begin{proof}
We refer the reader to \cite[Thm. 6.7, Proof step 3]{GuiUnitarityII} for complete details of the argument.
The argument relies on the fact that the fusion product $\boxtimes$ makes the category $\mathrm{Rep}(V)$ of $V$-modules into a rigid tensor category in a natural way \cite{HuangModularity}, and we refer the reader to \cite[\S2.4]{GuiUnitarityI} for more details regarding the tensor category $\mathrm{Rep}(V)$.

We use the notation of Remark~\ref{rem: transport matrices}.
As the transport endomorphism $\Lambda$ acts on $M \boxtimes N = \bigoplus_\alpha K_\alpha$ via the matrix $\Lambda^{\alpha \beta}$ (after identifying each $K_\alpha$ with an appropriate isomorphism class representative $M_{i(\alpha)}$), it suffices to prove that the matrix $\Lambda^{\alpha \beta}$ is invertible.
It suffices to consider when $M$ and $N$ are simple modules, as decomposing $M$ and $N$ into simples yields a corresponding direct sum decomposition of the matrix $\Lambda^{\alpha \beta}$, and it suffices to show that each direct summand is invertible.

As described above, in the tensor category $\mathrm{Rep}(V)$ the hom spaces $\Hom(M \boxtimes N, K)$ are identified with spaces of intertwining operators $I\binom{K}{M \, N}$.
The associators provide a natural isomorphism 
\begin{equation}\label{eqn: associator 6j}
\bigoplus_{i=1}^n\Hom(\overline{M} \boxtimes M, M_i) \otimes \Hom(M_i \boxtimes N, N) \to \bigoplus_{j=1}^n \Hom(\overline{M} \boxtimes M_j, N) \otimes \Hom(M \boxtimes N, M_j).
\end{equation}
Identifying $\Hom$ spaces with intertwining operators, we obtain an isomorphism
\[
\bigoplus_{i=1}^n I\binom{M_i}{\overline{M} \, M} \otimes  I\binom{N}{M_i \boxtimes N} \to \bigoplus_{j=1}^n I\binom{N}{\overline{M} \, M_j} \otimes I\binom{M_j}{M\, N}.
\]
The tensor category structure on $\Rep(V)$ is defined so that the calculation \eqref{eqn: transport matrix} is equivalent to the fact that
\begin{equation}\label{eqn: assoc iso intertwiners}
\cY_M^- \otimes Y^N \mapsto \sum_{\alpha,\beta} \Lambda^{\alpha \beta} \cY_\beta^\dagger \otimes \cY_\alpha
\end{equation}
under this isomorphism.

Since $M$ is simple the morphism in $\Hom(\overline{M} \boxtimes M, V)$ corresponding to the intertwining operator $\cY^-_M \in I\binom{V}{\overline{M} \, M}$ must be a non-zero multiple of the evaluation $\operatorname{ev}_M$.
Similarly, since $N$ is simple $Y^N$ must correspond to a non-zero multiple of the left unitor $\ell_N \in \Hom(V \boxtimes N, N)$. 
Thus \eqref{eqn: assoc iso intertwiners} tells us that under the associativity isomorphism \eqref{eqn: associator 6j} we have
\begin{equation}\label{eqn: ev otimes unitor goes to}
\operatorname{ev}_M \otimes \ell_N \mapsto \mu \sum_{\alpha,\beta} \Lambda^{\alpha \beta} u_\beta \otimes v_\alpha^*
\end{equation}
where $\mu$ is a non-zero complex number, $u_\beta$ is the morphism in $\Hom(\overline{M} \boxtimes M_{i(\beta)}, N)$ corresponding to $\cY^\dagger_\beta$, and we note that by construction $v_\alpha^*$ is the morphism  in $\Hom(M \boxtimes N, M_{i(\alpha)})$ corresponding to $\cY_\alpha$.
%As noted in Remark~\ref{rem: transport matrices}, the $X_\alpha$ provide a basis for $\bigoplus_{i=1}^n \Hom(M \boxtimes N, M_i)$, and similarly the $X_\beta^\dagger$ provide a basis for $\bigoplus_{i=1}^n \Hom(\overline{M} \boxtimes M_i, N)$.

It is now a purely categorical fact that the matrix $\Lambda^{\alpha \beta}$ is invertible.
Suppressing associators and unitors, \eqref{eqn: ev otimes unitor goes to} yields an identity in $\Hom(\overline{M} \boxtimes M \boxtimes N, N)$
\[
\mathrm{ev}_M \otimes \id_N = \mu \sum_{\alpha,\beta} \Lambda^{\alpha \beta} u_\beta(\id_{\overline{M}} \otimes v_\alpha^*).
\]
Precomposing by $\mathrm{coev}_M \otimes \id_{M \boxtimes N}$ and applying the zig-zag identity (and continuing to suppress associators and unitors) yields an identity in $\End(M \boxtimes N)$
\[
\id_{M \boxtimes N} = \mu \sum_{\alpha,\beta} \Lambda^{\alpha \beta} \tilde u_\beta v_\alpha^*,
\]
where $\tilde u_\beta = (\id_M \otimes u_\beta) (\mathrm{coev}_M \otimes \id_{M_{i(\beta)}}) \in \Hom(M_{i(\beta)}, M \boxtimes N)$.
It follows that $\Lambda^{\alpha \beta}$ is invertible, as otherwise the identity morphism of $M \boxtimes N$ would factor through the projection onto a proper subobject.
\end{proof}

The following easy proposition outlines the relationship between positivity, fusion products, and transport modules for regular VOAs.
In Section \ref{sec: transport and tensor products}, we will expand on this relationship for arbitrary unitary VOAs.

\begin{Proposition}\label{prop: transport for regular VOAs}
Let $V$ be a simple regular unitary VOA, and let $M$ and $N$ be unitary $V$-modules.
Let $(M \boxtimes N, \cY_\boxtimes)$ be the fusion product of $M$ and $N$, and suppose that $M \boxtimes N$ admits a unitary structure, which we fix.
Then the following are equivalent.
\begin{enumerate}
\item For every $(z,z) \in \cR$, the sesquilinear form $\ip{ \, \cdot \, , \, \cdot \, }_{z,z}$ on $M \otimes N$ is positive semi-definite.
\item For some $(z,z) \in \cR$, the sesquilinear form $\ip{ \, \cdot \, , \, \cdot \, }_{z,z}$ on $M \otimes N$ is positive semi-definite.
\item $\Lambda$ is a positive definite operator on $\cH_{M \boxtimes N}$
\item There exists a unitary structure on $M \boxtimes N$ such that $(M \boxtimes N, \cY_\boxtimes)$ is a transport module for $M$ and $N$.
\item There exists a transport module for $M$ and $N$.
\end{enumerate}
\end{Proposition}
\begin{proof}
We prove (1) $\implies$ (2) $\implies$ (3) $\implies$ (4) $\implies$ (5) $\implies$ (1).
The implication (1) $\implies$ (2) is a tautology.
Now assume (2).
By Proposition \ref{prop: exists transport iso}
$$
\ip{a_1 \otimes b_1, a_2 \otimes b_2}_{z,z} = \ip{\Lambda \cY_\boxtimes(a_1,z)b_1,\cY_\boxtimes(a_2,z)b_2}
$$
as single-valued functions on $\cR$.
Recall (preceding Definition \ref{def: transport module}) that the right-hand side is defined on $\cR$ by requiring that $\cY_\boxtimes(a_1,z)b$ and $\cY_\boxtimes(a_2,z)b_2$ be interpreted with the same value of $\log z$.
Since vectors $\cY_\boxtimes(a,z)b$ span a dense subspace of $\cH_{M \boxtimes N}$ by Lemma \ref{lem: intertwiner image is a core}, the operator $\Lambda$ is positive semi-definite because the form $\ip{ \, \cdot \, , \, \cdot \,}_{z,z}$ is.
The operator $\Lambda$ is positive definite (i.e. invertible) by Lemma~\ref{lem: transport invertible}.

Now assume (3). 
Then the invariant inner product $\ip{ \Lambda \, \cdot \, , \, \cdot \,}$ on $M \boxtimes N$ makes $M \boxtimes N$ into a transport module by Proposition \ref{prop: exists transport iso}.

The implication (4) $\implies$ (5) is a tautology so we now assume (5).
If $(M \boxtimes_t N, \cY_t)$ is a transport module for $M$ and $N$, then 
$$
\bip{a_1 \otimes b_1, a_2 \otimes b_2}_{z,z} = \bip{\cY_t(a_1,z)b_2, \cY_t(a_2,z)b_2}
$$
is clearly positive semi-definite for any $(z,z) \in \cR$.
\end{proof}

\begin{Remark}\label{rem: transport as regular fusion}
When the conditions of Proposition \ref{prop: transport for regular VOAs} hold, we obtain a new construction of the fusion product $M \boxtimes N$ as a subspace of the Hilbert space completion of $M \otimes N$ with respect to the transport form.
The intertwining operator $\cY_\boxtimes$ can be recovered from the fact that the canonical inclusion of $M \otimes N$ into its Hilbert space completion should be a $P(z)$-intertwining map.
Thus a unitary fusion product $M \boxtimes N$ is constructed from just the transport form; this is described in more detail in \cite{TenerPositivity} (see also Section \ref{sec: transport forms}).
Note that this construction is `manifestly unitary,' in that $M \boxtimes N$ comes already endowed with an inner product (in fact, it was constructed from the inner product).
We will verify the conditions of Proposition \ref{prop: transport for regular VOAs} in many examples in Section \ref{sec: examples}, and thus this new manifestly unitary construction produces unitary fusion products in these examples.
\end{Remark}

\begin{Remark}
The hypothesis throughout this section that $M \boxtimes N$ admits a unitary structure is expected to be vacuous, as it is conjectured that for unitary regular VOAs every module admits a unitary structure.
However, it has not even been shown that if $M$ and $N$ are unitary modules, then $M \boxtimes N$ admits a unitary structure.
The construction described in Remark \ref{rem: transport as regular fusion} provides one possibility for proving this second, weaker conjecture.
If the transport forms are positive semidefinite, then one can attempt to build $M \boxtimes N$ inside the Hilbert space completion of $M \otimes N$ in a manifestly unitary way.
\end{Remark}
%
%\begin{Remark}
%If the transport form is positive for every pair of $V$-modules $M$ and $N$, then Gui has shown that this unitary structure on $M \boxtimes N$ can be used to make the category of unitary $V$-modules into a unitary modular tensor category \cite{GuiUnitarityI,GuiUnitarityII}.
%We conjecture that the hypotheses of Proposition \ref{prop: transport for regular VOAs} always hold.
%In fact, we conjecture that the transport form is positive semidefinite even without the assumption that $V$ is regular (see Section \ref{sec: transport and tensor products}).
%\end{Remark}

\subsection{Transport, fusion products, and positivity}\label{sec: transport and tensor products}\label{sec: transport conjectures}

We conclude Section \ref{sec: positivity} by describing the role of transport modules in a unitary fusion product theory for unitary VOAs, and by presenting conjectures related to transport modules.

\subsubsection{The positivity conjecture}

Let $V$ be a simple unitary VOA and let $M$ and $N$ be unitary $V$-modules.
In Definition \ref{def: transport form}, we defined a family of sesquilinear forms $\ip{a_1 \otimes b_1, a_2 \otimes b_2}_{z,w}$ on $M \otimes N$ which are a single-valued function of $(z,w) \in \cR$, where
$$
\cR= \{(z,w) \in \bbC^2 : 1 > \abs{z} > \abs{\overline{w}^{-1} - z} > 0 \mbox{ and } 1 > \abs{w} > 0\}.
$$
In Proposition \ref{prop: transport for regular VOAs} and Remark \ref{rem: transport as regular fusion} we described how the transport forms provide a `manifestly unitary' approach to the study of fusion products of modules for regular VOAs, so long as we know that a certain sesquilinear form is positive semidefinite.
Moreover, we saw that this gave the same answer as the standard approach to fusion products of modules.

For VOAs which are not regular, the fusion product theory can be quite a bit more complicated.
As a first hurdle, our definition of module (with $L_0$ diagonalizable and finite-dimensional weight spaces) is too restrictive to study the representation theory of general unitary VOAs, as even the fusion product of simple unitary modules should not be a direct sum of such modules.
In \cite{TenerPositivity}, we sketched a proposal for a class of unitary modules which should be well-behaved from the perspective of fusion product theory.
For the rest of Section \ref{sec: transport and tensor products} we will use the word `module' to mean something more general than the strong sense used throughout the rest of the article.

Beyond the problem of selecting the correct category of modules, there is a fundamental issue of defining the fusion product $M \boxtimes N$, even when $M$ and $N$ are simple.
As an example of what can go wrong, we consider the Virasoro VOAs $L(c,0)$ when $c > 1$.
For such VOAs, the simple modules are parametrized by $h \in \bbC$, and $\dim I \binom{h_3}{h_1 h_2} = 1$ regardless of the values of $h_i$.
In light of this, it seems unwise to attempt to use the same definition of fusion product as in the case of regular VOAs.

Instead, we propose the notion of transport module as a candidate fusion product of unitary $V$-modules.
The motivation for this conjecture comes from the study of conformal nets.
Like unitary VOAs, conformal nets are supposed to axiomatize two-dimensional chiral conformal field theory, and so one expects a dictionary between the phenomena which arise in the two contexts.
However, the study of fusion products in the context of conformal nets looks quite different than in the context of VOAs.
While the fusion product of VOA modules is understood in terms of a universal property, the fusion product of conformal net representations is given by an explicit construction.
This construction is independent of any `niceness' properties of the conformal net, and produces a braided tensor category \cite{BKLR15}.
If one believes in the dictionary between conformal nets and VOAs, then there should be a corresponding VOA construction.
The notion of `transport module' is precisely this analogue.
Since the fusion product of conformal net representations always exists, we conjecture the same for transport modules.

\begin{Conjecture}\label{conj: existence of transport}
Let $V$ be a simple unitary VOA and let $M$ and $N$ be $V$-modules.
Then the transport form for $M$ and $N$ exists and there exists a transport module for $M$ and $N$.
\end{Conjecture}

Note that a part of Conjecture \ref{conj: existence of transport} is finding the correct class of (generalized) modules for which the statement is true.
One immediate consequence of Conjecture \ref{conj: existence of transport} would be the following.

\begin{Conjecture}[Positivity conjecture]\label{conj: positivity conjecture}
Let $V$ be a simple unitary VOA and let $M$ and $N$ be $V$-modules.
Then for every $(z,z) \in \cR$ the transport forms $\ip{ \, \cdot \, , \, \cdot \,}_{z,z}$ on $M \otimes N$ are positive semidefinite.
\end{Conjecture}

In Section \ref{sec: positivity from BLVO}, we will show that the positivity conjecture holds for many modules over many VOAs, even badly non-rational ones like $\Vir_c$ with $c \ge 1$. 
We will do so by identifying the transport form with the Connes fusion inner product from conformal nets.
By Proposition \ref{prop: transport for regular VOAs}, this implies the existence of many transport modules for regular VOAs.
We believe these results provide strong evidence for Conjectures \ref{conj: existence of transport} and \ref{conj: positivity conjecture}.
In fact, the positivity conjecture should be a major component of proving the existence of transport modules in the non-regular regime.

\subsubsection{From the positivity conjecture to transport modules}

It is clear that the existence of transport modules implies the positivity conjecture.
We now discuss how to go in the opposite direction and construct a transport module given positivity.
This is a conceptually important step, as the positivity conjecture feels more approachable than the problem of constructing a transport module, as it is entirely intrinsic to the modules $M$ and $N$.
What we will now see is that the existence of transport modules can also be set up as a problem intrinsic to the modules $M$ and $N$.

If the positivity conjecture holds for $M$ and $N$, we let $\cH_t$ be the Hilbert space completion/quotient of $M \otimes N$ with respect to the transport form for some $z$.
This Hilbert space comes equipped with a map $M \otimes N \to \cH_t$, which we denote $\boxtimes_t$.
To construct a transport module, we need to put a $V$-module structure on a dense subspace $K$ of $\cH_t$, and show that there is an intertwining operator $\cY_t \in I \binom{K}{M N}$ such that $\cY_t(a,z)b = a \boxtimes_t b$.

In general the image of $\boxtimes_t$ will have trivial intersection with our desired subspace $K \subset \cH_t$, but because of Lemma \ref{lem: intertwiner image is a core} we can still hope that the image of $\boxtimes_t$ is a core for the action of modes on finite energy vectors.
Thus we proceed in the only way we can: we define an action of $V$ on the image of $\boxtimes_t$ by 
$$
v_{(n)} (a \boxtimes_t b) = a \boxtimes_t v_{(n)}b + \sum_{j=0}^\infty \binom{n}{j} z^{n-j} v_{(j)}a \boxtimes_t b.
$$
We conjecture that the induced densely defined action of $L_0$ on $\cH_t$ is essentially self-adjoint, and let $K = \bigcup_{A \ge 0} p_{[0,A]} \cH_t$, where $p_{[0,A]}$ is the spectral projection for $L_0$ corresponding to the interval $[0,A]$.
The modes $v_{(n)}$ should be closable with $K$ a core and satisfy the Borcherds identity, at which point we will have successfully constructed a transport module.

While there are many steps in the above construction which look difficult to prove, we point out that if a transport module exists, then \emph{a posteriori} this construction indeed recovers that transport module.
Since we demonstrate the existence of many transport modules in Section \ref{sec: examples}, we have many examples where this construction works.
Thus Conjecture \ref{conj: existence of transport} is equivalent to the conjecture that this construction always produces a transport module.
In particular, in the case of regular VOAs, we conjecture that this construction always produces a fusion product $M \boxtimes N$ (in light of Proposition \ref{prop: transport for regular VOAs}).
Our construction is in principal not so different than the `working backwards' construction of Huang-Lepowsky \cite{HuangLepowskyPartIII}.
The primary difference is that Huang-Lepowsky first construct a $V$-module on a certain  complicated subspace of the algebraic dual $(M \otimes N)^*$, whereas our proposed construction produces the fusion product module on the nose.

\subsubsection{From positivity to transport in vertex tensor categories}\label{sec: pos to transport}

In a vertex tensor category of $V$-modules, for any modules $M$ and $N$ one expects to be able to find a module $K$ and intertwining operators $\cY_1 \in I \binom{K}{M N}$ and $\cY_2 \in I\binom{N}{\overline{M} K}$ such that 
\begin{equation}\label{eqn: non transport splitting}
Y^N(\cY_M^-(a_2,w-z)a_1,z) = \cY_2(a_2,w)\cY_1(a_1,z)
\end{equation}
in the sense of matrix coefficients and multi-valued functions.
Thus if there exists a vertex tensor category of unitary $V$-modules, we expect to be able to find a unitary $K$ with this property.
Crucially, this should hold even when the category is not rigid.

It is not hard to show that if a decomposition \eqref{eqn: non transport splitting} exists and the positivity conjecture holds for $M$ and $N$, then there exists a transport module for $M$ and $N$.
We now sketch the argument.

Set $\tilde K = K \oplus K = K \otimes \bbC^2$,
$$
\cY = \begin{pmatrix} \cY_1\\ \cY_2^\dagger \end{pmatrix} \in I \binom{\tilde K}{M N},
\qquad
\varphi = \operatorname{id_K} \otimes \begin{pmatrix} 0 & 0\\ 1 & 0\end{pmatrix} \in \End(\tilde K).
$$
Then the identity \eqref{eqn: non transport splitting} becomes
$$
Y^N(\cY_M^-(a_2,w-z)a_1,z) = \cY^\dagger(a_2,w)\varphi\cY(a_1,z).
$$
From the convergence of the transport form, we see that $\cY \in I_{Hilb}$.
Let $M \boxtimes_t N \subset \tilde K$ be the subspace generated by the modes of $\cY$ acting on $N$, and let $v:\cH_{\tilde K} \to \cH_{M \boxtimes_t N}$ be the projection.
If we set $\Lambda:=v\varphi v^*$ and $\cY_t = v \cY$, then we have
$$
Y^N(\cY_M^-(a_2,w-z)a_1,z) = \cY_t^\dagger(a_2,w)\Lambda\cY_t(a_1,z),
$$
and now $\cY_t$ is dominant.
From here we can argue as in Proposition \ref{prop: transport for regular VOAs} that there is an inner product on $M \boxtimes_t N$ making $(M \boxtimes_t N, \cY_t)$ in to a transport module.
Indeed, since the transport from is positive and $\Lambda$ is manifestly a bounded operator, we can explicitly define the new inner product on $M \boxtimes_t N$ via $\ip{\Lambda \, \cdot \, , \, \cdot \,}$.

\subsubsection{Positive vectors in VOAs}

%We conclude this section with one final digression on the subject of positivity in VOAs.
Positivity is a crucial aspect of the approach to quantum field theory using von Neumann algebras, and a key fact is that positivity is preserved by representations of von Neumann algebras.
The study of unitary VOAs would be greatly aided by a corresponding tool, which we now present a candidate for.

\begin{Definition}
Let $V$ be a simple unitary VOA, let $\xi \in \cH_V$, let $N$ be a unitary $V$-module, and suppose $0 < \abs{z} < 1$.
We say that $\xi$ is $z$-positive acting on $M$ if for every $b_1,b_2 \in N$ the series defining $\ip{Y^N(\xi,z)b_1,b_2}$ converges absolutely and
$$
\ip{Y^N(\xi,z)b,b} \ge 0
$$
for all $b \in N$.
\end{Definition}

In general, one should not expect to be able to find $z$-positive vectors lying in $V$ itself, only in some completion.
Observe that if $(z,z) \in \cR$, then by Proposition \ref{prop: ann cr transport module} for any unitary $V$-module $M$ we have
\begin{align*}
\ip{Y(\cY_M^+(\theta_{\overline{z}} a,\overline{z}^{-1}-z)a,z)b,b} = \bip{Y_M^+(a,z)b,\cY_M^+(a,z)b} \ge 0
\end{align*}
and thus $\cY_M^-(\theta_{\overline{z}}a,\overline{z}^{-1}-z)a$ is $z$-positive acting on $V$ for every $a \in M$.
The positivity conjecture, on the other hand, would imply that $\cY_M^-(\theta_{\overline{z}}a,\overline{z}^{-1}-z)a$ is $z$-positive acting on every unitary $V$-module.
We extend the positivity conjecture to a more general conjecture about positive vectors in VOAs:
\begin{Conjecture}
Let $V$ be a simple unitary VOA, let $\xi \in \cH_V$, and let $N$ be a unitary $V$-module.
If $\xi$ is $z$-positive acting on $V$ then $\xi$ is $z$-positive acting on $N$.
\end{Conjecture}

%% file: Sec_4_FusionRules.tex
\newpage
\section{Identification of VOA and conformal net fusion products}\label{sec: fusion rules for bimodules}

\subsection{VOA positivity via conformal nets}\label{sec: positivity from BLVO}

In Section \ref{sec: positivity}, we gave several conjectures regarding positivity of certain matrices and sesquilinear forms and described their role in the fusion product theory of unitary VOAs.
In general it is quite difficult to establish positivity of sesquilinear forms unless the positivity is manifest from the construction.
In contrast to the situation with VOAs, when chiral CFTs are studied using von Neumann algebras (i.e. as conformal nets), then the sesquilinear forms which arise are manifestly positive in light of basic results in the theory of operators on Hilbert spaces.

All of the existing positivity results for specific models require detailed case by case analysis to obtain certain analytic properties of intertwining operators (which are often not conjectured to hold for all VOAs).
In this section, we present a method for proving positivity results about VOAs via conformal nets which in contrast uses essentially no information about intertwining operators, and only uses analytic properties which are conjectured to hold for all unitary VOAs.
These analytic properties have been checked in a large class of examples, and thus by applying our method we can easily verify positivity properties in new classes of examples as well as systematically recovering earlier examples.
Our primary tool is the notion of bounded localized vertex operators and the correspondence $M \leftrightarrow \pi^M$ of representations introduced in \cite{GRACFT1, TenerGRACFT2} and summarized in Section \ref{sec: BLVO}.

The main idea in linking conformal net representations and the transport form is that if a representation $\pi^N$ exists, we can ``apply $\pi^N$'' to the transport form for $M$ and $V$ and obtain the transport form for $M$ and $N$.
This is extremely useful, as we completely understand the transport form for $M$ and $V$ since we have shown that $(M,\cY_M^+)$ is a transport module (Proposition \ref{prop: ann cr transport module}).
Since $\pi^N$ preserves positivity, we will be able to deduce positivity of the transport form for $M$ and $N$ from the one for $M$ and $V$.

The following lemma uses geometric ideas to show that the transport form $[ \, \cdot \, , \, \cdot \,]_{z_1,z_2}$ for $M$ and $N$ can be computed via $\pi^N$ when the $z_i$ lie in the interior of generalized annuli and are sufficiently close to each other.

\begin{Lemma}\label{lem: first positivity}
Let $V$ be a simple unitary VOA with bounded localized vertex operators, let $M$ and $N$ be unitary $V$-modules, and assume that $\cY_M^- \in I_{Hilb} \binom{V}{\overline{M} M}$ and that $\pi^N$ exists.
Let $I$ be an interval, let $(1,A_1),(1,A_2) \in \scA_I$, and suppose there is some $w_0 \in \Int(A_2^*,A_1) \cap S^1$.
For $a_i \in M$, let $x_{a_i}(z) = \cY_M^+(a_i,z)A_i$.

There exists a non-empty open subset $\cU \subset \Int(A_1) \cap \Int(A_2)$ with the following properties.
The boundary $\partial \cU$ contains an interval of the circle that contains the point $w_0$.
We have $\cU \times \cU \subset \cR$, and whenever $z_1,z_2 \in \cU$, $a_1,a_2 \in M$, and $b_1,b_2 \in N$:
\begin{enumerate}
\item The series defining the transport form for $M$ and $N$
$$
\ip{a_1 \otimes b_1, a_2 \otimes b_2}_{z_1,z_2} := \bip{Y^N(\cY_M^-(\theta_{\overline{z_2}}a_2, \overline{z_2}^{-1} - z_1)a_1,z_1)b_1,b_2}
$$ 
converges absolutely.
\item The sesquilinear forms $x_{a_1}(z_1)$ and $x_{a_2}(z_2)$ define bounded operators, and we have $x_{a_2}(z_2)^*x_{a_1}(z_1) \in \cA_V(I)$.
\item If $\xi_i \in \cH_N$ with $A_i \xi_i = b_i$, then 
$$
\ip{\pi^N(x_{a_2}(z_2)^*x_{a_1}(z_1))\xi_1,\xi_2} = \ip{a_1 \otimes b_1, a_2 \otimes b_2}_{z_1,z_2}.
$$
\end{enumerate}
\end{Lemma}
\begin{proof}
First we choose the set $\cU$.
Since $V$ has bounded localized vertex operators, for every $z \in \Int(A_2^*,A_1)$ there exists some $s > 0$ such that $A_2^* Y(s^{L_0} \cdot, z)A_1$ defines a bounded map $\cH_{V} \otimes \cH_V \to \cH_V$.
Note that we used the fact that if $(1,A_i) \in \scA_I$ then $(A_2^*,A_1) \in \scA_I$ by the definition of system of generalized annuli.
%Let $F \subset \Int(A_2^*, A_1)$ be a compact set set such that $J \subset \Int(F)$.
By \cite[Cor. 4.5]{TenerGRACFT2}, we can find a neighborhood $w_0 \in \tilde \cU \subset \Int(A_2^*,A_1)$ and some $s > 0$ such that $A_2^* Y(s^{L_0} \cdot, z)A_1$  is bounded for every $z \in \tilde \cU$.
Since $\pi^N$ exists and is unitarily implemented, $A_2^* Y^N({s}^{L_0} \cdot, z)A_1$  is bounded for every $z \in \tilde \cU$ as well.
Without loss of generality we may take $s < 1$.
By shrinking $\tilde \cU$ about $w_0 \in S^1$, we may ensure that $\abs{z_1} > s > \abs{\overline{z_2}^{-1} - z_1}$ for all $z_1,z_2 \in \tilde \cU$.
We now set $\cU = \tilde \cU \cap \Int(A_1) \cap \Int(A_2)$, which is non-empty since $w_0 \in \partial \Int(A_1) \cap \partial \Int(A_2)$.
By construction, $\cU \times \cU \subset \cR$ (where $\cR$ is as in \eqref{eqn: R def}) and there is some interval contained in $\partial\cU$ that contains $w_0$.
$$
\begin{array}{ccc}
\begin{tikzpicture}[baseline={([yshift=-.5ex]current bounding box.center)}]
	\coordinate (aa) at (150:1cm);
	\coordinate (bb) at (270:1cm);
	\coordinate (cc) at (210:1.5cm);
	\coordinate (a) at (120:1cm);
	\coordinate (b) at (240:1cm);
	\coordinate (c) at (180:.25cm);
% OUTER CURVER
	\filldraw[fill=red!10!blue!20!gray!30!white] (aa)  .. controls ++(250:.6cm) and ++(120:.4cm) .. (cc) .. controls ++(300:.4cm) and ++(190:.6cm) .. (bb) arc (270:510:1cm);
% INNER CURVE
	\fill[fill=white] (a)  .. controls ++(210:.6cm) and ++(90:.4cm) .. (c) .. controls ++(270:.4cm) and ++(150:.6cm) .. (b) -- ([shift=(240:1cm)]0,0) arc (240:480:1cm);
	\draw ([shift=(240:1cm)]0,0) arc (240:480:1cm);
	\draw (a) .. controls ++(210:.6cm) and ++(90:.4cm) .. (c);
	\draw (b) .. controls ++(150:.6cm) and ++(270:.4cm) .. (c);
% INTERVAL J
 %\draw[line width=0.7mm] (185:1cm) arc (185:207:1cm);
 \node at (195:1cm) {\tiny \textbullet};
% \tilde \cU
 \fill[pattern= crosshatch, pattern color=red] (183:0.9cm) -- (183:1.09cm) arc (180:215:1.09cm) -- (212:0.9cm) arc (212:180:0.9cm);
  	\draw[densely dotted] (0,0) circle (1cm);
% COORDINATE LABELS
%	\node at (a) {(a)};
%	\node at (b) {(b)};
%	\node at (c) {(c)};
\end{tikzpicture}
\quad&\quad
\begin{tikzpicture}[baseline={([yshift=-.5ex]current bounding box.center)}]
	\coordinate (a) at (120:1cm);
	\coordinate (b) at (240:1cm);
	\coordinate (c) at (180:.25cm);
% BIG DISK
	\fill[fill=red!10!blue!20!gray!30!white] (0,0) circle (1cm);
	\draw (0,0) circle (1cm);
	% dotted circle

% CURVED BOUNDARY REGION
	\fill[fill=white] (a)  .. controls ++(210:.6cm) and ++(90:.4cm) .. (c) .. controls ++(270:.4cm) and ++(150:.6cm) .. (b) -- ([shift=(240:1cm)]0,0) arc (240:480:1cm);
	\draw ([shift=(240:1cm)]0,0) arc (240:480:1cm);
	\draw (a) .. controls ++(210:.6cm) and ++(90:.4cm) .. (c);
	\draw (b) .. controls ++(150:.6cm) and ++(270:.4cm) .. (c);
% INTERVAL J
 \draw[line width=0.4mm] (185:1.0cm) arc (185:210:1.0cm);
  \node at (195:1cm) {\tiny \textbullet};
% \tilde \cU
 \fill[pattern= crosshatch, pattern color=purple] (183:0.9cm) -- (183:1.0cm) arc (180:215:1.0cm) -- (212:0.9cm) arc (212:180:0.9cm);
% COORDINATE LABELS
%	\node at (a) {(a)};
%	\node at (b) {(b)};
%	\node at (c) {(c)};
\end{tikzpicture}
\quad& \quad
\begin{tikzpicture}[baseline={([yshift=-.5ex]current bounding box.center)}]
	\coordinate (a) at (150:1cm);
	\coordinate (b) at (270:1cm);
	\coordinate (c) at (210:0.5cm);
% CURVED BOUNDARY REGION
	\filldraw[fill=red!10!blue!20!gray!30!white] (0,0) circle (1cm);
	\fill[white] (a)  .. controls ++(250:.6cm) and ++(120:.4cm) .. (c) .. controls ++(300:.4cm) and ++(170:.6cm) .. (b) arc (270:510:1cm);
	\draw ([shift=(270:1cm)]0,0) arc (270:510:1cm);
	\draw (a) .. controls ++(250:.6cm) and ++(120:.4cm) .. (c);
	\draw (b) .. controls ++(170:.6cm) and ++(300:.4cm) .. (c);
% BIG DISK
	\draw (0,0) circle (1cm);
% INTERVAL J
 \draw[line width=0.4mm] (185:1.0cm) arc (185:210:1.0cm);
  \node at (195:1cm) {\tiny \textbullet};
% \tilde \cU
 \fill[pattern= crosshatch, pattern color=purple] (183:0.9cm) -- (183:1.0cm) arc (180:215:1.0cm) -- (212:0.9cm) arc (212:180:0.9cm);
% COORDINATE LABELS
%	\node at (a) {(a)};
%	\node at (b) {(b)};
%	\node at (c) {(c)};
\end{tikzpicture}
\\
w_0 \in \tilde \cU \subset \Int(A_2^*, A_1)
&
\cU \subset \Int(A_1)
&
\cU \subset \Int(A_2)
\end{array}
$$
We now prove (1).
Fix a holomorphic branch of $\log$ on $\cU$ (which is possible since $\cU \times \cU \subset \cR$, see Remark \ref{rmk: U times U log}).
Without loss of generality assume that the $a_i$ and $b_i$ are homogeneous.
Fix $z_1,z_2 \in \cU$, and set 
\begin{equation*}\label{eqn: xi in rH}
\xi = \cY_M^-(\theta_{\overline{z_2}}a_2, \overline{z_2}^{-1} - z_1)a_1.
\end{equation*}
Note that the expression defining $\xi$ is multi-valued, but since $(z_1,z_2) \in \cR$ we may make the standard choice, as in the definition of the transport form (Definition \ref{def: transport form}).
Since $\cY_M^- \in I_{Hilb}\binom{V}{\overline{M} M}$ (by Proposition \ref{prop: ann cr transport module}) we have $\xi \in r^{L_0} \cH_V$ whenever $r > \abs{\overline{z_2}^{-1} - z_1}$.
Since $s >  \abs{\overline{z_2}^{-1} - z_1}$, we may fix some $r < s$ such that $\xi \in r^{L_0} \cH_V$.
Let $\eta = r^{-L_0}\xi$, and let $t=r/s$.
Then we have $\xi = s^{L_0} t^{L_0}\eta$ and $t < 1$.

Since $\pi^N$ exists, we have that $A_i \in \Ann^{in}(\cH_N)$ (see Definition~\ref{def: piM}).
This ensures that there exist $\xi_i \in \cH_N$ such that $A_i \xi_i = b_i$ (Definition~\ref{def: incoming and outgoing annuli}).
Observe that since the $b_i$ are homogeneous, there is some $\Delta \in \bbR$ such that 
$\bip{Y^N(P_\alpha\xi,z)b_1,b_2} = 0$ unless $\alpha \in \Delta + \bbZ_{\ge 0}$, where $P_\alpha$ is the projection onto vectors of conformal weight $\alpha$.
We now compute
\begin{align}\label{eqn: geom comparison}
\ip{a_1 \otimes b_1, a_2 \otimes b_2}_{z_1,z_2} &= \sum_{\alpha \in \Delta + \bbZ_{\ge 0}} \bip{Y^N(P_\alpha \xi, z_1)b_1,b_2} \nonumber\\
&=\sum_{\alpha \in \Delta + \bbZ_{\ge 0}} \bip{A_2^*Y^N(P_\alpha s^{L_0}t^{L_0} \eta, z_1)A_1 \xi_1,\xi_2} \nonumber\\
&=\sum_{\alpha \in \Delta + \bbZ_{\ge 0}} t^\alpha \bip{A_2^*Y^N(s^{L_0} P_\alpha \eta, z_1)A_1 \xi_1,\xi_2}
\end{align}
Now, since $A_2^* Y^N(s^{L_0} \cdot, z_1)A_1$ is a bounded operator and $\norm{P_\alpha \eta} \le \norm{\eta}$, we see that 
$$
\abs{\bip{A_2^*Y^N(s^{L_0} P_\alpha \eta, z_1)A_1 \xi_1,\xi_2}}
$$
is bounded independent of $\alpha$.
Thus \eqref{eqn: geom comparison} converges absolutely by comparison to a geometric series.
This completes the proof of (1).

We now proceed to (2).
Let $u_i \in V$ be homogeneous and pick $\xi_i^\prime$ such that $A_i \xi_i^\prime = u_i$.
Since $(Y^M, \cY_M^+)$ is a transport module for $M$ and $V$ (Proposition \ref{prop: ann cr transport module}),
\begin{align}\label{eqn: awk lem split form}
\bip{x_{a_1}(z_1)\xi_1^\prime, x_{a_2}(z_2)\xi_2^\prime} &= 
\bip{\cY_M^+(a_1,z_1)u_1, \cY_M^+(a_2,z_2)u_2}
 = \bip{Y(\xi,z_1)u_1,u_2} \nonumber\\ 
 &= \bip{A_2^*Y(\xi,z_1)A_1\xi_1^\prime, \xi_2^\prime}. 
\end{align}
Note that the series defining $\bip{Y(\xi,z_1)u_1,u_2}$ converges absolutely by part (1) (applied to $N=V$).
In using the definition of the transport module, it was important that $\log z_i$ were chosen from a common branch of $\log$ on $\cU$.

Since $\xi \in s^{L_0}\cH_V$ and $z_1 \in \cU$ the operator $A_2^*Y(\xi,z_1)A_1$ is bounded by the definition of bounded localized vertex operators.
Thus examining \eqref{eqn: awk lem split form} in the case $a_1 = a_2$ and $z_1=z_2$ we see that $x_{a_1}(z_1)$ and $x_{a_2}(z_1)$ are bounded operators.
Moreover we have $x_{a_2}(z_2)^*x_{a_1}(z_1) = A_2^*Y(\xi,z_1)A_1 \in \cA_V(I)$, again by the definition of bounded localized vertex operators.
This completes the proof of (2).

Finally we prove (3).
Since $\xi \in s^{L_0}\cH_V$, there exists a sequence of finite energy vectors $\eta_m$ such that $\eta_m \to s^{-L_0}\xi$.
Since $ A_2^*Y(s^{L_0}-,z_1)A_1$ is bounded, we have 
$$
A_2^*Y(s^{L_0}\eta_m,z_1)A_1 \to A_2^*Y(\xi,z_1)A_1
$$
in norm.
Thus from the definition of $\pi^N$ we have
$$
\pi^N_I(x_{a_2}(z_2)^*x_{a_1}(z_1)) = \pi^N_I(A_2^*Y(\xi,z_1)A_1) = A_2^*Y^N(\xi,z_1)A_1.
$$
Part (3) now follows immediately.
\end{proof}

As an easy consequence of Lemma \ref{lem: first positivity} we can verify the positivity conjecture (Conjecture \ref{conj: positivity conjecture}) under mild hypotheses, which we will be able to verify for many models (including non-rational ones).

\begin{Proposition}\label{prop: general positivity}
Let $V$ be a simple unitary VOA with bounded localized vertex operators, and let $M$ and $N$ be unitary $V$-modules.
Assume that $\pi^N$ exists and that $\cY_M^- \in I_{Hilb} \binom{V}{\overline{M} M}$.
Then there is some $r < 1$ such that whenever $(z,z) \in \cR$ and $\abs{z} > r$ the transport form
$$
\bip{a_1 \otimes b_1, a_2 \otimes b_2}_{z,z} := \bip{Y^N(\cY_M^-(\theta_{\overline{z}} a_2,\overline{z}^{-1}-z)a_1,z)b_1,b_2}
$$
on $M \otimes N$ exists and is positive semidefinite.
\end{Proposition}
\begin{proof}
Observe that 
$$
\bip{e^{i \theta L_0}a_1 \otimes e^{i \theta L_0}b_1, e^{i \theta L_0}a_2 \otimes e^{i \theta L_0}b_2}_{z,z} = \bip{a_1 \otimes b_1, a_2 \otimes b_2}_{e^{i \theta} z, e^{i \theta} z}
$$
and so the existence and positivity of the transport form only depends on $\abs{z}$.

So choose some interval $I$ and some $(1,A) \in \scA_I$ with non-empty interior (which exists by the definition of system of generalized annuli).
The boundary of $\Int(1,A)$ contains an interval of $S^1$, and so we may choose some $w_0 \in \Int(A^*,A) \cap S^1$.
Let $\cU$ be the open set provided by Lemma \ref{lem: first positivity}, which necessarily contains points $z = (1-\varepsilon)w_0$ for $\varepsilon$ sufficiently small.
Thus it suffices to prove the positivity of the transport form for $z \in \cU$.

Since $\cH_N$ is a separable Hilbert space (as a part of our definition of $V$-module), there is a unitary $W: \cH_V \to \cH_N$ such that $\pi^N_I = \Ad W$.
Let $\tilde A = (A|_{\ker(A)^\perp})^{-1}$ as a map $N \to \cH_N$.
Now with notation as in Lemma \ref{lem: first positivity} part (3), we have
\begin{align*}
\bip{a_1 \otimes b_1, a_2 \otimes b_2}_{z,z} &= \bip{Wx_{a_2}(z)^*x_{a_1}(z)W^*\tilde A b_1,\tilde A b_2}\\
&= \bip{x_{a_1}(z)W^* \tilde A b_1, x_{a_2}(z)W^* \tilde A b_2}.
\end{align*}
Note that if $T$ is any linear map from a vector space to a Hilbert space, the sesquilinear form $\ip{T(v_1),T(v_2)}_{\cH}$ on $V$ is clearly positive semidefinite.
Applying this to $T:M \otimes N \to \cH_M$ given by
$$
T(a \otimes b) = x_a(z)W^* \tilde A b
$$
completes the proof.
\end{proof}

\begin{Remark}
The positivity conjecture states that $\ip{ \, \cdot \, , \, \cdot \,}_{z,z}$ is positive semidefinite whenever $(z,z) \in \cR$; that is, when $1 > \abs{z} > 1/\sqrt{2}$.
Thus to prove the conjecture for a pair of modules $M$ and $N$, one has to get the number $r$ from Proposition \ref{prop: general positivity} down to $1/\sqrt{2}$.
In practice, we do not have much control over the value of $r$ provided by Proposition \ref{prop: general positivity} as it is stated.
This is mainly due to the fact that we are using an abstract system of generalized annuli.
In future work, we expect to give a complete treatment of the representation theory of the semigroup of partially thin annuli, which will provide a system of generalized annuli for which the value of $r$ provided by Proposition \ref{prop: general positivity} is the optimal $1/\sqrt{2}$.
\end{Remark}

Even in its present form we regard Proposition \ref{prop: general positivity} as very strong evidence for the positivity conjecture.
Indeed, in Section \ref{sec: examples} we will use it to demonstrate that the transport form is positive semidefinite for a wide range of modules for rational and non-rational VOAs, including many representations $L(c,h)$ of the Virasoro algebra with $c \ge 1$.
Applying the argument of Section \ref{sec: pos to transport} we can construct transport modules for VOAs which are not necessarily rational, so long as they have a tensor category of modules.

We can now provide even stronger evidence of the positivity conjecture for regular VOAs.

\begin{Theorem}\label{thm: positivity conjecture for regular}
Let $V$ be a simple regular unitary VOA, let $M$ and $N$ be unitary $V$-modules.
Assume that $M \boxtimes N$ admits a unitary structure and that $\pi^N$ exists.
Then the transport form $\ip{ \, \cdot \, , \, \cdot \,}_{z,z}$ exists and is positive semidefinite for all $(z,z) \in \cR$.
\end{Theorem}
\begin{proof}
Note that $\cY_M^- \in I_{Hilb}$ and the transport forms exist since $V$ is regular (Proposition \ref{prop: rational intertwiner Hilb}).
By Proposition \ref{prop: general positivity} the transport form is positive semidefinite for some $z$, and thus by Proposition \ref{prop: transport for regular VOAs} it is positive semidefinite for all $(z,z) \in \cR$.
\end{proof}

Theorem \ref{thm: positivity conjecture for regular} provides many examples of transport modules.

\begin{Corollary}\label{cor: transport for regular}
Let $V$, $M$ and $N$ be as in Theorem \ref{thm: positivity conjecture for regular}, and let $(M \boxtimes N, \cY_\boxtimes)$ be a fusion product.
Then there is a unique unitary structure on $M \boxtimes N$ making $(M \boxtimes N, \cY_\boxtimes)$ into a transport module.
\end{Corollary}
\begin{proof}
Theorem \ref{thm: positivity conjecture for regular} guarantees that for all $(z,z) \in \cR$ the transport form $\ip{ \, \cdot \, , \, \cdot \,}_{z,z}$ is positive semidefinite.
By Proposition \ref{prop: transport for regular VOAs}, this is equivalent to the desired conclusion.
\end{proof}

In Section \ref{sec: examples} we describe many models to which we can apply Theorem \ref{thm: positivity conjecture for regular} and Corollary \ref{cor: transport for regular}.

We call the inner product on $M \boxtimes N$ from Corollary \ref{cor: transport for regular} the standard inner product.
Since the inner product depends on the choice of $\cY_\boxtimes$, by Proposition \ref{prop: transport module unique} a transport module pair $(M \boxtimes N, \cY_\boxtimes)$ is unique up to unique unitary isomorphism.
Thus the standard unitary structure on $M \boxtimes N$ is unique up to canonical unitary isomorphism.

If $V$ is a simple regular unitary VOA, then the work of Huang and Lepowsky demonstrates that the category $\Rep(V)$ of $V$-modules is a modular tensor category in a natural way \cite{HuangModularity} (see also \cite{CreutzigKanadeMcRae2017ax} or \cite{GuiUnitarityI} for a description of this structure).
In order to study the representation theory of regular unitary VOAs from a categorical perspective, however, one needs to consider the category $\Rep^u(V)$ of unitary $V$-modules, and one would like to know that it has the structure of a unitary modular tensor category in a natural way.
In order to define a tensor product on $\Rep^u(V)$, it is necessary to choose an invariant inner product on $M \boxtimes N$ given such a choice for $M$ and $N$.
One would like to do so in such a way that the associators, braidings, and twists which Huang and Lepowsky construct for $\Rep(V)$ are unitary.

The problem of defining such an inner product is quite different than the analogous problem for unitary representations of groups, for which there is a canonical, easy-to-construct invariant inner product on the tensor product of unitary representations.
On the other hand, it is not obvious that the fusion product of unitary VOA modules even admits an invariant inner product (although the transport form provides a candidate), but this is widely believed to be true.
In fact, it is widely believed that every module for a regular unitary VOA admits a unitary structure, although this fails quite badly for VOAs which are not regular (e.g. the Heisenberg VOA).

Gui has recently shown that $\Rep^u(V)$ is naturally a unitary modular tensor category provided that certain sesquilinear forms are positive \cite[Thm. 7.9]{GuiUnitarityII}.
In light of Corollary \ref{cor: transport for regular}, we are now able to verify this positivity hypothesis in a much wider class of models.

\begin{Corollary}\label{cor: UMTC}
Let $V$ be a simple regular unitary VOA with bounded localized vertex operators.
Suppose that every $V$-module $M$ admits a unitary structure, and that $\pi^M$ exists.
Then the standard inner product of Corollary \ref{cor: transport for regular} makes $\Rep^u(V)$ into a unitary modular tensor category in a natural way.
\end{Corollary}
\begin{proof}
We will apply \cite[Thm. 7.9]{GuiUnitarityII}.
In order to apply this theorem, we must verify that the transport matrices $\Lambda^{\alpha \beta}$ of \eqref{eqn: transport matrix} are positive (semi)definite.
By Proposition \ref{prop: transport for regular VOAs}, this is equivalent to the positivity of the transport form proven in Theorem \ref{thm: positivity conjecture for regular}.
\end{proof}

We apply the above result to examples such as WZW models and $W$-algebras in Section \ref{sec: examples}. 

\subsection{Transport modules and Connes fusion}\label{sec: transport and fusion}

Our original motivation for studying the standard inner product on the fusion product of unitary $V$-modules is that it is closely connected with the inner product on the Connes fusion of conformal net representations.
In this section, we will systematically develop that connection using bounded localized vertex operators.
We will work under the hypothesis that $V$ has bounded localized vertex operators, $M$ and $N$ are unitary $V$-modules such that $\pi^M$ and $\pi^N$ exists, and that a transport module $M \boxtimes_t N$ exists.
The results of Section \ref{sec: positivity from BLVO} provide large families of VOA and modules satisfying these hypotheses, and when $V$ is regular we have that $M \boxtimes_t N$ is just the usual fusion product $M \boxtimes N$ equipped with the standard inner product.

The goal of this section is to construct unitary maps 
$
U_I: \cH_M \boxtimes_I \cH_N \to \cH_{M \boxtimes_t N}
$
relating the Connes fusion to the transport module in a natural way.
These maps will depend on both the interval $I$ and a choice of a branch of the logarithm on $I$, but this additional dependence will be suppressed.

The key step in constructing a unitary 
\begin{equation}\label{eqn: fusion iso}
U_I:\cH_M \boxtimes_I \cH_N \to \cH_{M \boxtimes_t N}
\end{equation}
is to define an appropriate map 
\begin{equation}\label{eqn: transport domain codomain}
\Hom_{\cA_V(I^\prime)}(\cH_V, \cH_M) \to \cB(\cH_N, \cH_{M \boxtimes_t N}),
\end{equation}
which we call the transport map.
Indeed, if $x \in \Hom_{\cA_V(I^\prime)}(\cH_V, \cH_M)$ and $\xi \in \cH_N$, then $x \boxtimes_I \xi$ defines a vector in $\cH_M \boxtimes_{I} \cH_N$, and we have a map 
$$
(x \boxtimes_I -) \, \in \, \cB(\cH_N, \cH_M \boxtimes_{I} \cH_N).
$$
Thus if we had a unitary equivalence \eqref{eqn: fusion iso}, then we could obtain from that a map of the form \eqref{eqn: transport domain codomain}.
One of Wassermann's ideas in \cite{Wa98} was to work backwards and first construct \eqref{eqn: transport domain codomain}, and then use that map to define an isomorphism \eqref{eqn: fusion iso}.

The following fundamental lemma combines Lemma \ref{lem: first positivity}, which shows that the transport form for $M$ and $N$ can be calculated using $\pi^N$, with the definition of transport module to exhibit a correspondence of the form \eqref{eqn: transport domain codomain}.

\begin{Lemma}\label{lem: transport near boundary}
Let $V$ be a simple unitary VOA with bounded localized vertex operators, and let $M$ and $N$ be unitary $V$-modules. 
Assume that $\pi^N$ exists, and that $\cY_M^- \in I_{Hilb} \binom{V}{\overline{M} M}$.
Let $I$ be an interval, let $(1,A_1),(1,A_2) \in \scA_I$ and suppose that $\Int(A_2^*,A_1) \cap S^1 \ne \emptyset$.
Let $\cU \subset \Int(A_1) \cap \Int(A_2)$ be the non-empty open subset provided by Lemma \ref{lem: first positivity} for some choice of $w_0 \in \Int(A_2^*,A_1) \cap S^1$.
Since $\cU \times \cU \subset \cR$, we may fix a holomorphic branch of $\log$ on $\cU$ (Remark \ref{rmk: U times U log}).

Let $(M \boxtimes_t N, \cY_t)$ be a transport module for $M$ and $N$.
For $a_i \in M$, let 
$$
x_i = \cY_M^+(a_i,z_i)A_i, 
\quad y_i = \cY_t(a_i,z_i)A_i,
$$
using our fixed branch of $\log$ on $\cU$.
Then for all $z_1,z_2 \in \cU$, $y_i$ are bounded operators and
$$
\pi_I(x_2^*x_1) = y_2^*y_1.
$$
\end{Lemma}
\begin{proof}
Let $b_1,b_2 \in N$, and since $\pi^N$ exists we may choose $\xi_i \in \cH_N$ such that $A_i \xi_i = b_i$ (see Definition~\ref{def: piM} and Definition~\ref{def: incoming and outgoing annuli}).
By Lemma \ref{lem: first positivity} and the definition of a transport module we have
\begin{align}\label{eqn: transport yields transport form}
&\bip{\pi^N(x_2^*x_1)\xi_1,\xi_2} = \ip{a_1 \otimes a_2, b_1 \otimes b_2}_{z_1,z_2} =\nonumber\\
=\,& \bip{\cY_t(a_1,z_1)b_1,\cY_t(a_2,z_2)b_2}
= \bip{y_1 \xi_1, y_2 \xi_2}.
\end{align}
In the second equality, we use the fact that $\cU \times \cU \subset \cR$, and the definition of transport module gives us an equality of single-valued functions on $\cR$.
This step requires that we chose $\log z_i$ from a single branch of $\log$ on $\cU$ to obtain the correct value of $\bip{\cY_t(a_1,z_1)b_1,\cY_t(a_2,z_2)b_2}$ (see Remark \ref{rmk: U times U log}).

Considering the case $A_1=A_2$, $z_1=z_2$ and $a_1 = a_2$, we have
$$
\bip{\pi^N(x^*x)\xi_1,\xi_2} = \ip{y\xi_1,y\xi_2}
$$
where $x = x_1 = x_2$ and $y = y_1 = y_2$.
Since vectors $\xi_i \in A^{-1}N$ are dense (since $\pi^N$ exists), and $\pi^N(x^*x)$ is bounded, we see that $y$ is bounded.
Returning to \eqref{eqn: transport yields transport form} we see that $\pi^N(x_2^*x_1) = y_2^*y_1$ as desired.
\end{proof}

The map $\cY_M^+(a,z)A \mapsto \cY_t(a,z)A$ which appears in Lemma \ref{lem: transport near boundary} will play an important role in our calculation of $\pi^M \boxtimes_I \pi^N$.
In order to relate this correspondence to Connes' fusion, we will need to know that $\cY_M^+(a,z)A$ actually lies in $\Hom_{\cA(I)}(\cH_V, \cH_M)$.
Thus our remaining results will require the hypothesis $\cY_M^+ \in I_{loc} \binom{M}{M V}$ (see Definition~\ref{def: localized intertwiner}).
In practice, we verify this hypothesis in any example for which we know that $\pi^M$ exists, using a suitable embedding of both $V$ and $M$ in a larger unitary vertex algebra $W$ with bounded localized vertex operators.
If we wish to verify this hypothesis for all irreducible modules $M$ for a given unitary vertex algebra $V$, this would require that $W$ be holomorphic.
Such embeddings are available for WZW models of type EFG and W-algebras of type E, among others.

As an intermediate step to constructing the natural unitary $U_I$ of \eqref{eqn: fusion iso}, we will construct a family of unitaries depending on additional choices (Lemma \ref{lem: fusion iso for component}), and then show that the resulting map does not depend on those choices (Lemmas \ref{lem: two annulus sharing interval transport}-\ref{lem: transport unitary doesnt invariant under rotation}).
The proofs of these lemmas are motivated by \cite{Wa98}.
We assemble the lemmas and construct the unitary $U_I$ in Theorem \ref{thm: transport unitary}.

\begin{Lemma}\label{lem: fusion iso for component}
Let $V$ be a simple unitary VOA with bounded localized vertex operators, and let $M$ and $N$ be unitary $V$-modules. 
Assume that $\pi^M$ and $\pi^N$ exist, and that $\cY_M^+ \in I_{loc} \binom{M}{M V}$.
Let $(M \boxtimes_t N, \cY_t)$ be a transport module for $M$ and $N$.

Let $I$ be an interval and let $(1,A) \in \scA_I$.
Then for every connected component $\cX$ of $\Int(A)$ and every branch of $\log z$ on $\cX$ there is a unique unitary 
$$
U_{I,A,\cX}:\cH_M \boxtimes_I \cH_N \to \cH_{M \boxtimes_t N}
$$
such that for every $a \in M$, $z \in \cX$, and $\xi \in \cH_N$
$$
U_{I,A,\cX}(\cY_M^+(a,z)A \boxtimes_I \xi) = \cY_t(a,z)A\xi.
$$
\end{Lemma}
\begin{proof}
Note that $\cY_M^+ \in I_{loc}$ implies that $\cY_M^- = (\cY_M^+)^\dagger \in I_{loc}$ and in particular $\cY_M^- \in I_{Hilb}$.
Recall that $\Int(A)$ is a degenerate annulus by definition, which is to say that $\Int(A) = \bbD \setminus \overline{D}$ where $\bbD$ is the open unit disk and $D \subset \bbD$ is a Jordan domain with $C^\infty$ boundary that contains $0$. 
Thus there is necessarily some $w_0 \in \Int(A^*,A) \cap \partial\cX \cap S^1$.
Let $\cU \subset \Int(A)$ be the non-empty subset obtained from Lemma \ref{lem: first positivity}, and observe that $\cU \cap \cX \ne \emptyset$.
By replacing $\cU$ with a smaller set, we may assume $\cU \subset \cX$ and we restrict our chosen branch of $\log z$ to $\cU$.

Let $\cS \subset \Hom_{\cA_V(I^\prime)}(\cH_V,\cH_M)$ be given by
$$
\cS = \Span \{ \cY_M^+(a,z)A \, : a \in M, z \in \cU \}.
$$
Note that $\cY_M^+(a,z)A$ is bounded since $\cY_M^+ \in I_{loc}$ and $z \in \cU \subset \Int(A)$.

We first observe that $\cS\cH_V$ is dense in $\cH_M$.
%By the Reeh-Schlieder property, $\overline{\cA_V(I)\Omega} = \cH_V$, and thus for a fixed $z \in \cU$ we have
%$$
%\Span_{a \in M} (\cY_M^+(a,z)A)\cH_V \subset \overline{\cS%\Omega}.
%$$
Indeed, by the definition of bounded localized vertex operators (in particular the fact that $A \in \Ann^{in}(\cH_V)$ as in Definition~\ref{def: bounded insertions}), $V \subset A \cH_V$ and thus $\cY_M^+(a,z)b \in \cS\cH_V$ for all $a \in M$ and $b \in V$.
By Lemma \ref{lem: intertwiner image is a core} it indeed follows that $\overline{\cS\cH_V} = \cH_M$. 

We next claim that the image of $\cS \otimes \cH_N$ under $\boxtimes_I$ is dense in $\cH_M \boxtimes_I \cH_N$.
If $x \in \cS$ and $r \in \cA(I)$ we have $xr \boxtimes_I \xi = x \boxtimes_I \lambda_{I}(r)\xi$, and thus $\cS \boxtimes_I \cH_N = \cS\cA(I) \boxtimes_I \cH_N$, where $\cS\cA(I)$ is the span of expressions $xr$ as above.
Thus it suffices to show that the image of $\cS\cA(I) \otimes \cH_N$ is dense, and in fact we will now argue that the image of the subspace $\cS\cA(I) \otimes \Hom_{\cA(I)}(\cH_V,\cH_N)$ is dense.
By the Reeh-Schlieder property and the previous paragraph we have $\overline{\cS\cA(I)\Omega} = \overline{\cS\cH_V} = \cH_M$.
Since $\boxtimes_I$ extends continuously to $\cH_M \otimes \Hom_{\cA(I)}(\cH_V,\cH_N)$, it indeed follows that the image of $\cS\cA(I) \otimes \Hom_{\cA(I)}(\cH_V,\cH_N)$ under $\boxtimes_I$ is dense, proving the claim.

Now let $x_1,x_2 \in \cS$ and $\xi_1,\xi_2 \in \cH_N$.
Write $x_i = \cY_M^+(a_i,z_i)A$, and set 
$$
y_i = \cY_t(a_i,z_i)A.
$$
Observe that by Lemma \ref{lem: transport near boundary} we have 
$$
\pi^N(x_2^*x_1) =  \pi^N((\cY_M^+(a_2,z_2)A)^*(\cY_M^+(a_1,z_1)A)) = y_2^*y_1.
$$
Thus we have
$$
\ip{x_1 \otimes \xi_1, x_2 \otimes \xi_2}_{\boxtimes_I}
=
\ip{\pi^N(x_2^* x_1) \xi_1, \xi_2}
=
\ip{y_2^*y_1\xi_1, \xi_2}
=
\ip{y_1\xi_1,y_2\xi_2}.
$$
That is, the map $\cS \otimes \cH_N \to \cH_{M \boxtimes_t N}$ given by $x \otimes \xi \mapsto y\xi$ is an isometry when $\cS \otimes \cH_N$ is given the fusion inner product.
Since $\cS \otimes \cH_N$ is dense in the fusion product, we obtain an isometry $U=U_{I,A,\cX}:\cH_M \boxtimes_I \cH_N \to \cH_{M \boxtimes_t N}$.

By construction, 
\begin{equation}\label{eqn: U property on cU}
U(\cY_M^+(a,z)A \boxtimes_I \xi) = \cY_t(a,z)A\xi
\end{equation}
for all $z \in \cU$.
Since $N \subset A \cH_N$, the image of $U$ contains $\cY_t(a,z)b$ for all $a \in M$ and $b \in N$.
Since $\cY_t$ is dominant (by the definition of transport module), we invoke Lemma \ref{lem: intertwiner image is a core} to conclude that the image of $U$ is dense in $\cH_{M \boxtimes_t N}$.
Thus $U$ is a unitary.

It remains to establish \eqref{eqn: U property on cU} for all $z \in \cX$.
The function $z \mapsto \cY_M^+(a,z)A$ is a holomorphic function from $\Int(A)$ into $\Hom_{\cA_V(I^\prime)}(\cH_V,\cH_M)$ (see Section \ref{sec: BLVO}).
Thus if we have $b_1 \in N$, $b_2 \in M \boxtimes_t N$, and $\xi \in \cH_N$ such that $A\xi = b_1$, we have a holomorphic function
$$
z \mapsto \bip{U(\cY_M^+(a,z)A \boxtimes_I \xi),b_2}
$$
on $\Int(A)$, and in particular on $\cX$.
By \eqref{eqn: U property on cU}, we have
$$
\bip{U(\cY_M^+(a,z)A \boxtimes_I \xi)b_1,b_2} = \ip{\cY_t(a,z)b_1,b_2}
$$
for $z \in \cU$.
Both sides of this equality define holomorphic functions on $\cX$, and thus we conclude that the equality holds for all $z \in \cX$.
Since $b_2$ was arbitrary and $A^{-1}\cH_N$ is dense, we conclude that \eqref{eqn: U property on cU} holds for all $z \in \cX$ as well.
\end{proof}

Note that the existence of a branch of $\log$ on $\cX$ was required in Lemma \ref{lem: fusion iso for component}, but such a branch always exists.
Indeed, by definition, $\Int(A)$ is a degenerate annulus not containing 0 such that $\partial \Int(A) \cap S^1 \subset I$, and thus the connected components of $\Int(A)$ are simply connected.

Now given $(1,A) \in \scA_I$ and a connected component $\cX$ of $\Int(A)$ equipped with a branch of $\log$, we have constructed a unitary $U_{I,A,\cX}$.
Our next step is to compare the maps $U_{I,A,\cX}$ as $A$ and $\cX$ vary, but in order to do so we need a reasonable way of comparing branches of $\log$ on $\cX \subset \Int(A)$ and $\cX^\prime \subset \Int(A^\prime)$ for some other $(1,A^\prime) \in \scA_I$.

%Recall that $\emptyset \ne \partial \cX \cap S^1 \subset I$ and that $\partial \cX \cap S^1$ contains an interval.

In order to systematically choose a branch of $\log$ on every connected component of every $A$ with $(1,A) \in \scA_I$, we will choose a continuous branch of $\log$ on the interval $I$ itself.
$$
\begin{array}{c}
\begin{tikzpicture}[baseline={([yshift=-.5ex]current bounding box.center)}]
	\coordinate (a) at (120:1cm);
	\coordinate (b) at (240:1cm);
	\coordinate (c) at (180:.25cm);
% BIG DISK
	\fill[fill=red!10!blue!20!gray!30!white] (0,0) circle (1cm);
% CURVED BOUNDARY REGION
	\fill[fill=white] (a)  .. controls ++(210:.61cm) and ++(90:.41cm) .. (c) .. controls ++(270:.41cm) and ++(150:.61cm) .. (b) -- ([shift=(240:1.01cm)]0,0) arc (240:480:1.01cm);
	\draw (115:1cm) arc (115:245:1cm);
% COORDINATE LABELS
%	\node at (a) {(a)};
%	\node at (b) {(b)};
%	\node at (c) {(c)};
\end{tikzpicture}\\
I \cup \Int(A)
\end{array}
$$
Since $\Int(A)$ is a degenerate annulus with $\partial \Int(A) \cap S^1 \subset I$, a continuous branch of $\log$ on $I$ extends canonically to a holomorphic branch on $\Int(A)$ thereby yields a unitary $U_{I,A,\cX}$ for every connected component $\cX$ of $\Int(A)$.
In the following we fix a branch of $\log$ on $I$ and show that the resulting unitaries $U_{I,A,\cX}$ do not depend on $A$ or $\cX$.

We build up to that result by considering a special case.

\begin{Lemma}\label{lem: two annulus sharing interval transport}
Let $V$ be a simple unitary VOA with bounded localized vertex operators, and let $M$ and $N$ be unitary $V$-modules. 
Assume that $\pi^M$ and $\pi^N$ exist, and that $\cY_M^+ \in I_{loc} \binom{M}{M V}$.
Let $(M \boxtimes_t N, \cY_t)$ be a transport module for $M$ and $N$.

Let $I$ be an interval equipped with a branch of $\log$, and let $(1,A_1),(1,A_2) \in \scA_I$.
Let $\cX_i \subseteq \Int(A_i)$ be connected components and suppose that $S^1 \cap \partial \cX_1 \cap \partial \cX_2$ contains an interval.
Then  $U_{I,A_1,\cX_1} = U_{I,A_2,\cX_2}$.
\end{Lemma}
\begin{proof}
Since $S^1 \cap \partial \cX_1 \cap \partial \cX_2$ contains an interval, we may choose a $w_0 \in S^1$ with the property that $w_0 \in \Int(A_j^*,A_i)$ for any choice of $i,j \in \{1,2\}$.
Let $\cU_{ij} \subset \Int(A_i) \cap \Int(A_j)$ be the open sets obtained by applying Lemma \ref{lem: transport near boundary}.
Since the boundary of each $\cU_{ij}$ contains an interval about $w_0$, the same is true for the boundary of $\cU := \bigcap \cU_{ij}$.
Hence $\cU \cap \cX_1 \cap \cX_2 \ne \emptyset$, and by taking $\cU$ smaller if necessary we assume $\cU \subset \cX_1 \cap \cX_2$.

Now for $i,j \in \{1,2\}$, $a_1,a_2 \in M$, $z_1, z_2 \in \cU$, if we set
$$
x_1 = \cY_M^+(a_1,z_1)A_i, \quad y_1 = \cY_t(a_1,z_1)A_i
$$
and
$$
x_2 = \cY_M^+(a_2,z_2)A_j, \quad y_2 = \cY_t(a_2,z_2)A_j,
$$
where $\log z_i$ is obtained by holomorphically extending the branch on $I$ to $\cU$,
then 
\begin{equation}\label{eqn: two annulus transport}
\pi_I(x_2^*x_1) = y_2^*y_1
\end{equation}
by Lemma \ref{lem: transport near boundary}.

We now construct a unitary $U: \cH_M \boxtimes_I \cH_N \to \cH_{M \boxtimes_t N}$ in the same way as in the proof of Lemma \ref{lem: fusion iso for component}.
%Let 
%$$
%\cS = \Span \big\{ \cY_M^+(a,z)A_i \, r : a \in M, \, z \in \cU, \, r \in \cA_V(I), \, i \in \{1,2\} \big\}.
%$$
Using the transport relation \eqref{eqn: two annulus transport}, we can argue just as in Lemma \ref{lem: fusion iso for component} that the map
$$
\cY_M^+(a,z)A_i \otimes \xi \mapsto \cY_t(a,z)A_i\xi
$$
induces a unitary $U:\cH_M \boxtimes_I \cH_N \to \cH_{M \boxtimes_t N}$ such that 
\begin{equation}\label{eqn: two annulus unitary}
U(\cY_M^+(a,z)A_i \boxtimes_I \xi) = \cY_t(a,z)A_i
\end{equation}
for $i \in \{1,2\}$.
Analytically continuing, we see that  \eqref{eqn: two annulus unitary} holds for all $z \in \cX_i$.
Thus $U = U_{I,A_1,\cX_1} = U_{I,A_2,\cX_2}$, as desired.
\end{proof}

Suppose $\tilde I \subset I$, and $\theta$ is a sufficiently small angle that $r_\theta(\tilde I) \subset I$.
For $(1,A) \in \scA_{\tilde I}$, set $A_\theta = e^{i \theta L_0} A e^{-i \theta L_0}$ and note that $(1,A_\theta) \in \scA_I$ by the definition of system of generalized annuli. 
Moreover $\Int(A_\theta) = r_{\theta}(\Int(A))$.
The next lemma identifies the maps $U_{I,A,\cX}$ for an annulus with the map for the rotated annulus.

\begin{Lemma}\label{lem: transport unitary doesnt invariant under rotation}
Let $V$ be a simple unitary VOA with bounded localized vertex operators, and let $M$ and $N$ be unitary $V$-modules. 
Assume that $\pi^M$ and $\pi^N$ exist, and that $\cY_M^+ \in I_{loc} \binom{M}{M V}$.
Let $(M \boxtimes_t N, \cY_t)$ be a transport module for $M$ and $N$.

Let $I$ be an interval equipped with a branch of $\log$, let $\tilde I \subset I$, and let $(1,A) \in \scA_{\tilde I}$.
Suppose $r_\theta(\tilde I) \subset I$.
Let $\cX$ be a connected component of $\Int(A)$ and let $\cX_\theta = r_\theta(\cX)$.
%Set $A_\theta = e^{i \theta L_0}A$ and $\cX_\theta = r_\theta(\cX)$.
Then $U_{I,A,\cX} = U_{I,A_\theta,\cX_\theta}$.
\end{Lemma}
\begin{proof}
Let $J$ be an interval compactly contained in $\partial \cX \cap S^1$.
Observe that for $\theta < \abs{J}$, $\cX$ and $\cX_\theta$ share a boundary interval and thus the result follows from Lemma \ref{lem: two annulus sharing interval transport}.
For the general case, we may iterate the above observation.
% note that $r_{\theta}(J) \subset \partial \cX_\theta \cap S^1$, and $\abs{r_{\theta} J} = \abs{J}$.
%Thus if we choose a sequence of angles $0=\theta_0 < \cdots < \theta_n=\theta$ with $\theta_j - \theta_{j-1} < \abs{J}$ and repeatedly apply the above, we obtain
%$$
%U_{I,A,\cX} =  U_{I,A_{\theta_1},\cX_{\theta_1}} = \cdots = U_{I,A_{\theta_n},\cX_{\theta_n}}.
%$$
\end{proof}

We are now ready to prove the main result of this section:

\begin{Theorem}\label{thm: transport unitary}
Let $V$ be a simple unitary VOA with bounded localized vertex operators, and let $M$ and $N$ be unitary $V$-modules. 
Assume that $\pi^M$ and $\pi^N$ exist, and that $\cY_M^+ \in I_{loc} \binom{M}{M V}$.
Let $(M \boxtimes_t N, \cY_t)$ be a transport module for $M$ and $N$.
Let $I$ be an interval equipped with a branch of $\log$.

Then there exists a unique unitary 
$$
U_I: \cH_M \boxtimes_I \cH_N \to \cH_{M \boxtimes_t N}
$$
such that for all $a \in M$, $(1,A) \in \scA_I$, $z \in \Int(A)$, and $\xi \in \cH_N$ we have
$$
U_I(\cY_M^+(a,z)A \boxtimes_I \xi) = \cY_t(a,z)A\xi
$$
where the value of $\log z$ is obtained by continuing the branch of $\log$ from $I$ to $\Int(A)$.
\end{Theorem}
\begin{proof}
We must show that for every $(1,A),(1,A^\prime) \in \scA_I$ and any connected components $\cX \subset \Int(A)$ and $\cX^\prime \subset \Int(A^\prime)$, we have $U_{I,A,\cX} = U_{I,A^\prime,\cX^\prime}$.

We first show this under the additional assumption that $A$ and $A^\prime$ can be localized in small subintervals of $I$.
Specifically, suppose that there are intervals $\tilde I, \tilde K \subset I$ with $\babs{\tilde I},\babs{\tilde K} < \abs{I}/2$ such that $(1,A) \in \scA_{\tilde I}$ and $(1,A^\prime) \in \scA_{\tilde K}$.
Then there exist angles $\theta,\theta^\prime$ such that $r_\theta(\tilde I), r_{\theta^\prime}(\tilde K) \subset I$ and $\partial \cX_\theta \cap \partial \cX^\prime_{\theta^\prime}$ contains an interval.
Thus we have
$$
U_{I,A,\cX} = U_{I, A_\theta,\cX_\theta} = U_{I,A^\prime_{\theta^\prime}, \cX^\prime_{\theta^\prime}} = U_{I,A^\prime,\cX^\prime},
$$
where the first equality is Lemma \ref{lem: transport unitary doesnt invariant under rotation}, the second is Lemma \ref{lem: two annulus sharing interval transport}, and the third is again Lemma \ref{lem: transport unitary doesnt invariant under rotation}.

Now suppose we have an arbitrary $(1,A) \in \scA_I$.
It suffices to find a $(1,A^\prime) \in \scA_{J}$ for $J \subset I$ and $\abs{J} < \abs{I}/ 2$ such that $U_{I,A,\cX} = U_{I,A^\prime,\cX^\prime}$ for some connected component $\cX^\prime$ of $\Int(A^\prime)$.
Indeed, let $J$ be an interval contained in $\partial \cX \cap S^1$ with $\abs{J} < \abs{I}/2$.
By the definition of system of generalized annuli there is some $(1,A^\prime) \in \scA_{J}$ with $\Int(A^\prime) \ne \emptyset$.
Let $\cX^\prime$ be any connected component of $\Int(A^\prime)$.
Then $\partial \cX^\prime$ contains a boundary interval $J^\prime$.
By construction $J^\prime \subset J \subset \partial \Int(\cX)$, and thus $U_{I,A,\cX} = U_{I,A^\prime,\cX^\prime}$ by Lemma \ref{lem: two annulus sharing interval transport}.
\end{proof}

We now consider the dependence of $U_I$ constructed in Theorem \ref{thm: transport unitary} upon the interval $I$.
Recall that if $\tilde I \subset I$, then the inclusion $\Hom_{\cA(\tilde I^\prime)}(\cH_V, \cH_M) \subset \Hom_{\cA(I^\prime)}(\cH_V, \cH_M)$ induces a canonical unitary
$$
U_{I \leftarrow \tilde I} : \cH_M \boxtimes_{\tilde I} \cH_N \to \cH_M \boxtimes_{I} \cH_N
$$
which is characterized by $U_{I \leftarrow \tilde I}(x \boxtimes_{\tilde I} \xi) = x \boxtimes_I \xi$.
The following compatibility of this unitary with $U_I$ and $U_{\tilde I}$ is immediate:

\begin{Proposition}\label{prop: compatible with canonical unitary}
Let $V$ be a simple unitary VOA with bounded localized vertex operators, and let $M$ and $N$ be unitary $V$-modules. 
Assume that $\pi^M$ and $\pi^N$ exist, and that $\cY_M^+ \in I_{loc} \binom{M}{M V}$.
Let $(M \boxtimes_t N, \cY_t)$ be a transport module for $M$ and $N$.

Let $I$ be an interval equipped with a branch of $\log$ and let $\tilde I \subset I$ be a subinterval (with the same branch of $\log$).
Then
$$
U_{\tilde I} = U_{I} U_{I \leftarrow \tilde I}.
$$
\end{Proposition}
\begin{proof}
Let $(1,A) \in \scA_I$.
Then for every $z \in \Int(A)$, $a \in M$, and $\xi \in \cH_N$ we have
$$
U_{I} U_{I \leftarrow \tilde I}(\cY_M^+(a,z)A \boxtimes_{\tilde I} \xi)
=
U_I (\cY_M^+(a,z)A \boxtimes_{I} \xi)
=
\cY_t(a,z)A\xi.
$$
Hence $U_{I} U_{I \leftarrow \tilde I} = U_{\tilde I}$.
\end{proof}

\subsection{Isomorphisms of representations}\label{sec: fusion rules and transport}

In Section \ref{sec: transport and fusion} we saw that given a transport module $(M \boxtimes_t N, \cY_t)$ and an interval $I$ with a branch of $\log$, there is a natural unitary $U_I: \cH_M \boxtimes_I \cH_N \to \cH_{M \boxtimes_t N}$.
We would like this to be an isomorphism of $\cA_V$-representations.
In order to formulate this precisely, we need an action of $\cA_V$ on $\cH_{M \boxtimes_t N}$; that is, we need to know that $\pi^{M \boxtimes_t N}$ exists.
Even at that point it is not immediate that $U_I$ is an intertwiner of representations, and the problem reduces to showing $\cY_t \in I_{loc}$.

We do not solve here the general problem of showing that $\pi^{M \boxtimes_t N}$ exists or that $\cY_t \in I_{loc}$ (although we conjecture that these both always hold).
Instead, we consider the maximal submodule $M \boxtimes_{loc} N$ of $M \boxtimes_t N$ for which the corresponding representation exists and for which the projection of $\cY_t$ is localized, and show that the isometry $U_I^*|_{\cH_{M \boxtimes_{loc} N}}$ is an embedding of representations.
Indeed, $M \boxtimes_{loc} N$ is the largest submodule of $M \boxtimes_t N$ for which we could have this intertwining property, so this is the best possible result which does not address the general problem of existence and locality of $\pi^{M \boxtimes_t N}$ and $\cY_t$.
We now carefully define the module $M \boxtimes_{loc} N$.

If $K$ is a unitary $V$-module, then it can be decomposed 
$$
K = \bigoplus_{i \in S} K_i \otimes X_i
$$
where the $K_i$ are pairwise non-isomorphic simple unitary $V$-modules and the $X_i$ are finite-dimensional Hilbert spaces.
Let $S^\prime \subset S$ be the subset consisting of $i$ such that $\pi^{K_i}$ exists.
Then by \cite[Prop. 5.5]{TenerGRACFT2}, the maximal submodule of $K$ such that $\pi^K$ exists is
$$
K_{ex} = \bigoplus_{i \in S^\prime} K_i \otimes X_i.
$$
Now suppose that $M$ and $N$ are unitary $V$-modules and that $\pi^N$ exists, and fix $\cY \in I \binom{K}{M N}$.
Let $p_i$ be the projection of $K$ onto the isotypical component $K_i \otimes X_i$, and let $\cY_i = p_i\cY$.
%Observe that $\pi^K(\cA_V)^\prime = \bigoplus_{i \in S} 1 \otimes \cB(X_i)$.
For $i \in S^\prime$, let 
$$
\cI_i = \{ x \in \cB(X_i) : (1 \otimes x) \cY_i \in I_{loc} \}.
$$
Each $\cI_i$ is a left-ideal and therefore of the form $\cI_i = \cB(X_i)p_{i,loc}$ for a projection $p_{i,loc}$.
We define $K_{loc}$, the local part of $K$ with respect to $\cY$, to be
$$
K_{loc} = \bigoplus_{i \in S^\prime} K_i \otimes p_{i,loc}X_i,
$$
and let $p_{loc} = \bigoplus 1 \otimes p_{i,loc}$.

\begin{Remark}\label{rmk: Yloc is probably local}
For any $(B,A) \in \scA_I$ and $z \in \Int(B,A)$ we have $(1 \otimes p_{i,loc})B\cY_i(a,z)A \in \Hom_{\cA(I^\prime)}(\cH_N, \cH_K)$.
In general we do not know that $p_{loc}B \cY(a,z)A$ is a bounded operator when $S^\prime$ is infinite, as $\norm{(1 \otimes p_{i,loc})B\cY_i(a,z)A}$ might not be bounded independent of $i$.
However if we know a priori that $p_{loc}B \cY(a,z)A$  is bounded then it lies in $\Hom_{\cA(I^\prime)}(\cH_N, \cH_K)$, as each component $p_{i,loc}B \cY(a,z)A$ is a local intertwiner.
In particular, this is the case if $V$ is regular, as there are only finitely many simple $V$-modules, and we have $p_{loc} \cY \in I_{loc} \binom{K_{loc}}{M N}$.
\end{Remark}

\begin{Definition}\label{def: local transport module}
Let $V$ be a simple unitary VOA with bounded localized vertex operators, let $M$ and $N$ be unitary $V$-modules, and suppose that a transport module $(M \boxtimes_t N, \cY_t)$ exists.
Then we denote by $M \boxtimes_{loc} N$ the maximal local submodule of $M \boxtimes_t N$ with respect to $\cY_t$, as described above.
We set $\cY_{loc} := p_{loc} \cY$.
\end{Definition}

Observe that if $(K, \cY_t)$ and $(\tilde K, \tilde \cY_t)$ are a pair of transport modules, and $u:K \to \tilde K$ is the canonical unitary equivalence of Proposition \ref{prop: transport module unique}, then $uK_{loc} = \tilde K_{loc}$ and $u\cY_{loc} = \tilde \cY_{loc}$.
Thus if a transport module for $M$ and $N$ exists then the pair $(M \boxtimes_{loc} N, \cY_{loc})$ is unique up to canonical unitary equivalence.

\begin{Remark}\label{rem: local for regular}
By Proposition \ref{prop: transport for regular VOAs}, if $V$ is regular and a transport module $M \boxtimes_t N$ exists, then 
$$
M \boxtimes_{loc} N \cong \bigoplus I_{loc} \binom{K_i}{M N} ^* \otimes K_i
$$
as a $V$-modules, where the direct sum runs over a complete family of simple modules $K_i$.
\end{Remark}

\begin{Theorem}\label{thm: voa and net fusion}
Let $V$ be a simple unitary VOA with bounded localized vertex operators, and let $M$ and $N$ be unitary $V$-modules.
Assume that $\pi^M$ and $\pi^N$ exist, and that $\cY_M^+ \in I_{loc} \binom{M}{M V}$.
Suppose that a transport module $M \boxtimes_t N$ exists, let $M \boxtimes_{loc} N$ be the maximal local submodule, and let $p_{loc}: \cH_{M \boxtimes_t N} \to \cH_{M \boxtimes_{loc} N}$ be the projection.
Let $I$ be an interval equipped with a branch of $\log$, and let $U_I:\cH_{M} \boxtimes_I \cH_N \to \cH_{M \boxtimes_t N}$ be the unitary of Theorem \ref{thm: transport unitary}.
Set $U_I^{loc} = p_{loc} U_I$.

Then $U^{loc}_I:\cH_{M} \boxtimes_I \cH_N \to \cH_{M \boxtimes_{loc} N}$ is an intertwiner of $\cA_V$-representations.
\end{Theorem}
\begin{proof}
If $(1,A) \in \scA_I$ and $z \in \Int(A)$, then we have
$$
U_I(\cY_M^+(a,z)A \boxtimes_I \xi) = \cY_t(a,z)A\xi.
$$
Since $\xi \mapsto U_I(\cY_M^+(a,z)A \boxtimes_I \xi)$ is manifestly bounded, we can conclude that $\cY_t(a,z)A$ is bounded.
Thus $p_{loc} \cY_t(a,z)A = \cY_{loc}(a,z)A \in \Hom_{\cA(I^\prime)}(\cH_N, \cH_{M \boxtimes_{loc} N})$ as in Remark \ref{rmk: Yloc is probably local}.
Observe that
$$
U_I^{loc}(\cY_M^+(a,z)A \boxtimes_I \xi) = p_{loc}\cY_t(a,z)A\xi = \cY_{loc}(a,z)A\xi.
$$

We first show that $U_I^{loc}$ intertwines the action of $\cA_V(I^\prime)$.
If $x \in \cA_V(I^\prime)$, then we have for all $(1,A) \in \scA_I$, all $z \in \Int(A)$, and all $a \in M$
\begin{align*}
U_I^{loc}&(\pi^M \boxtimes_I \pi^N)_{I^\prime}(x)(\cY_M^+(a,z)A \boxtimes_I \xi) =\\
&= U_I^{loc}(\cY_M^+(a,z)A \boxtimes_I \pi^N_{I^\prime}(x) \xi)\\
& = \cY_{loc}(a,z)A\pi^N_{I^\prime}(x)\xi\\
&=\pi^{M \boxtimes_{loc} N}_{I^\prime}(x)\cY_{loc}(a,z)A\xi\\
&=\pi^{M \boxtimes_{loc} N}_{I^\prime}(x)U_I^{loc}(\cY_M^+(a,z)A \boxtimes_I \xi).
\end{align*}
The first equality is the definition of $\pi^M \boxtimes_I \pi^N$, the second is the characterizing property of $U_I^{loc}$, the third uses the fact that $\cY_{loc} \in I_{loc}$, and the fourth is again the characterizing property of $U_I^{loc}$.
Vectors of the form $\cY_M^+(a,z)A \boxtimes_I \xi$ span a dense subspace of $\cH_M \boxtimes_I \cH_N$, so the above calculation establishes that $U_I^{loc}$ is an intertwiner of $\cA_V(I^\prime)$-modules.

Next we show that $U_I^{loc}$ is an intertwiner of $\cA_V(I)$-modules.
By the additivity property of a conformal net, it suffices to show that $U_I^{loc}$ commutes with the action of $\cA_V(\tilde I)$ for all intervals $\tilde I$ compactly contained in $I$.
Given such a $\tilde I$, we may choose an interval $J \subset I$ such that $J$ and $\tilde I$ are disjoint.
Choose $(1,A) \in \scA_{J}$ such that $\Int(A) \ne \emptyset$, and let $z \in \Int(A)$.
As before, vectors of the form $\cY_M^+(a,z)A \boxtimes_I \xi$ span a dense subspace of $\cH_M \boxtimes_I \cH_N$.
Now for $x \in \cA_V(\tilde I)$ we have
\begin{align*}
U_I^{loc}&(\pi^M \boxtimes_I \pi^N)_{\tilde I}(x)(\cY_M^+(a,z)A \boxtimes_I \xi) =\\
&=  U_I^{loc}(\pi^M_{\tilde I}(x) \cY_M^+(a,z)A \boxtimes_I  \xi)\\
&=  U_I^{loc}( \cY_M^+(a,z)A x \boxtimes_I  \xi)\\
&=  U_I^{loc}( \cY_M^+(a,z)A  \boxtimes_I  \pi_{\tilde I}^N(x) \xi)\\
&= \cY_{loc}(a,z)A\pi^N_{\tilde I}(x)\xi\\
&= \pi^{M \boxtimes_{loc} N}_{\tilde I}(x) \cY_{loc}(a,z)A\xi\\
&= \pi^{M \boxtimes_{loc} N}_{\tilde I}(x) U_I^{loc}(\cY_M^+(a,z)A \boxtimes_I \xi).
\end{align*}
The first equality is the the definition of the action of $\cA_V(I)$ (and thus of $\cA_V(\tilde I)$) on $\cH_M \boxtimes \cH_N$.
The second equality is the fact that $\cY_M^+ \in I_{loc}$ and $(1,A) \in \scA_J$, and thus $\cY_M^+(a,z)A \in \Hom_{\cA(\tilde I)}(\cH_V, \cH_M)$ since $\tilde I$ and $J$ are disjoint.
The third equality is the middle linearity of $\boxtimes_I$ and the fourth is the characterizing property of $U_I^{loc}$.
The fifth equality uses that $\cY_{loc} \in I_{loc}$ just as in the second equality, and the sixth equality is the characterizing property of $U_I^{loc}$.

We conclude that $U_I^{loc}$ intertwines the actions of $\cA_V(\tilde I)$, and since $\tilde I$ was arbitrary we in fact have that $U_I^{loc}$ intertwines the actions of $\cA_V(I)$.

So far we have shown that $U_I^{loc}$ commutes with the actions of $\cA_V(I)$ and $\cA_V(I^\prime)$.
By the additivity property of a conformal net, it only remains to verify that $U_I^{loc}$ intertwines the actions of small intervals $K$ containing endpoints of $I$.
Given such an interval $K$, choose an interval $\tilde I \subset I$ such that $\tilde I$ and $K$ are disjoint.
Recall that if $U_{I \leftarrow \tilde I}: \cH_M \boxtimes_I \cH_N \to \cH_M \boxtimes_{\tilde I} \cH_N$ is the canonical isomorphism of sectors, then by Proposition \ref{prop: compatible with canonical unitary} we have 
$$
U_{\tilde I} = U_I U_{I \leftarrow \tilde I}.
$$
Hence 
$$
U^{loc}_{\tilde I} = U^{loc}_I U_{I \leftarrow \tilde I}.
$$
Since $K \subset \tilde I^\prime$, and both $U_{I \leftarrow \tilde I}$ and $U^{loc}_{\tilde I}$ commute with the actions of $\cA_V(\tilde I^\prime)$, so does $U^I_{loc}$.
In particular, $U^I_{loc}$ commutes with the actions of $\cA_V(K)$.
Thus we conclude  that $U^I_{loc}$ is map of $\cA_V$-representations.
\end{proof}

The intertwiner $U^{loc}_I = p_{loc} U_I$ constructed in Theorem \ref{thm: voa and net fusion} is the adjoint of an isometry.
The embedding $(U_I^{loc})^*$ is unitary precisely when $p_{loc} = 1$, which is to say $M \boxtimes_{loc} N = M \boxtimes_t N$.
Thus we have:

\begin{Corollary}\label{cor: voa and net fusion rules}
Under the hypotheses of Theorem \ref{thm: voa and net fusion}, $\pi^{M \boxtimes_{loc} N}$ is isomorphic to a subsector of $\pi^M \boxtimes_I \pi^N$, and if $M \boxtimes_t N = M \boxtimes_{loc} N$ then $U_I$ gives a unitary equivalence of $\cA_V$-representations $\pi^{M \boxtimes_{t} N} \cong \pi^M \boxtimes_I \pi^N$.
\end{Corollary}

In the special case when $V$ is regular, we have an explicit formula for $M \boxtimes_{loc} N$, provided a transport module exists (see Remark \ref{rem: local for regular}).
We can restate Corollary \ref{cor: voa and net fusion rules} more explicitly in terms of fusion rules in this case:

\begin{Corollary}\label{cor: voa and net regular fusion inequality}
Let $V$, $M$ and $N$ be as in Theorem \ref{thm: voa and net fusion}, and suppose that $V$ is regular.
Then we have
$$
\dim \Hom_{\cA_V}(\cH_M \boxtimes_I \cH_N, \cH_K) \ge \dim I_{loc}\binom{K}{M N}
$$
for any simple $V$-module $K$ such that $\pi^K$ exists.

If $\pi^K$ exists and $I_{loc}\binom{K}{M N} = I \binom{K}{M N}$ for all simple submodules $K \subset M \boxtimes N$, then $U_I$ gives a unitary equivalence of $\cA_V$-representations $\pi^{M \boxtimes N} \cong \pi^M \boxtimes_I \pi^N$, where $M \boxtimes N$ is given the standard unitary structure.
\end{Corollary}

\subsection{Finite index subfactors and rationality of conformal nets}\label{sec: rationality}

If $\pi$ is a representation of a conformal net, then we obtain a subfactor 
$$
\pi_I(\cA(I)) \subseteq \pi_{I^\prime}(\cA(I^\prime))^\prime
$$
for every interval $I$, called a \emph{Jones-Wassermann} subfactor.
By diffeomorphism covariance (in fact, by M\"{o}bius covariance), the subfactors corresponding to the intervals $I$ and $J$ are unitarily conjugate, so we can talk about ``the'' subfactor corresponding to $\pi$.
The problem of showing that these subfactors have finite index for key examples of conformal nets is widely accepted to be a challenging, but crucial, open problem.
The first major achievement in this direction was the work of Wassermann \cite{Wa98} for WZW models of type A, with further progress made by Loke \cite{Loke}, Toledano Laredo \cite{TL97}, and Gui \cite{GuiG2}.

One purpose of this article is to adapt Wassermann's approach to provide a systematic approach to the relationship between VOA and conformal net fusion, which can be used to show the finiteness of the index of subfactors.
The approach taken here is more flexible than previous work, in two important ways.
First, we do not rely on any special analytic properties of certain VOAs which are not conjectured to hold for all unitary VOAs.
Second, earlier approaches required large amounts of information about intertwining operators to obtain any information about the fusion product of conformal net representations.
For example, one might need analytic information about every intertwining operator of type $\binom{K}{M N}$ to obtain information about the fusion $\pi^M \boxtimes \pi^N$.
In contrast, our methods provide extremely useful information even with limited information about intertwining operators.
We demonstrate that now by applying the results of Section \ref{sec: fusion rules for bimodules} to the problem of finiteness of the index of representations, and more generally rationality of conformal nets.

A conformal net is called rational if it has finitely many irreducible representations all of which have finite index.
By the work of Longo-Xu \cite{LongoXu04} and Morinelli-Tanimoto-Weiner \cite{MorinelliTanimotoWeiner18}, rationality is equivalent to ``complete rationality,'' which was introduced by Kawahigashi-Longo-M\"{u}ger \cite{KaLoMu01}, who showed that the representation category of a (completely) rational conformal net is modular.
We will use two results of Longo to show that conformal nets are completely rational.
The first is \cite[Thm. 24]{Longo03}, which says that if $\cB \subset \cA$ is a conformal subnet and $[\cA : \cB] < \infty$, then $\cA$ is rational if and only if $\cB$ is.
Here, $[\cA : \cB]$ is the index of the subfactor $\cB(I) \subseteq \cA(I)$ for any(/every) interval $I$.
The second result we will use is \cite[Thm. 4.1]{Longo90}, which says if $\pi$ and $\rho$ are irreducible and $\pi \boxtimes \rho$ contains the vacuum representation as a subsector, then $\pi$ and $\rho$ have finite index and are mutually contragredient.\footnote{%
Longo's article is written in the language of endomorphisms rather than bimodules.
A bimodule version of Longo's result is a consequence of \cite[Prop. 7.17]{BartelsDouglasHenriquesDualizability}.
}
We use here the fact that the category of representations is braided.

We thus obtain the following.
\begin{Theorem}\label{thm: finite index general}
Let $V$ be a simple unitary VOA with bounded localized vertex operators.
Let $M$ be a simple unitary $V$-module such that $\pi^M$ and $\pi^{\overline M}$ exist, and $\cY_M^+ \in I_{loc} \binom{M}{M V}$.
Suppose that there exists a transport module $\overline{M} \boxtimes_t M$ which contains a submodule isomorphic to $V$.
Then $\pi^M$ and $\pi^{\overline{M}}$ have finite index and are mutually contragredient.
\end{Theorem}
\begin{proof}
By Theorem \ref{thm: voa and net fusion} (in fact, Corollary \ref{cor: voa and net fusion rules}), $\pi^{\overline{M}} \boxtimes \pi^{M}$ contains a subsector isomorphic to $\pi^{\overline{M} \boxtimes_{loc} M}$.
By \cite[Prop. 5.5]{TenerGRACFT2}, $\pi^M$ and $\pi^{\overline{M}}$ are irreducible, so it suffices to show that $\overline{M} \boxtimes_{loc} M$ contains a submodule isomorphic to $V$.
Since $V \subseteq \overline{M} \boxtimes_t M$, it suffices to show that $I_{loc} \binom{V}{\overline{M} M} = I \binom{V}{\overline{M} M}$.
Since $M$ is simple, this is equivalent to showing that $\cY_{M}^- \in I_{loc}$.
By Proposition \ref{prop: dagger localized} this is equivalent to the condition that $\cY_{M}^+ \in I_{loc}$.
This last condition holds by assumption, completing the proof.
\end{proof}

When $V$ is regular, it is possible to state a simpler result.
\begin{Theorem}\label{thm: finite index for regular}
Let $V$ be a simple regular unitary VOA with bounded localized vertex operators, and let $M$ be a simple unitary $V$-module.
Suppose that $\pi^M$ and $\pi^{\overline{M}}$ exist and that $\cY_M^+ \in I_{loc} \binom{M}{M V}$.
Then $\pi^M$ and $\pi^{\overline{M}}$ have finite index and are mutually contragredient.
\end{Theorem}
\begin{proof}
By Corollary \ref{cor: transport for regular}, there is a unitary structure on $\overline{M} \boxtimes M$ making it into a transport module $\overline{M} \boxtimes_t M$.
Since $V$ is a submodule of $\overline{M} \boxtimes M$, we may now repeat the proof of Theorem \ref{thm: finite index general}.
By Corollary \ref{cor: voa and net fusion rules}, $\pi^{\overline{M}} \boxtimes \pi^{M}$ contains a subsector isomorphic to $\pi^{\overline{M} \boxtimes_{loc} M}$.
Since $V \subseteq \overline{M} \boxtimes M = \overline{M} \boxtimes_t M$, and $\cY_M^{\pm} \in I_{loc}$, we have 
\[
V \subset \overline{M} \boxtimes_{loc} M \subset \overline{M} \boxtimes_t M.
\]
Hence $\pi^{\overline{M}} \boxtimes \pi^{M}$ contains the vacuum representation $\pi^V$, and so by  \cite[Prop. 7.17]{BartelsDouglasHenriquesDualizability} (or \cite[Thm. 4.1]{Longo90} in terms of endomorphisms) we have that $\pi^M$ and $\pi^{\overline{M}}$ have finite index and are mutually contragredient.
\end{proof}

\begin{Corollary}\label{cor: finite index for regular in vacuum}
Let $V$ be a simple unitary VOA with bounded localized vertex operators, let $W$ be a unitary (not necessarily conformal) subalgebra of $V$ which is regular, and let $M$ be a $W$-submodule of $V$.
Then $\pi^M$ and $\pi^{\overline{M}}$ are finite index representations of $\cA_W$.
\end{Corollary}
\begin{proof}
By definition, the $W$-submodule $M$ has finite-dimensional eigenspaces for the operator $L_0^W$. 
Since the regular VOA $W$ has finitely many isomorphism classes of simple modules, $M$ must therefore be a \emph{finite} direct sum of simple modules, and we may assume without loss of generality that $M$ is simple.
Observe that $\pi^M$ exists by \cite[Thm. 5.6]{TenerGRACFT2}, and by \cite[Thm. 5.15]{TenerGRACFT2} there is a number $\Delta$ such that $x^\Delta p_M Y^V |_{M \otimes W} \in I_{loc} \binom{M}{M W}$.
If $a \in M$, $p_M Y^V(a,x)\Omega = a + \cdots$, and thus $p_M Y^V|_{M \otimes W} \ne 0$.
Since $M$ is simple we conclude that $I_{loc} \binom{M}{MW} = I \binom{M}{MW}$.

If $\theta$ is the PCT automorphism of $V$, then $\theta M$ is a $W$-submodule of $V$ which is isomorphic to $\overline{M}$, and thus we conclude that $\pi^{\overline{M}}$ exists.
% and $I_{loc} \binom{\overline{M}}{\overline{M}W} = I \binom{\overline{M}}{\overline{M}W}$.
Thus we may apply Theorem \ref{thm: finite index for regular} to obtain the desired result.
\end{proof}

\begin{Remark}
The argument for Corollary \ref{cor: finite index for regular in vacuum} works just as well if there is a $V$-module $N$ containing $M$ as a $W$-submodule such that $\pi^N$  and $\pi^{\overline{N}}$ exist and $\cY_N^+, \cY_{\overline{N}}^+ \in I_{loc}$.
\end{Remark}

There are several ways to apply the above results to obtain rationality of conformal nets.
Applying Theorem \ref{thm: finite index for regular} we immediately obtain the following.
\begin{Corollary}\label{cor: every module finite index}
Let $V$ be a simple regular unitary VOA with bounded localized vertex operators.
Suppose that every irreducible representation of $\cA_V$ is equivalent to one of the form $\pi^M$ for a unitary $V$-module $M$, and suppose that $\cY_M^+ \in I_{loc}$ for all such $M$.
Then $\cA_V$ is rational.
\end{Corollary}
Thus we can obtain rationality with only minimal information about intertwining operators.
In future work, we hope to show that every representation of $\cA_V$ is of the form $\pi^M$, which would streamline Corollary \ref{cor: every module finite index} further; see Section \ref{sec: outlook}.

Another way to apply Theorem \ref{thm: finite index for regular} is via conformal inclusions.
We obtain a method for establishing rationality without any discussion of intertwining operators.

\begin{Corollary}\label{cor: rationality of nets}
Let $V$ be a simple unitary VOA with bounded localized vertex operators, let $W$ be a unitary subalgebra which is regular, and assume that the coset $W^c$ is regular.
Then $[\cA_V : \cA_W \otimes \cA_{W^c}] < \infty$.

In particular, if $W$ is a conformal subalgebra then $[\cA_V : \cA_W] < \infty$, so that $\cA_V$ is rational if and only if $\cA_W$ is.
\end{Corollary}
\begin{proof}
Observe that $W \otimes W^c$ is a regular, unitary, conformal subalgebra of $V$, and recall that $\cA_{W \otimes W^c} = \cA_W \otimes \cA_{W^c}$ \cite[Prop. 4.11]{GRACFT1}.
By Corollary \ref{cor: finite index for regular in vacuum}, $\pi^V$ is a finite index representation of $\cA_W \otimes \cA_{W^c}$, which is to say that $[\cA_V : \cA_W \otimes \cA_{W^c}] < \infty$.
The case when $W$ is a conformal subalgebra, so that $W^c = 0$, follows immediately, and the statement about rationality is \cite[Thm. 24]{Longo03}.
\end{proof}

In Section \ref{sec: examples} we apply these results to many important examples of unitary VOAs.

%% file: Sec_5_Examples.tex
\newpage
\section{Applications and outlook}\label{sec: examples}

In this section we consider bounded localized vertex operators with respect to the system of generalized annuli constructed in \cite[\S3.4]{TenerGRACFT2}.

\subsection{Applications to rationality of conformal nets}\label{sec: applications to rationality}

The machinery developed in Section \ref{sec: fusion rules for bimodules} allows one to study representations of conformal nets using corresponding information about VOAs.
The goal of Section \ref{sec: applications to rationality} is to use this machinery to give a straightforward proof of the following theorem.

\begin{Theorem}\label{thm: rationality of WZW and W}
Let $V$ be a VOA which is either:
\begin{itemize}
\item a WZW model $V(\frg,k)$ for $\frg$ a finite-dimensional simple complex Lie algebra and $k$ a positive integer
\item a $W$-algebra $W(\frg,\ell)$ in the discrete series, for $\frg$ of type ADE
\end{itemize}
Then $V$ has bounded localized vertex operators and $\cA_V$ is rational.
\end{Theorem}
It was shown in \cite[Thm. 7.4]{TenerGRACFT2} that WZW models have bounded localized vertex operators and the resulting conformal nets are equivalent to the standard loop group conformal nets.
Thus Theorem \ref{thm: rationality of WZW and W} resolves the problem of rationality of loop group conformal nets, which first appeared in \cite{GabbianiFrohlich93}.
Theorem \ref{thm: rationality of WZW and W} had previously been established for $V(\frg,k)$ when $\frg$ was of type $ACG$ (all levels) and $D$ (odd levels) \cite{Wa98, Xu00b, TL97,gui21categoricalextension}.
It had also been established for the $W$-algebras $W(\fsl_2,\ell)$, the Virasoro minimal models \cite{KaLo04}.
Our proof covers these cases as well.

The proof of Theorem \ref{thm: rationality of WZW and W} will be done in two parts.
In Lemma \ref{lem: rationality ADE}, we will prove the theorem for WZW models of type ABCDE and $W$-algebras of type ADE using conformal inclusions (applying Corollary \ref{cor: rationality of nets}).
In Proposition \ref{prop: rationality for F and G}, we will prove the theorem for WZW models of type F and G by showing directly that every irreducible representation has finite index (via Corollary \ref{cor: finite index for regular in vacuum}).
The union of these two results gives Theorem \ref{thm: rationality of WZW and W}.
Many of the models considered could be addressed by either method.

First consider when $\frg$ is of type $ADE$.
For a positive integer $k$, we have the diagonal embedding $V(\frg, k+1) \subset V(\frg,k) \otimes V(\frg,1)$, where $V(\frg,k)$ is the WZW model of type $\frg$ at level $k$.
A theorem of Arakawa-Creutzig-Linshaw \cite{ArakawaCreutzigLinshaw19} says that the coset $V(\frg,k+1)^c$ of this inclusion is the $W$-algebra $W(\frg,\ell)$, where $\ell$ is given by $\ell + \check h = \frac{k+ \check h}{k + \check h +1}$.
Moreover, the theorem asserts that $V(\frg,k+1)^{cc} = V(\frg,k+1)$.
The $W$-algebras which arise in this context are precisely the discrete series (see the introduction of \cite{ArakawaCreutzigLinshaw19} for more detail), which are known to be regular by work of Arakawa \cite{Arakawa15AssociatedVarieties, Arakawa15PrincipalNilpotent}.
In the special case when $\frg = \fsl_2$, these $W$-algebras give the discrete series of Virasoro minimal models, recovering the well-known result of Goddard-Kent-Olive \cite{GKO}.

Using these inclusions, we may apply Corollary \ref{cor: rationality of nets} to obtain a short proof of rationality of the WZW and $W$-algebras of type ADE.
For the remaining classical types, we apply a similar argument using conformal inclusions arising from inclusions of groups.

\begin{Lemma}\label{lem: rationality ADE}
Let $V$ be a VOA which is either:
\begin{itemize}
\item $V = V(\frg,k)$ for positive integer $k$ and $\frg$ is of type ABCDE
\item $V = W(\frg,\ell)$ is in the discrete series for $\frg$ of type ADE
\end{itemize}
Then $\cA_V$ is rational.
\end{Lemma}
\begin{proof}
We first consider $\frg$ of type ADE, and begin with the case $V = V(\frg, 1)$.
The Frenkel-Kac-Segal construction asserts that $V$ is a lattice VOA.
Bischoff gave a version of the Frenkel-Kac-Segal construction for conformal nets, and verified that the loop group conformal nets corresponding to $\frg$ at level $1$ were rational \cite[Prop. 4.1.17]{BischoffThesis} (see also \cite{DongXu06}).
By \cite[\S7.1]{TenerGRACFT2}, $\cA_V$ is isomorphic to the loop group conformal net, and so we have rationality of $\cA_V$ in these cases.

We can now give an inductive argument to obtain rationality at higher levels.
By \cite{ArakawaCreutzigLinshaw19} we have a conformal inclusion 
$$
V(\frg,k+1) \otimes W(\frg,\ell) \subseteq V(\frg,k) \otimes V(\frg,1)
$$
for the appropriate value of $\ell$.
Since $V(\frg, k+1)$ is regular by \cite{DongLiMason97} (see also \cite{FrZh92}) and $W(\frg,\ell)$ is regular by \cite{Arakawa15AssociatedVarieties, Arakawa15PrincipalNilpotent}, and tensor products of regular VOAs are regular \cite{DongLiMason97}, we can invoke Corollary \ref{cor: rationality of nets} and conclude that
$$
[ \cA_{V(\frg,k)} \otimes \cA_{V(\frg,1)}: \cA_{V(\frg,k+1)} \otimes \cA_{W(\frg,\ell)}] < \infty.
$$
We have used here that $\cA_{V_1 \otimes V_2} = \cA_{V_1} \otimes \cA_{V_2}$ \cite[Prop. 4.11]{GRACFT1}.
Thus by \cite[Thm. 24]{Longo03}, we conclude that $\cA_{V(\frg,k)} \otimes \cA_{V(\frg,1)}$ is rational if and only if $\cA_{V(\frg,k+1)} \otimes \cA_{W(\frg,\ell)}$ is.
Since $\cA_1 \otimes \cA_2$ is rational if and only if both $\cA_i$ are (e.g. by $\mu$-index considerations; see \cite{KaLoMu01}), we can conclude that if $\cA_{V(\frg,k)}$ is rational then so are $\cA_{V(\frg,k+1)}$ and $\cA_{W(\frg,\ell)}$.
Repeating this procedure completes the proof in the case where $\frg$ is of type ADE.

When $\frg$ is classical, we can appeal to inclusions of groups.
First recall that $V(\fso(n),1)$ is the even part of the (real) free fermion VOA $\cF^{\otimes n}$.
Thus $\cA_{V(\fso(n),1)}$ is the even part of $\cA_{\cF}^{\otimes n}$, which has $\mu$-index $1$ (as we explicitly identified $\cA_{\cF}^{\otimes 2}$ with the complex free fermion net in \cite{GRACFT1}, which has $\mu$-index 1).
Thus $\cA_{V(\fso(n),1)}$ has $\mu$-index 4 by \cite{KaLoMu01}, and is in particular rational.
We now apply Corollary \ref{cor: rationality of nets} to the conformal inclusions
$$
V(\fso(n),m) \otimes V(\fso(m),n) \subset V(\fso(nm),1)
$$ 
and $V(C_n,m) \otimes V(C_m,n) \subset V(\fso(4nm),1)$ to obtain rationality of WZW models of type B and C (and D, again).
\end{proof}

We now turn our attention to $V = V(\frg,k)$ when $\frg$ is of type F or G.
We make use of the well known inclusion of groups $F_4 \times G_2 \subset E_8$ \cite{DynkinIndex}.
This gives rise to a conformal inclusion of VOAs $V(F_4,1) \otimes V(G_2,1) \subset V(E_8,1)$ (note that $F_4$ and $G_2$ must both occur at level 1, as otherwise the central charge of $V(F_4,1) \otimes V(G_2,1)$ would exceed that of $V(E_8,1)$).
Both $V(F_4,1)$ and $V(G_2,1)$ have only two simple modules.
On the other hand, $V(E_8,1)$ is holomorphic and therefore cannot be equal to $V(F_4,1) \otimes V(G_2,1)$, and it follows that every $V(\frg,1)$ module occurs inside $V(E_8,1)$ for $\frg= F_4$ and $G_2$.
We were unable to find a proof of the following higher level analog of this fact in the literature, and so we include one here. 

\begin{Lemma}\label{lem: all modules F and G}
Let $\frg$ be $G_2$ or $F_4$.
Then every simple $V(\frg,k)$-module arises as a submodule of $V(E_8,1)^{\otimes k}$.
\end{Lemma}
\begin{proof}
We will only do the case $\frg=F_4$, as the case of $G_2$ is very similar but slightly simpler (and also appears implicitly in \cite[\S5.2]{GuiG2}).
We assume that the reader has some familiarity with the representation theory of affine VOAs (see \cite{FrZh92} for background).
A set of simple roots for $F_4$ is $\{[0,1,-1,0], [0,0,1,-1], [0,0,0,1], \tfrac12 [1,-1,-1,-1]\}$.
The corresponding fundamental weights are $\{[1,1,0,0], [2,1,1,0], \tfrac12 [3,1,1,1], [1,0,0,0]\}$.
The highest root is $\theta = [1,1,0,0]$.
The simple $V(F_4,k)$-modules are given by $L_k(M_\lambda)$ for $\lambda$ a dominant weight with $\lambda \cdot \theta \le k$, where $M_\lambda$ is the corresponding irreducible $\frg$-module.
The weights satisfying this property at called admissible at level $k$.
We will make use of the fact that if $\lambda$ is admissible at level $n$ and $\mu$ is admissible at level $k-n$, then $L_k(M_\lambda \otimes M_\mu) \subseteq L_{n}(M_\lambda) \otimes L_{k-n}(M_\mu)$.

We now prove by induction that every irreducible $V(\frg,k)$-module arises as a submodule of $V(E_8,1)^{\otimes k}$.
We address $k=1,2,3$ by hand.
The case $k=1$ is described in the paragraph preceding the lemma.
If the dominant weight $\lambda$ has coordinates $(n_1,n_2,n_3,n_4)$ with respect to the fundamental weights above, then $\lambda \cdot \theta \le k$ if and only if $2n_1 + 3n_2 + 2n_3 + n_4 \le k$.
Thus the new admissible weights at level $k=2$ are $(1,0,0,0)$, $(0,0,1,0)$, and $(0,0,0,2)$.
We have
$$
L_2(\underline{(0,0,0,1)}^{\otimes 2}) \subset L_1(\underline{(0,0,0,1)})^{\otimes 2} \subset V(E_8,1)^{\otimes 2},
$$
where $\underline{(n_1,n_2,n_3,n_4)}$ is the irreducible $\frg$-module with coordinates $n_i$.
One can check directly that $\underline{(0,0,0,1)}^{\otimes 2}$ contains subrepresentations $\underline{(1,0,0,0)}$, $\underline{(0,0,1,0)}$, and $\underline{(0,0,0,2)}$, which completes the proof for $k=2$.

At level $k=3$, there are four new admissible dominant weights, $(0,0,0,3), (0,0,1,1),$ $ (0,1,0,0),$ and $(1,0,0,1)$.
We have 
$$
L_3\underline{(0,0,0,3)}) \subset L_1(\underline{(0,0,0,1)})^{\otimes 3} \subset V(E_8,1)^{\otimes 3}
$$
On the other hand, $\underline{(0,0,1,0)} \otimes \underline{(0,0,0,1)}$ contains the irreducible representations corresponding to the three other new admissible weights.
Thus from the inclusion
$$
L_3\big(\underline{(0,0,1,0))} \otimes \underline{(0,0,0,1)}\big) \subset \ L_2(\underline{(0,0,1,0)}) \otimes L_1(\underline{(0,0,0,1)}) \subset V(E_8,1)^{\otimes 3},
$$
we conclude that the lemma holds for $k=3$ as well.

Now suppose that $k \ge 4$ and that the lemma has been established for all smaller levels.
Let $\lambda = (n_1,n_2,n_3,n_4)$ be admissible at level $k$.
If all $n_i = 0$ then $L_k(\underline{\lambda})$ is the vacuum sector, and the result is clear.
If $n_1 \ne 0$, then $(n_1-1,n_2,n_3,n_4)$ is admissible at level $k-2$.
Thus from the cases of level $k-2$ and level 2, we have
$$
L_k(\underline{(n_1,n_2,n_3,n_4)}) \subset L_{k-2}(\underline{(n_1-1,n_2,n_3,n_4)}) \otimes L_2(\underline{(1,0,0,0)}) \subset V(E_8,1)^{\otimes k}.
$$
Alternatively, if some other $n_i \ne 0$ we can apply the same argument, using the inductive hypothesis for $k-3$, $k-2$, or $k-1$.
\end{proof}

We point out that Lemma \ref{lem: all modules F and G} gives a proof of \emph{local equivalence} for loop group conformal nets of type F and G (see \cite[\S7.2]{TenerGRACFT2}).
We can in fact deduce much more from this branching information.

\begin{Proposition}\label{prop: rationality for F and G}
Let $V = V(\frg, k)$ for $\frg$ of type $E$, $F$, or $G$.
Then $\cA_V$ has bounded localized vertex operators and every irreducible $\cA_V$-representation is of the form $\pi^M$ for a simple unitary $V$-module $M$.
Moreover $\pi^M$ exists and has finite index for all unitary $V$-modules $M$.
Hence $\cA_V$ is rational.
\end{Proposition}
\begin{proof}
Bounded localized vertex operators was established in \cite[Thm. 7.4]{TenerGRACFT2}.
We have that every $V(\frg,k)$-module occurs inside $V(E_8,1)^{\otimes k}$ for type $F$ and $G$ by Lemma \ref{lem: all modules F and G}.
The same fact for $\frg$ of type E is more standard (a short proof can be given using the regularity of cosets considered above, see \cite[Thm. 7.11]{TenerGRACFT2}).
It follows that $\pi^M$ exists for every unitary $V(\frg,k)$-module $M$ \cite[Thm. 5.6]{TenerGRACFT2}.

Henriques showed that every irreducible representation of the loop group nets can be obtained from a representation of the loop group \cite[Thm. 26]{HenriquesColimits}, which gives an upper bound on the number of irreducible representations.
Using this fact, one can use a counting argument \cite[Lem. 7.8]{TenerGRACFT2} to show  that if $\pi^M$ exists for every unitary $V$-module, then every irreducible representation of $\cA_V$ is of the form $\pi^M$.
Finally, since every such $M$ is a $V(\frg,k)$-submodule of $V(E_8,1)^{\otimes k}$, we have that $\pi^M$ has finite index by Corollary \ref{cor: finite index for regular in vacuum}.
\end{proof}

\subsection{Unitary and positivity for VOA tensor product modules}

A consequence of the work in Section \ref{sec: positivity from BLVO} is that positivity of transport matrices and unitarity of tensor categories for VOAs can be obtained from the existence of representations $\pi^M$.
Proposition \ref{prop: rationality for F and G} demonstrates the existence of such representations for $V=V(\frg,k)$ when $\frg$ is of type EFG.
The classical cases (along with $G_2$) were addressed by Gui \cite[Thm. 6.1]{GuiG2}, and putting these together we obtain a complete solution for WZW models.

\begin{Theorem}\label{thm: UMTC for WZW and W}
Let $V$ be a VOA which is either
\begin{itemize}
\item a WZW model $V(\frg,k)$ for $\frg$ a finite-dimensional simple complex Lie algebra and $k$ a positive integer
\item a $W$-algebra $W(\frg,\ell)$ in the discrete series, for $\frg$ of type A or E
\end{itemize}
Then for every pair of simple unitary $V$-modules $M$ and $N$, the corresponding transport matrix is positive and there is a unitary structure on $(M \boxtimes N, \cY_\boxtimes)$ making it into a transport module.
Hence Gui's construction \cite{GuiUnitarityI,GuiUnitarityII} makes the category of unitary $V$-modules into a unitary modular tensor category.
\end{Theorem}
\begin{proof}
By \cite[Thm. 7.8]{GuiUnitarityII}, it suffices to verify positivity of transport matrices, which is equivalent to existence of transport modules by Proposition \ref{prop: transport for regular VOAs}.
The WZW models of type ABCDG were addressed by Gui in \cite{GuiUnitarityII,GuiG2} (and implicitly by Wassermann \cite{Wa98} in type A).
For the remaining WZW models, we have the existence of $\pi^M$ for every simple $V$-module $M$ by Proposition \ref{prop: rationality for F and G}, and for the $W$-algebras we have the same conclusion by \cite[Prop. 7.12]{TenerGRACFT2}.
Thus we have the desired positivity by Corollary \ref{cor: transport for regular}.
\end{proof}

The only obstruction to extending the argument of Theorem \ref{thm: UMTC for WZW and W} to give a reproof of positivity for WZW models of type BD or a proof for $W$-algebras of type D is the need to control representations of $V(\fso(n),1)$.
We cannot show the existence of all such modules from embeddings into fermions without discussing the Ramond sector, which we have avoided so far for the sake of simplicity.
In practice, such an argument should be possible.
However we instead prefer to direct future research towards obtaining model-independent results, as described in Section \ref{sec: outlook}.

We are also interested in verifying the positivity conjecture (Conjecture \ref{conj: positivity conjecture}) for (highly) non-rational VOAs, with an eye towards studying their unitary fusion product theory via transport modules.
A typical application of our results here is the following.

\begin{Proposition}\label{prop: positivity for Vir}
Let $V$ be a simple unitary regular VOA of central charge $c \ge 1$, let $M$ and $N$ be unitary $V$-modules, and suppose that $\pi^N$ exists.
Suppose that $M$ has a Virasoro primary state of conformal dimension $h$ and $N$ has a Virasoro primary state of conformal dimension $h^\prime$.
Then there is some $\varepsilon > 0$ such that the transport form on the $\Vir_c$-modules $L(c,h)$  and $L(c,h^\prime)$ exists and is positive semi-definite for $1 > \abs{z} > 1 - \varepsilon$.
\end{Proposition}
\begin{proof}
By Proposition \ref{prop: rational intertwiner Hilb}, $\cY_M^- \in I_{Hilb}\binom{V}{\overline{M} M}$, and thus the restriction to a $L(c,0)$-intertwining operator lands in $I_{Hilb}\binom{L(c,0)}{\overline{L(c,h)} L(c,h)}$.
By \cite[Thm. 5.6]{TenerGRACFT2}, $\pi^{L(c,h^\prime)}$ exists.
Thus by Proposition \ref{prop: general positivity} we have the desired conclusion.
\end{proof}

We have numerous examples of regular VOAs such that $\pi^M$ exists for every $V$-module, and so we can apply Proposition \ref{prop: positivity for Vir} to all of their Virasoro primary states, as well as the primary states in the tensor products of modules, and so on.
This provides a wide class of $\Vir_c$ modules for which positivity has been verified, and we regard this as very strong support for Conjecture \ref{conj: positivity conjecture}.
We can apply the same argument to obtain positivity for many other models as well, and if those models have vertex tensor categories of representations, we may use positivity to construct a transport module (see Section \ref{sec: pos to transport}).
Thus we regard this result as strong suppose for Conjecture \ref{conj: existence of transport} as well.

\subsection{On fusion rule calculations}

By Corollary \ref{cor: voa and net fusion rules}, if $V$ is regular, $\pi^M$ and $\pi^N$ exist, and $\cY_M^+ \in I_{loc}$, then $\pi^M \boxtimes \pi^N$ contains a submodule isomorphic to the local submodule $M \boxtimes_{loc} N \subset M \boxtimes N$.
Moreover, if $M \boxtimes_{loc} N = M \boxtimes N$, then we have $\pi^M \boxtimes \pi^N \cong \pi^{M \boxtimes N}$.
It is not an easy problem to directly show that $M \boxtimes_{loc} N = M \boxtimes N$, although it should be true in general for regular unitary VOAs.
In practice, this can be done in many examples by embedding $V$ in a larger unitary VOA $W$, and inferring locality properties of $V$-intertwining operators from those of $W$-intertwining operators.

This strategy is quite similar to the one employed in \cite{Wa98,TL97,GuiG2} to compute fusion rules for WZW models of type $ABCDG$ and Virasoro minimal models.
In fact, the results in Section \ref{sec: fusion rules for bimodules} can be understood as a way of systematizing these kinds of calculations, and going directly from information about VOA branching rules to theorems about conformal nets without needing to prove any analytic results along the way.

For example, Wassermann uses the embedding of $V(\fsl_n,k)$ into free fermions to deduce analytic properties of $\fsl_n$-primary fields whose charge space is the vector representation (i.e. the defining representation $\bbC^n$ corresponding to the Young diagram $\square$).
The same inclusions show that $I\binom{N}{M_\square K} = I_{loc} \binom{N}{M_\square K}$ where $M_\square$ is the $V(\fsl_n,k)$ representation corresponding to the vector representation.
Thus we can show $\pi^{M_\square} \boxtimes \pi^N \cong \pi^{M_\square \boxtimes N}$.
Similarly, the inclusions used to study WZW models of type $BCDG$ in \cite{GuiG2} produce localized intertwining operators, and many of the same conformal net fusion rule computations can be obtained directly from our results.
Moreover, these results can be extended to WZW models of type $EF$ as well, as we do not encounter any of the difficulties described in \cite[\S6]{GuiG2} following our approach.
It is possible that when $\frg$ is of type $E$ or $F$, enough information about can be extracted from the inclusions $\frg \subset E_8$ and $\frg_{k+1} \subset \frg_k \times \frg_1$ to fully compute the fusion rules for $\cA_{V(\frg,k)}$.
We do not attempt such a calculation here.
Instead, we will explore in future work the general framework described in Section \ref{sec: outlook} for obtaining fusion rule calculations without a case-by-case analysis.

\subsection{Outlook}\label{sec: outlook}

The long-term goal is to understand the relationship between VOAs and conformal nets well enough to be able to move seamlessly between representations of $V$ and representations of $\cA_V$.
We have seen that the close relationship between transport modules and Connes fusion provides a potential bridge between the two notions, and in this article we have used this connection to establish new results for both conformal nets and VOAs.
In order to invoke our theorems, however, various analytic properties of the VOAs in question must be verified.
The assumptions which appear include $V$ having bounded localized vertex operators, $\pi^M$ existing, $\cY_M^+ \in I_{loc}$, and more generally locality of intertwining operators.
We will use the rest of this section to describe a program to establish these properties broadly.

Given a unitary VOA $V$, there are two approaches to constructing a conformal net: the `standard' approach which was pursued in \cite{CKLW18} and the `bounded localized vertex operator' approach used here \cite{GRACFT1}.
At present there does not appear to be a method which could prove that an arbitrary unitary VOA leads to a conformal net.
Instead, one begins with a collection of examples for which there is such a correspondence, and shows that the correspondence is preserved under suitable constructions (subtheories, tensor products, extensions, and so on).
Work to this effect has been completed in \cite{CKLW18,TenerGRACFT2,GuiQSystem}.
%We will discuss how one might obtain new examples of VOAs with bounded localized vertex operators at the end of this section.

The next step is to develop the correspondence between modules and representations.
The following conjecture appears to be in reach of current methods.
\begin{Conjecture}\label{conj: net reps come from modules}
Let $V$ be a simple unitary VOA with bounded localized vertex operators, and let $\pi$ be a representation of $\cA_V$.
Then there exists a (generalized) unitary $V$-module $M$ such that $\pi \cong \pi^M$ and $\cY_M^+ \in I_{loc}$.
\end{Conjecture}
%Indeed, one can attempt to define $Y^M$ on $\cH_\pi$ by $\pi^M(v,z)b = \pi(Y(v,z)A)A^{-1}b$, although it would take some work to verify that $Y^M$ defines a $V$-module.
%Note that the converse problem of integrating an arbitrary (generalized) $V$-module to a representation of $\cA_V$ is likely to be quite a bit more difficult, and would require some hypotheses on the module when $V$ is not e.g. regular.

If one can establish Conjecture \ref{conj: net reps come from modules}, then one knows a priori that given a pair of representations $\pi^M$ and $\pi^N$, there must be a module $K$ such that $\pi^M \boxtimes \pi^N = \pi^K$.
In fact, $K$ should be a transport module; that is precisely the motivation for the definition of transport module.

\begin{Conjecture}\label{conj: Connes fusion gives intertwining operators}
Let $V$ be a simple unitary VOA with bounded localized vertex operators, let $M$, $N$ and $K$ be unitary $V$-modules such that $\pi^M \boxtimes \pi^N = \pi^K$.
Then $\cY(a,z)b := \cY_M^+(a,z)A \boxtimes A^{-1}b$ defines an intertwining operator $\cY \in I \binom{K}{M N}$.
\end{Conjecture}
The power of this approach is that the definition of Connes fusion implies that $(K, \cY)$ is a transport module, and moreover that $\cY \in I_{loc}$.
Thus we could conclude that $\pi^M \boxtimes \pi^N \cong \pi^{M \boxtimes_t N}$.
In particular if $V$ is regular we would have $\pi^M \boxtimes \pi^N \cong \pi^{M \boxtimes N}$ for all modules $M$ and $N$ such that $\pi^M$ and $\pi^N$ exist and $\cY_M^+ \in I_{loc}$.
Thus a corollary of Conjecture \ref{conj: Connes fusion gives intertwining operators} would be that transport modules exist and $\boxtimes_t$ is associative, even for poorly behaved VOAs like non-rational $\Vir_c$; we see this as evidence that one could build a tensor category of appropriate $V$-modules around the notion of transport module.
On the other hand, for regular $V$ with bounded localized vertex operators, these conjectures would immediately imply that the conformal net $\cA_V$ is rational.

We see both Conjecture \ref{conj: net reps come from modules} and Conjecture \ref{conj: Connes fusion gives intertwining operators} as approachable with current methods.
Positive results of these conjectures would dramatically reduce the amount of case-by-case analysis necessary to identify the fusion product theory of a conformal net with that of a regular VOA, and provide new insight into unitary VOAs with wild representation theory.

%% file: gracft3.bbl
\def\lfhook#1{\setbox0=\hbox{#1}{\ooalign{\hidewidth
  \lower1.5ex\hbox{'}\hidewidth\crcr\unhbox0}}}